\documentclass[11pt]{article}
\usepackage[T1]{fontenc}                                                  
\usepackage[utf8]{inputenc}
\usepackage{amsmath,amsfonts,amssymb}
\usepackage{stmaryrd}
\usepackage{physics}
\usepackage{tikz}
\usepackage{tikz-cd}
\usepackage{quiver}
\usepackage{amsthm}
\usepackage{mathrsfs}
\usepackage{enumitem}
\usepackage[numbers,sort&compress]{natbib}
\usepackage{doi}
\renewcommand{\doitext}{\noexpand\textsc{doi:}}
\newcommand{\arxiv}[2]{{\textsc{arXiv}}:\href{https://arxiv.org/abs/#1}{#1} {[#2]}}

\theoremstyle{definition}
\newtheorem{defi}{Definition}[section] 
\newtheorem{example}[defi]{Example} 

\theoremstyle{plain}
\newtheorem{thm}[defi]{Theorem} 
\newtheorem{prop}[defi]{Proposition} 
\newtheorem{cor}[defi]{Corollary} 

\theoremstyle{remark}
\newtheorem*{remark}{Remark}

\numberwithin{equation}{section}

\newcommand{\R}{\mathbb{R}}
\newcommand{\T}{\mathrm{T}}

\newcommand{\cH}{\mathcal{H}} 
\newcommand{\set}[1]{\left\{ #1 \right\}} 
\newcommand{\cL}{\mathcal{L}}
\newcommand{\Adm}{\mathrm{Adm}}
\newcommand{\g}{\mathfrak{g}}
\newcommand{\J}{\mathbf{J}}
\newcommand{\cD}{\mathcal{D}}

\newcommand{\cR}{\mathcal{R}}
\newcommand{\h}{\mathfrak{h}}
\newcommand{\p}{\mathfrak{p}}
\newcommand{\pr}{\mathrm{pr}}

\usepackage{hyperref}
\hypersetup{
   colorlinks   =  true,
    linkcolor    = blue,
    citecolor    = red,
     urlcolor	=blue,     
	urlbordercolor = {1 1 0}
}

\begin{document}

\vskip0.5cm

\begin{center}

    \noindent
    {\Large \bf
        Mixed superposition rules for Lie systems, \\[5pt]
        compatible geometric structures,  and applications}

\end{center}

\medskip

\begin{center}

    {\sc  Rutwig Campoamor-Stursberg$^{1,2}$, Oscar Carballal$^{1,2,3}$  \\[2pt]
        Francisco J. Herranz$^{4}$ and Javier de Lucas$^{3}$\begingroup
        \renewcommand{\thefootnote}{\fnsymbol{footnote}}%
        \addtocounter{footnote}{-1}%
        \endgroup}

\end{center}

\medskip

\noindent
$^{1}$ Instituto de Matem\'atica Interdisciplinar, Universidad Complutense de Madrid, E-28040 Madrid,  Spain

\noindent
$^{2}$ Departamento de \'Algebra, Geometr\'{\i}a y Topolog\'{\i}a,  Facultad de Ciencias
Matem\'aticas, Universidad Complutense de Madrid, Plaza de Ciencias 3, E-28040 Madrid, Spain

\noindent
{$^3$ Department of Mathematical Methods in Physics, University of Warsaw,
    ul. Pasteura 5, 02-093 Warszawa, Poland. ORCID: 0000-0001-8643-144X (corresponding author).
}

\noindent
$^{4}$ Departamento de F\'isica, Universidad de Burgos,
E-09001 Burgos, Spain

\medskip

\noindent  E-mail: {\small \href{mailto:rutwig@ucm.es}{\texttt{rutwig@ucm.es}}, \href{mailto:oscarbal@ucm.es}{\texttt{oscarbal@ucm.es}},   \href{mailto:fjherranz@ubu.es}{\texttt{fjherranz@ubu.es}}, \href{mailto:javier.de.lucas@fuw.edu.pl}{\texttt{javier.de.lucas@fuw.edu.pl}}
}

\begin{abstract}
    \noindent
    Mixed superposition rules are, in short, a method to describe the general solutions of a time-dependent system of first-order differential equations, a so-called Lie system, in terms of particular solutions of other ones. This article is concerned with the theory of mixed superposition rules and their connections with geometric structures. We provide methods to obtain mixed superposition rules for systems admitting an imprimitive finite-dimensional Lie algebra of vector fields or given by a semidirect sum. In particular, we develop a novel   mixed coalgebra method for Lie systems that are Hamiltonian relative to a Dirac structure, which is quite general, although we restrict to symplectic and contact manifolds in applications.  This provides us with practical methods to derive  mixed superposition rules and extends the coalgebra method to a new field of application while solving minor technical issues of the known formalism. Throughout the paper, we apply our results to  physical systems including Schr\"odinger Lie systems, Riccati systems, time-dependent Calogero--Moser systems with external forces, time-dependent harmonic oscillators, and time-dependent thermodynamical systems, where general solutions can be obtained from reduced system solutions. Our results are finally extended to Lie systems of partial differential equations and a new source of such PDE Lie systems, related to the { formal construction of solutions close to a given particular solution of PDEs}, is provided. An example based on the Tzitz\'eica equation and a related system { illustrating the potential applicability of these methods} is given.
\end{abstract}



\smallskip
\noindent
\textbf{PACS 2010 codes}:
Primary: 
02.20.Sv, 
02.40.Yy, 
02.30.Hq 
Secondary:
02.40.Ky, 
03.65.Fd

\medskip

\newpage

\tableofcontents

\newpage


\section{Introduction}
The theory of Lie systems, originally developed through the seminal contributions of Lie, Vessiot, and Guldberg, among others \cite{Vessiot1894,Lucas2020,Carinena2011}, provides a  framework to study systems of time-dependent ordinary differential equations (ODEs) whose general solutions can be expressed in terms of a finite set of particular solutions and some parameters related to initial conditions. Lie systems are characterised by the Lie--Scheffers theorem, which states that a system admits a superposition rule if and only if it is associated with a finite-dimensional Lie algebra of vector fields, now known as a {\it Vessiot--Guldberg Lie algebra} \cite{Lucas2020,Carinena2011}. Over the past decades, this theoretical foundation has become central in understanding integrability properties, explicit solvability, and symmetries of nonlinear differential systems \cite{Carinena2000,Ramos2002,Lazaro-Cami2009,Lucas2020,Amirzadeh-Fard2022}.

Lie systems are characterised by admitting  superposition rules, enabling the derivation of their general solutions through a geometric and algebraic mechanism. This approach has been pivotal in fields ranging from control theory and classical mechanics to mathematical physics and geometry (see \cite{Lucas2020,Carinena2011,Carinena2000,Ramos2002} and references therein). Recent developments have further expanded this classical concept to encompass mixed superposition rules \cite{Grabowski2013}, Lie--Hamilton systems (where the vector fields are Hamiltonian relative to a Poisson or symplectic structure) \cite{Carinena2013}, and Dirac--Lie systems (which generalise the latter to Dirac structures), providing deeper connections with modern geometric frameworks~\cite{Carinena2014}. For instance, in the works of   Cari\~nena, de Lucas  and their collaborators, Lie--Hamilton systems have been thoroughly studied, highlighting their relevance in both finite- and infinite-dimensional settings and their role in the geometric analysis of integrable systems \cite{Carinena2013,Ballesteros2013,Ballesteros2015,Ballesteros2021,deLucas2024ASetting}.

Despite these advances,   classical superposition rules restrict the construction of general solutions to depend only on solutions of the same system. In various applied and geometric contexts, it is natural and sometimes necessary to consider combinations involving solutions of auxiliary systems, leading to what are termed {\it mixed superposition rules} \cite{Grabowski2013}. Early glimpses of such ideas appeared implicitly in studies of inhomogeneous linear systems, where particular solutions are combined with solutions of homogeneous counterparts, and Ermakov systems \cite{Ermakov1880,Leach1991,Maamache1995,Carinena2008} (or Milne--Pinney equations~\cite{Milne1930,Pinney1950}), whose solutions can be described via particular solutions of harmonic oscillators~\cite{LA_08}.

In this work, we propose a comprehensive framework for mixed superposition rules for Lie systems, incorporating both algebraic and geometric perspectives, particularly via the formalism of Poisson coalgebras and comodules. Coalgebra techniques, initially introduced in the context of integrable systems and factorization problems, have demonstrated their efficacy in generating conserved quantities and encoding symmetries in an algebraic way. Notably, coalgebras and their dual structures have become essential in describing the algebraic backbone of classical and quantum   (super)integrable systems (see, e.g., \cite{BR1998,Ballesteros2009,annals2009,Turbiner1988,Tempesta2001}). Here, we extend these ideas to Lie systems, providing a natural language for describing mixed superposition structures.

A central novel contribution of this paper lies in elucidating the role of {\it imprimitive Vessiot--Guldberg Lie algebras}, which preserve certain regular, integrable distributions on manifolds \cite{GonzalezLopez1992,Shnider1984,Shnider1984a}. Such algebras facilitate reductions via foliations, allowing us to project complex systems onto lower-dimensional leaf spaces. This geometric mechanism is fundamental for constructing mixed superposition rules: by solving simpler projected systems (for example, Riccati equations or reduced oscillators), we can reconstruct the solutions of the original system through an explicit procedure. Concrete examples, such as the Riccati system and  time-dependent harmonic oscillators, showcase this approach and provide new geometric insights.

Moreover, we investigate Lie systems whose Vessiot--Guldberg Lie algebras admit {\it semidirect sum decompositions}, a structure well-known in the theory of affine and inhomogeneous systems \cite{Carinena2011,Grabowski2013,Carinena2003}. By exploiting the algebraic properties of such decompositions, one can derive mixed superposition rules in a purely algebraic fashion without recourse to geometric reduction. This unifies and generalises the superposition mechanisms of linear and affine systems, extending classical results on inhomogeneous linear ODEs.

Our approach also highlights connections to momentum maps and Poisson geometry, framing mixed superposition rules within a broader algebraic--geometric context. In particular, the constants of motion necessary for constructing superposition rules are obtained via Poisson algebra morphisms and Casimir functions of Poisson coalgebra structures. These constructions resonate with the methods used in integrable Hamiltonian systems and symmetry reduction \cite{Ortega2004}, thus opening avenues for future applications in geometric mechanics and mathematical physics.

Furthermore, we recall an extended Lie--Scheffers theorem for mixed superposition rules, establishing necessary and sufficient conditions for their existence in terms of a distributional injectivity criteria (see Theorem \ref{th:mixed_injective}). This result represents a significant generalisation of classical Lie theory and provides a rigorous foundation for analysing the solvability of more general differential systems.

Some of our fundamental results are described in Section \ref{Section:MixedNatural}. In particular, Theorem \ref{thm:imprimitive} describes practical conditions to ensure that a Lie system admits a mixed superposition rule in terms of solutions of the system and one of its projections onto a manifold of smaller dimension, which is eventually easier to study. Moreover, { Proposition~\ref{prop:mixed}} analyses Lie systems admitting a Vessiot--Guldberg Lie algebra that is a semidirect sum of Lie algebras. This is used to obtain the general solution of the Lie system in terms of  solutions of the same Lie system and Lie systems associated with an ideal of the Vessiot--Guldberg Lie algebra. These results are very useful and indeed describe some of the basic mixed superposition rules known. For instance, the fact that the general solution of an inhomogeneous linear system  is a function of a particular solution of the inhomogeneous system and a family of solutions of its homogeneous counterpart is a particular case of this instance. Nevertheless, since our result is based on the Lie algebra structure of the Vessiot--Guldberg Lie algebra, being an inhomogeneous linear system does not play a role, and this result can be much generalised to a generally geometric realm.

Section~\ref{Section:Dirac} briefly surveys Dirac geometry and analyses several results on Lie systems with a compatible Vessiot--Guldberg Lie algebra of Hamiltonian vector fields relative to a Dirac structure. In particular, results in \cite{Grabowski2013} are generalised to cover the direct product of Dirac--Lie systems. For instance, the direct product of Dirac--Lie systems is proven to be a Dirac--Lie system. These results are necessary so as to provide a coalgebra method for mixed superposition rules.

A coalgebra formalism for mixed superposition rules is presented in Section \ref{Section:coalgebras}. It is worth noting that this section fixes a technical problem presented in previous coalgebra methods. In particular, the usual approach in the coalgebra method relied on identifying $C^\infty(M)\otimes C^\infty(M)$ with $ C^\infty(M \times M)$. Nevertheless, this is not true for the standard tensor product, which only allows for finite sums of products of functions on the first and second copies of $M$. Meanwhile, we recall how the topology of nuclear and Fr\'echet spaces can be used to define a topological tensor product $\widehat{\otimes}$ yielding the isomorphism $C^{\infty}(M) \, \widehat{\otimes} \, C^{\infty}(M) \simeq C^{\infty}(M \times M)$ (see \cite{Grothendieck1952, Kriegl1997,Treves1967}). This shows that the coalgebra method is fully well-defined, thereby solving a long-standing technical issue of the method. In particular, it addresses one of the main reasons why certain parts of the theory of Lie systems with associated Poisson coalgebras were confined to  polynomial Casimir functions, although these remain the most common in applications.

All results of this paper are exemplified with relevant physical examples, which illustrate their versatility. These include Schr\"odinger Lie systems,  Riccati systems, time-dependent Calogero--Moser systems with external forces, harmonic oscillators with a time-dependent frequency and time-dependent thermodynamical systems.

The present work is structured as follows.  In Section  \ref{section:fundamentals}, we review fundamental notions concerning Lie systems, Poisson coalgebras, and comodules, setting the algebraic and geometric stage of the paper. The notion of mixed superposition rules is introduced  in Section   \ref{Section:mixedsuperpositions} by studying some Schr\"odinger Lie systems,  allowing us to present  their formal definition in illustrative way, as well as to propose  a generalisation of the classical Lie--Scheffers theorem, developing  the necessary tools via the coalgebra techniques and momentum maps. Section \ref{Section:MixedNatural} focuses on Lie systems with imprimitive Vessiot--Guldberg algebras, detailing the geometric reduction mechanisms and explicit reconstruction formulas. We discuss concrete examples that illustrate the theory, such as the Riccati system and time-dependent harmonic oscillators. Section \ref{Section:Dirac} is devoted to Dirac--Lie systems, emphasising their compatibility with mixed superposition rules and the geometric structures involved. In Section \ref{Section:coalgebras}, we develop the coalgebra formalism for Dirac--Lie systems, providing a unified algebraic viewpoint on the constants of motion and their role in superposition constructions. Section \ref{Section:Applications} presents applications to various physical and geometric models, including time-dependent thermodynamical systems and partial differential equation (PDE) Lie systems. Finally, Section \ref{Section:Concluding} summarises our main findings, discusses potential extensions, and outlines open problems and future research directions.


\section{Fundamentals}
\label{section:fundamentals}

Let us introduce the key concepts and definitions used throughout the paper. It is
assumed that all structures are smooth. Manifolds are real, Hausdorff, paracompact, and separable, and are typically denoted by $M$ and $N$. All Lie algebras under consideration are real ones. Given two finite-dimensional Lie algebras $\h, \p$, and a Lie algebra morphism $R: \h \to \mathrm{Der}(\p)$ of $\h$ in the space of derivations on $\p$ relative to its Lie bracket, a so-called {\it representation}, the semidirect sum of $\mathfrak{h}$ with $\mathfrak{p}$ is denoted by $\h \,\overrightarrow{\oplus}_{R}\, \p$. When  $R$ is understood from the context,  $\overrightarrow{\oplus}_{R}$ is denoted simply by  $\overrightarrow{\oplus}$. { As it happens in almost the whole literature on Lie systems
        \cite{Carinena2011,Lucas2020,Carinena2000,Carinena2007,BlazquezSanz2010}, we focus on local aspects of differential equations. Mathematically, one generally assumes that all mathematical structures are defined globally, which simplifies many technical details while being enough to obtain all practical applications.}


\subsection{Poisson (co)algebras and comodules} \label{subsection:coalgebras}

Let $(V, [\cdot, \cdot])$ denote a real Lie algebra with   Lie bracket $[\cdot, \cdot] : V \times V \to V$. If the Lie bracket is implicitly assumed,  the Lie algebra { is}  simply denoted by $V$. Given a subset $\mathcal{A} \subset V$, we write $\mathrm{Lie}(\mathcal{A}, [\cdot, \cdot])$  for the smallest Lie subalgebra of $V$ (in the sense of the inclusion) containing $\mathcal{A}$.

A \textit{Poisson algebra} is a pair $ (\mathcal{A}, \set{\cdot, \cdot}_{\mathcal{A}}  )$, where $\mathcal{A}$ is an $\R$-algebra, and  $\{\cdot, \cdot\}_{\mathcal{A}}$ is a Lie bracket on $\mathcal{A}$ satisfying the \textit{Leibniz rule}, namely,
$$
    \{bc, a\}_{\mathcal{A}} = b\{c, a\}_{\mathcal{A}} + \{b, a\}_{\mathcal{A}}c, \qquad a,b,c \in \mathcal{A}.
$$
The bracket $\{\cdot,\cdot\}_{\mathcal{A}}$ is called a {\it Poisson bracket}.
A \textit{Casimir element} of $\mathcal{A}$ is an element $C \in \mathcal{A}$ such that $\{C, a\}_{\mathcal{A}} = 0$ for all $a \in \mathcal{A}$. Given Poisson algebras $\mathcal{A}$ and $\mathcal{B}$, their tensor product $\mathcal{A} \otimes \mathcal{B}$ is a Poisson algebra with the Poisson bracket
$$
    \{a_1 \otimes b_1, a_2 \otimes b_2\}_{\mathcal{A} \otimes \mathcal{B}} := \{a_1, a_2\}_{\mathcal{A}} \otimes b_1 b_2 + a_1 a_2 \otimes \{b_1, b_2\}_{\mathcal{B}}, \qquad a_{1}, a_{2} \in \mathcal{A},\qquad b_{1}, b_{2} \in \mathcal{B}.
$$
A \textit{morphism of Poisson algebras} $\varphi: \mathcal{A} \to \mathcal{B}$ is a morphism of $\mathbb{R}$-algebras such that $
    \varphi(\{a, b\}_{\mathcal{A}}) = \{\varphi(a), \varphi(b)\}_{\mathcal{B}}$ for all $a,b \in \mathcal{A}.$
A \textit{Poisson coalgebra} is a triple $(\mathcal{A}, \set{\cdot, \cdot}_{\mathcal{A}}, \Delta_{\mathcal{A}})$, where $(\mathcal{A}, \{\cdot, \cdot\}_{\mathcal{A}})$ is a Poisson algebra and  $\Delta_{\mathcal{A}}: \mathcal{A} \to \mathcal{A} \otimes \mathcal{A}$ is a morphism of Poisson algebras, the so-called \textit{coproduct}, which is {\it coassociative}, that is,
\begin{equation*}
    \begin{tikzcd}
        {\mathcal{A}} & {\mathcal{A} \otimes \mathcal{A}} \\
        {\mathcal{A} \otimes \mathcal{A}} & {\mathcal{A} \otimes \mathcal{A} \otimes \mathcal{A}}
        \arrow["{\Delta_{\mathcal{A}}}", from=1-1, to=1-2]
        \arrow["{\Delta_{\mathcal{A}} }"', from=1-1, to=2-1]
        \arrow["\Delta_{\mathcal{A}} \otimes \mathrm{id}", from=1-2, to=2-2]
        \arrow["\mathrm{id} \otimes \Delta_{\mathcal{A}}"', from=2-1, to=2-2]
    \end{tikzcd}
\end{equation*}
is a commutative diagram of morphism of Poisson algebras \cite{Chari1995}. A {\it (right) coaction} of a Poisson coalgebra $(\mathcal{B}, \set{\cdot, \cdot}_{\mathcal{B}}, \Delta_{\mathcal{B}})$ on a Poisson algebra $(\mathcal{A}, \set{\cdot, \cdot}_{\mathcal{A}})$ is a morphism of Poisson algebras $\varphi: \mathcal{A} \to \mathcal{A} \otimes \mathcal{B}$ such that
\begin{equation*}
    \begin{tikzcd}
        {\mathcal{A}} & {\mathcal{A} \otimes \mathcal{B}} \\
        {\mathcal{A} \otimes \mathcal{B}} & {\mathcal{A} \otimes \mathcal{B} \otimes \mathcal{B}}
        \arrow["\varphi", from=1-1, to=1-2]
        \arrow["\varphi"', from=1-1, to=2-1]
        \arrow["{\varphi \otimes \mathrm{id}}", from=1-2, to=2-2]
        \arrow["{\mathrm{id} \otimes \Delta_{\mathcal{B}}}"', from=2-1, to=2-2]
    \end{tikzcd}
\end{equation*}
is a commutative diagram of Poisson algebra morphisms. In such case, $\mathcal{A}$ is said to be a {\it Poisson $\mathcal{B}$-comodule algebra} \cite{Ballesteros2002,Ballesteros2009,Agore2021}.


\subsection{Lie systems}
A {\it $t$-dependent vector field on $M$} is a map $X : \R \times M \ni (t, x) \mapsto X(t, x) \in \mathrm{T}M$ such that, for each $t \in \R$, the map $X_{t}:= X(t, \cdot): M \to \T M$ is a vector field \cite{Carinena2011,Lucas2020}.  Thus, $X$ amounts to a one-parameter family of vector fields $\set{X_{t}}_{t \in \R}$, and vice versa.  An {\it integral curve} of $X$ is map $\R \ni t  \mapsto x(t)\in M$ so that $\gamma : \R \ni t  \mapsto (t, x(t))\in \R \times M$ is an integral curve of the {\it autonomisation} of $X$, namely, the vector field $\partial / \partial t + X$ on $\R \times M$. Every $t$-dependent vector field $X$ on $M$ gives rise to its referred to as {\it associated system},
\begin{equation}
    \dv{x}{t} = X(t, x), \qquad (t, x) \in \R \times M.
    \label{eq:ODEs}
\end{equation}
Every solution $x(t)$ of this system is an integral curve of $X$, and any system of ODEs on $M$ in normal form, that is \eqref{eq:ODEs}, uniquely defines a $t$-dependent vector field $X$ on $M$. Hence, $X$ is unambiguously identified with its associated system, and we use $X$ to refer to both. The {\it smallest Lie algebra} of a $t$-dependent vector field $X$ is the smallest Lie algebra in the sense of inclusion, $V^{X}:= \mathrm{Lie}(\set{X_{t}}_{t \in \R})$, containing the vector fields $\set{X_{t}}_{t \in \R}$.

A {\it Lie system} is a $t$-dependent vector field $X$ on $M$ taking values in a finite-dimensional Lie algebra $V$ of vector fields, that is, $\set{X_{t}}_{t \in \R} \subset V$. In such a case, $V$ is called a {\it Vessiot--Guldberg \it Lie algebra} (VG Lie algebra) of $X$. Note that $X$ is a Lie system if and only if $V^{X}$ is finite-dimensional \cite{Carinena2011,Lucas2020}.
The main property of Lie systems is the so-called superposition rule \cite{Carinena2011,Lucas2020,Carinena2007,Winternitz1983}. A {\it superposition rule} for a system $X$ depending on $k$ particular solutions is a map { $\Phi: U\subset M^{k} \times M \to M$} such that the general solution $x(t)$ of $X$ can be expressed as $
    x(t) = \Phi(x_{(1)}(t), \ldots, x_{(k)}(t); p),
$
where $x_{(1)}(t), \ldots, x_{(k)}(t)$ is a family of particular solutions of $X$ such that $(x_{(1)}(t), \ldots, x_{(k)}(t),p)$ takes values in $U$, and $p \in M$ is a point related to some initial conditions. { It is worth noting that there is no precise criterion about $U$, which is considered a certain open subset. It is generally said that $x_{(1)}(t), \ldots, x_{(k)}(t)$ is a {\it generic family}. Strictly speaking, $\Phi$ may only retrieve a family of particular solutions in an open subset of $M$, which uses to be dense in $M$ \cite{Carinena2011,Lucas2020}. Nevertheless, since we focus on local aspects of differential equations and one can restrict \eqref{eq:ODEs} to a small enough subset of $M$, all previous comments on $U$ or the image of $\Phi$ can be ignored in practical examples \cite{Lucas2020,BlazquezSanz2010}.}

The celebrated Lie--Scheffers Theorem asserts that a system $X$ admits a superposition rule if and only if $X$ is a Lie system (see for instance \cite[Section 3.8]{Lucas2020}).

Every Lie algebra of vector fields $V$ on $M$ gives rise to { its} referred to as {\it generalised distribution} $\mathcal{D}^{V} \subset \mathrm{T}M$ of the form $\mathcal{D}^{V}_{x} := \{X_{x} : X \in V\}$ for all $x \in M$. In general, the rank of $\mathcal{D}^V$ is not constant on $M$, so that $\mathcal{D}^V$ constitutes a generalised subbundle of $\mathrm{T}M$  \cite{Lewis2023}. Nevertheless, the rank of $\mathcal{D}^V$ is constant on each connected component of a certain open and dense subset of $M$, where $\mathcal{D}^V$ restricts to a regular and integrable distribution in the classical sense \cite{Carinena2011}. The case of principal relevance for us arises when $V$ is finite-dimensional, since this entails that $\mathcal{D}^V$ is integrable on the whole  $M$, in the sense of Stefan--Sussmann (cf.~\cite{Lavau2018} and references therein).

A finite-dimensional Lie algebra of vector fields $V$ on $M$ is  {\it locally automorphic} if $\dim V = \dim M$ and $\mathcal{D}^{V} = \mathrm{T}M$.
A {\it locally automorphic Lie system} is a triple $(M, X, V)$, where $X$ is a Lie system on $M$ and $V$ is a VG Lie algebra of $X$ that is locally automorphic \cite{Gracia2019}.

A {\it Lie symmetry} of a Lie algebra of vector fields $V$ on $M$ is a vector field $Z$ on $M$ such that $\mathcal{L}_{Y}Z = 0$ for every $Y \in V$. The space $\mathrm{Sym}(V)$ of Lie symmetries of $V$ forms a Lie algebra. A differential form $\alpha$ on $M$ is said to be {\it invariant} relative to $V$ if $\mathcal{L}_{Y} \alpha = 0$ for all $Y \in V$. If $(M, X, V)$ is a locally automorphic Lie system, then $\mathrm{Sym}(V)$ is isomorphic to $V$, and all invariant forms relative to $V$ can be computed effectively via geometric and algebraic methods \cite{Gracia2019}. This property allows one to construct geometric structures associated with a locally automorphic Lie system in a natural way \cite{Gracia2019,Lucas2023}. For non-locally automorphic Lie systems, no general procedure is yet known. However, the following useful result, previously observed in \cite{Carballal2025}, and whose proof is straightforward, can be applied in certain cases.
\begin{prop} \label{prop:invariants}
    Let $V$ be a finite-dimensional Lie algebra of vector fields on $M$ that possesses a Lie subalgebra $V'$ that is locally automorphic.  Consider $\alpha_{1}, \ldots, \alpha_{n}$ to be the dual frame of a basis of $\mathrm{Sym}(V')$. Then, every invariant form relative to $V$ is a linear combination with real coefficients of $\alpha_{1}, \ldots, \alpha_{n}$ or their exterior products.
\end{prop}


\subsection{{ Three-dimensional Riccati system}}
\label{ex:Riccati_system}

Let us  illustrate the previous concepts by
considering the {\it Riccati system} on $\R^{3}$ with global coordinates $(x,y,z)$, given by \cite{Carinena2014,Suazo2011}
\begin{equation}
    \begin{split}
         & \dv{x}{t} = - b_{2}(t) + 2 b_{3}(t) x + 4 b_{1}(t) x^{2}, \\[2pt]
         & \dv{y}{t} = \big(b_{3}(t) + 4 b_{1}(t) x\big) y,          \\[2pt]
         & \dv{z}{t} = -\big(4b_{1}(t) x +2 b_{4}(t) \big)z,
    \end{split}
    \label{eq:ex:Riccati_system}
\end{equation}
where $b_{1}(t) ,\ldots,b_{4}(t) \in C^{\infty}(\R)$ are arbitrary $t$-dependent functions. If $(x(t), y(t), z(t))$ is a particular solution such that $y(t_{0}) = 0$ (resp., $z(t_{0}) = 0$) for some $t_{0} \in \R$, then $y(t) = 0$ (resp., $z(t) = 0$) identically. Therefore, we  restrict ourselves to studying system \eqref{eq:ex:Riccati_system} on the open subset $M := \set{(x, y, z) \in \R^{3}: yz \neq 0} \subset \R^{3}$.

Geometrically, system \eqref{eq:ex:Riccati_system} describes the integral curves of the $t$-dependent vector field $X$ on $M$ given by
\begin{equation}
    X := \sum_{i = 1}^{4} b_{i}(t) X_{i},
    \label{eq:ex:Riccati_tdep}
\end{equation}
where the vector fields
\begin{equation}
    \begin{aligned}
         & X_{1} := 4 x^{2} \pdv{x} + 4xy \pdv{y} -4 xz \pdv{z}, &  & \qquad X_{2}:= - \pdv{x},   \\[2pt]
         & X_{3}:= 2 x  \pdv{x} + y \pdv{y},                     &  & \qquad X_{4}:= -2z \pdv{z},
    \end{aligned}
    \label{eq:ex:Riccati_vf}
\end{equation}
generate a four-dimensional Lie algebra, $V$, with non-vanishing commutation relations
\begin{equation}
    [X_{1}, X_{2}] = 4X_{3} + 2X_{4}, \qquad [X_{1}, X_{3}] = -2 X_{1}, \qquad [X_{2}, X_{3}] = 2 X_{2}.
    \label{eq:ex:Riccati_cr}
\end{equation}
Hence, $X$ is a Lie system, and $V$ is an associated VG Lie algebra, isomorphic to $\mathfrak{sl}(2, \R) \oplus \R$, where $\langle X_{1}, X_{2}, X_{3} + \frac{1}{2}X_{4} \rangle \simeq \mathfrak{sl}(2, \R)$ and $ X_{4}$ spans the centre of $V$. Since $\dim V \neq \dim M$, the Lie algebra $V$ is not locally automorphic. Let $Y_{i} := X_{i}$ for $i = 1, 2, 4$, and $Y_{3}:= X_{3} + \frac{1}{2}X_{4}$. The Lie subalgebra $\langle Y_{1}, Y_{2}, Y_{3} \rangle \simeq \mathfrak{sl}(2, \R)$, although three-dimensional, is not locally automorphic, since $Y_{1} \wedge Y_{2} \wedge Y_{3}$ vanishes identically on $M$.
However,
$$
    Y_{2} \wedge Y_{3} \wedge Y_{4} = 4 y z \pdv{x} \wedge \pdv{y} \wedge \pdv{z}
$$
is non-vanishing on $M$, and $Y_{2}, Y_{3}, Y_{4}$ satisfy the sole non-vanishing commutation relation $[Y_{2}, Y_{3}] = 2 Y_{2}$. Therefore, $V':= \langle Y_{2}, Y_{3}, Y_{4} \rangle \simeq \mathfrak{b}_{2} \oplus \R$ is a Lie subalgebra of $V$ that is locally automorphic, where $\mathfrak{b}_{2}$ denotes the Borel subalgebra of $\mathfrak{sl}(2, \R)$. The Lie algebra of Lie symmetries of $V'$, let us say $\mathrm{Sym}(V')$, is spanned by
$$
    Z_{1} := y^{2} \pdv{x}, \qquad Z_{2}:= y \pdv{y}, \qquad Z_{3}:= z \pdv{z},
$$
and is isomorphic to $\mathfrak{b}_{2} \oplus \R$ as well. Now, let  $\set{\alpha_{1}, \alpha_{2}, \alpha_{3}}$ be the dual frame to $\set{Z_{1}, Z_{2}, Z_{3}}$, namely
$$
    \alpha_{1} := \frac{1}{y^{2}} \dd x, \qquad \alpha_{2} := \frac{1}{y} \dd y, \qquad \alpha_{3} := \frac{1}{z} \dd z,
$$
satisfying
\begin{equation}
    \dd \alpha_{1} = 2 \alpha_{1} \wedge \alpha_{2}, \qquad \dd \alpha_{2} = \dd \alpha_{3} = 0.
    \label{eq:ex:Riccati_MC}
\end{equation}
Thus, $\alpha := \sum_{i = 1}^{3} c_{i} \alpha_{i}$, with $c_{i} \in \R$, is an invariant form for $V$ if and only if
$$
    \cL_{Y_{1}} \alpha = 4(c_{2}- c_{3}) \dd x = 0 \Longleftrightarrow  c_{3} = c_{2}.
$$
Therefore, $\alpha = c_{1} \alpha_{1} + c_{2}(\alpha_{2} + \alpha_{3})$ is an invariant form for $V$ for all $c_{1}, c_{2} \in \R$. Among these, and taking into account \eqref{eq:ex:Riccati_MC}, we find that $\eta  := c_{1} \alpha_{1} + c_{2}(\alpha_{2} + \alpha_{3})$  is a {\it contact form}; that is, a one-form $\eta$ on $M$ such that $\eta \wedge \dd \eta$ is a volume form, whenever $c_{1} c_{2} \neq 0$. Indeed,
$$
    \eta \wedge \dd \eta =    2 c_{1} c_{2} \alpha_{1} \wedge \alpha_{2} \wedge \alpha_{3} \neq 0 \Longleftrightarrow c_{1} c_{2} \neq 0.
$$
For instance, choosing $c_{1} = c_{2} = 1$ yields the contact form
$$
    \eta = \alpha_{1} + \alpha_{2} + \alpha_{3} =   \frac{1}{y^2} \dd x +  \frac{1}{y} \dd y  + \frac{1}{z} \dd z,
$$
and the pair $(M, \eta)$ is a {\it co-orientable contact manifold}.
The {\it Reeb vector field} $\cR \in \mathfrak{X}(M)$ of $(M, \eta)$, determined by the conditions $\iota_{\cR} \eta = 1$ and $\iota_{\cR} \dd \eta = 0$, is given by $\cR := Z_{3}$. This turns system \eqref{eq:ex:Riccati_system} into a {\it contact Lie system} \cite{Lucas2023,Campoamor-Stursberg2025a}. That is, the vector fields \eqref{eq:ex:Riccati_vf} are {\it contact Hamiltonian vector fields} with respect to the contact form $\eta$, satisfying
$$
    \iota_{X_{i}} \eta = - h_{i}, \qquad \iota_{X_{i}} \dd \eta = \dd h_{i} - \cR h_{i},
$$
where $h_{i} \in C^{\infty}(M)$  given by
\begin{equation}
    h_{1} := -  \frac{4x^{2}}{y^{2}}, \qquad h_{2} := \frac{1}{y^{2}}, \qquad h_{3} := - \frac{2x}{y^{2}} -1, \qquad  h_{4} := 2,
    \label{eq:ex:Riccati_ham}
\end{equation}
are the {\it contact Hamiltonian functions} associated with the vector fields $X_{i}$, for $i = 1, \ldots, 4$. In this particular instance, the contact Hamiltonian functions \eqref{eq:ex:Riccati_ham} are constants of motion of the Reeb vector field $\cR$. A contact Lie system with this property is said to be of {\it Liouville type} \cite{Lucas2023}.

In addition,  let us consider the canonical projection
$$
    \pi_{\cR}: M \ni (x,y,z) \mapsto (x,y) \in M/\cR \simeq \R^{2}_{y\ne 0}:= \set{(x, y) \in \R^{2}: y  \neq 0} \subset \R^{2}.
$$
Since
$\cL_{\cR} \dd \eta = 0$ and $\iota_{\cR}\dd\eta=0$, there exists a  unique symplectic form $\omega \in \Omega^{2}\big(  \R^{2}_{y\ne 0} \big)$ such that $\pi_{\cR}^{*}(\omega) = \dd \eta$, which turns out to be
$$
    \omega=\frac{2}{y^3}\dd x\wedge \dd y.
$$
The contact Hamiltonian functions  (\ref{eq:ex:Riccati_ham}) satisfy the following non-zero commutation rules with respect to the Poisson bracket $\set{\cdot, \cdot}_{\omega}$ induced by $\omega$:
\begin{equation*}
    \{ h_1,h_2\}_\omega = - 4 h_3 - 2 h_4 ,\qquad \{ h_1,h_3\}_\omega =  2 h_1 , \qquad \{ h_2,h_3\}_\omega =- 2 h_2 ,
\end{equation*}
to be compared with (\ref{eq:ex:Riccati_cr}).      Consequently, applying \cite[Proposition~3.9]{Lucas2023}, every contact Lie system of Liouville type on $(M, \eta)$ is projected onto a Lie--Hamilton (LH) system on $\big(\R^{2}_{y\ne 0}, \omega\big)$, that is, on a Lie system endowed with a symplectic structure~\cite{Lucas2020,Carinena2013}.


\section{The geometry of mixed superposition rules}\label{Section:mixedsuperpositions}

Prior to analysing in detail the geometrical features of mixed superposition rules,  it is { worth  considering} a relevant hierarchy of Lie systems related to the (extended) Schr\"odinger algebra in 1+1 dimensions that faithfully illustrates the procedure and the subtleties of the geometric interpretation that will be studied afterwards.


\subsection{{ The notion of mixed superposition rules: Schr\"odinger Lie systems}}
\label{subsection:rodinger}

Consider the first-order system of ODEs on the real plane $\R^{2}$ with global coordinates $(x,y)$ given by
\begin{equation}
    \dv{x}{t} = a_{1}^{1}(t) x + a_{1}^{2}(t) y + b_{1}(t), \qquad \dv{y}{t} = a_{2}^{1}(t) x - a_{1}^{1}(t) y  + b_{2}(t),
    \label{eq:mixed:inhom}
\end{equation}
where $a_{1}^{1}(t), a_{1}^{2}(t), a_{2}^{1}(t), b_{1}(t), b_{2}(t) \in C^{\infty}(\R)$ are arbitrary $t$-dependent functions. The general solution, $(x(t), y(t))$, of this system can  be written as
\begin{equation}
    (x(t), y(t)) = (x_{\mathrm{p}}(t), y_{\mathrm{p}}(t)) + k (x_{\mathrm{h}}(t), y_{\mathrm{h}}(t)),
    \label{eq:mixed:msup}
\end{equation}
where $(x_{\mathrm{h}}(t), y_{\mathrm{h}}(t))$  is a  particular solution of the associated homogeneous system
\begin{equation}
    \dv{x}{t} = a_{1}^{1}(t) x + a_{1}^{2}(t) y, \qquad \dv{y}{t} =  a_{2}^{1}(t) x- a_{1}^{1}(t) y,
    \label{eq:mixed:hom}
\end{equation}
while $(x_{\mathrm{p}}(t), y_{\mathrm{p}}(t))$ is a particular solution of \eqref{eq:mixed:inhom}, and $k \in \R$ is a certain constant.

System \eqref{eq:mixed:inhom} is associated with the $t$-dependent vector field $X$ on $\R^{2}$ given by
\begin{equation*}
    X = a_{1}^{1}(t) \left( x \pdv{x} - y \pdv{y} \right) + a_{1}^{2}(t) y \pdv{x} + a_{2}^{1}(t) x \pdv{y} +  b_{1}(t) \pdv{x} + b_{2}(t) \pdv{y},
\end{equation*}
which takes values in the five-dimensional Lie algebra of vector fields $V$  spanned by
\begin{equation}
    X_{1} := x \pdv{x} - y \pdv{y}, \qquad X_{2} := y \pdv{x}, \qquad X_{3} := x \pdv{y}, \qquad X_{4} := \pdv{x}, \qquad X_{5} := \pdv{y},
    \label{eq:mixed:vf}
\end{equation}
with the following non-vanishing commutation relations:
\begin{equation}
    \begin{aligned}
         & [X_{1}, X_{2}] = - 2 X_{2}, \qquad &  & [X_{1}, X_{3}] = 2 X_{3}, \qquad &  & [X_{1}, X_{4}] = - X_{4},  \qquad &  & [X_{1}, X_{5}] = X_{5}, \\[2pt]
         & [X_{2}, X_{3}] = - X_{1}, \qquad   &  & [X_{2}, X_{5}] = - X_{4}, \qquad &  & [X_{3}, X_{4}] = - X_{5}.
    \end{aligned}
    \label{eq:mixed:VG_cr}
\end{equation}
Accordingly, $V$ is isomorphic to the Schr{\"o}dinger algebra $\mathcal{S}(1) = \mathfrak{sl}(2, \R) \,\overrightarrow{\oplus}_{R_{1/2}} \, \R^{2}$, where $R_{1/2}$ denotes the fundamental representation of $\mathfrak{sl}(2, \R)$ as traceless matrices on $\R^{2}$. In a similar vein, the homogeneous system \eqref{eq:mixed:hom} corresponds to the $t$-dependent vector field
\begin{equation*}
    Y := a_{1}^{1}(t) X_{1} + a_{1}^{2}(t) X_{2} + a_{2}^{1}(t) X_{3},
\end{equation*}
which takes values in the Lie algebra $V_{\mathrm{h}} := \langle X_{1}, X_{2}, X_{3} \rangle \simeq \mathfrak{sl}(2, \R)$, this being precisely the Levi factor of $V$.

Then, $X$ and $Y$ are Lie systems, and $V$ and $V_{\mathrm{h}}$ are VG Lie algebras of them, respectively. A superposition rule for $X$ was obtained in \cite{Blasco2015} in terms of three particular solutions of $X$. It is worth emphasising that this superposition rule was derived through the coalgebra method, which is applicable to LH systems \cite{Ballesteros2013,Blasco2015,Lucas2020}. More concretely,  \eqref{eq:mixed:vf} are Hamiltonian vector fields relative to the canonical symplectic form on $\R^{2}$ defined by
\begin{equation*}
    \omega := \dd x \wedge \dd y.
\end{equation*}
Some associated Hamiltonian functions for $X_{1}, \ldots, X_{5}$ are given by
\begin{equation}
    h_{1} := xy, \qquad h_{2} := \frac{1}{2} y^{2}, \qquad h_{3} := - \frac{1}{2} x^{2}, \qquad h_{4} := y, \qquad h_{5} := -x,
    \label{eq:mixed:Ham}
\end{equation}
and satisfy $\iota_{X_{i}} \omega = \dd h_{i}$ for $i= 1, \ldots, 5$. These Hamiltonian functions span a LH algebra $\cH_{\omega}$ with respect to the Poisson bracket $\set{\cdot, \cdot}_{\omega}$ induced by $\omega$, with the non-zero commutation relations:
\begin{equation}
    \begin{aligned}
         & \set{h_{1}, h_{2}}_{\omega} = 2 h_{2}, \qquad &  & \set{h_{1}, h_{3}}_{\omega} = -2 h_{3}, \qquad &  & \set{h_{1}, h_{4}}_{\omega} = h_{4}, \qquad &  & \set{h_{1}, h_{5}}_{\omega} = - h_{5}, \\[2pt]
         & \set{h_{2}, h_{3}}_{\omega} = h_{1}, \qquad   &  & \set{h_{2}, h_{5}}_{\omega} = h_{4}, \qquad    &  & \set{h_{3}, h_{4}}_{\omega} = h_{5}, \qquad &  & \set{h_{4}, h_{5}}_{\omega} = h_{0},
    \end{aligned}
    \label{eq:mixed:LH_cr}
\end{equation}
where the constant function $h_{0}:= 1$ is a central element. Consequently, $\mathcal{H}_{\omega}$ is isomorphic to the centrally extended Schr{\"o}dinger algebra
$\widehat{\mathcal{S}}(1) = \mathfrak{sl}(2, \R) \,\overrightarrow{\oplus}_{R_{1/2} \oplus R_{0}}\, \mathfrak{h}_{3}$, where $R_{0}$ denotes the trivial multiplet, and    $\mathfrak{h}_{3}=\langle h_4,h_5,h_0 \rangle$ is the three-dimensional Heisenberg  algebra. Note also that the oscillator Lie algebra also appears as the Lie  subalgebra  $\mathfrak{h}_{4}=\langle h_1,h_4,h_5,h_0 \rangle$.
It is worth remarking that $\widehat{\mathcal{S}}(1)$ is isomorphic to the so-called  two-photon Lie algebra $\mathfrak{h}_{6}\supset \mathfrak{h}_{4}\supset \mathfrak{h}_{3}$,  first considered in \cite{Zhang1990} in the context of coherent states,   which  has subsequently been  applied in classical integrable systems (see \cite{BBH09} and references therein) and, very recently, in relation
to intrinsic LH systems and their reduction to  lower dimensional systems by   invariants \cite{Campoamor2025sp4}.

Formula \eqref{eq:mixed:msup} has a nice geometrical interpretation in terms of the so-called {\it mixed superposition rules} introduced in \cite{Grabowski2013}. Following the procedure  set out therein, we seek two constants of motion $I_{1}, I_{2}$ of the $t$-dependent vector field on $(\R^{2})^{4}$ given by
\begin{equation}
    \begin{split}
        X_{E}:= & \sum_{i = 1}^{4} \left[  a_{1}^{1}(t)\left( x_{(i)} \pdv{x_{(i)}} - y_{(i)} \pdv{y_{(i)}} \right) + a_{1}^{2}(t) y_{(i)} \pdv{x_{(i)}} + a_{2}^{1}(t)  x_{(i)} \pdv{y_{(i)}} \right] \\
                & + \sum_{i = 1}^{2}\left[ b_{1}(t) \pdv{x_{(i)}} + b_{2}(t)\pdv{y_{(i)}}\right],
    \end{split}
    \label{XE4}
\end{equation}
that satisfy the regularity condition
\begin{equation}
    \pdv{(I_{1}, I_{2})}{\big(x_{(1)}, y_{(1)} \big)} \neq 0.
    \label{eq:mixed:reg}
\end{equation}
We denote $X_{E} := X \times X \times Y \times Y$  for reasons to be seen soon.  Now, consider the Hamiltonian functions on $(\R^{2})^{3}$ given by
\begin{equation}
    \begin{aligned}
                                               & h_{1}^{[3]} = \sum_{i = 1}^{3} x_{(i)} y_{(i)}, \qquad &                                                          & h_{2}^{[3]} = \frac{1}2
        \sum_{i = 1}^{3} y_{(i)}^{2}  , \qquad &                                                        & h_{3}^{[3]} = -\frac{1}{2} \sum_{i = 1}^{3} x_{(i)}^{2},                                   \\
                                               & h_{4}^{[2]} = \sum_{i = 1}^{2} y_{(i)} ,\qquad         &                                                          & h_{5}^{[2]} = - \sum_{i= 1}^{2}
        x_{(i)}, \qquad
                                               &                                                        & h_{0}^{[2]} = 2.
    \end{aligned}
    \label{eq:mixed:Ham3}
\end{equation}
A direct computation shows that
\begin{equation}
    \begin{split}
        I_{1} & := 2\left (h_{3}^{[3]} \big(h_{4}^{[2]} \big)^{2} - h_{2}^{[3]}\big(h_{5}^{[2]}\big)^{2} - h_{1}^{[3]}h_{4}^{[2]}h_{5}^{[2]} \right) - h_{0}^{[2]}\left( \big(h_{1}^{[3]}\big)^{2} + 4 h_{2}^{[3]}h_{3}^{[3]} \right) \\
              & \ = \big(x_{(3)}\big(y_{(1)}-y_{(2)}\big) - y_{(3)}\big(x_{(1)}-x_{(2)}\big)\big)^{2}
    \end{split}
    \label{eq:mixed:I1}
\end{equation}
is a constant of motion of the $t$-dependent vector field on $(\R^{2})^{3}$ given by
\begin{equation}
    \begin{split}
        X \times X \times Y & := \sum_{i = 1}^{3} \left[  a_{1}^{1}(t)\left( x_{(i)} \pdv{x_{(i)}} - y_{(i)} \pdv{y_{(i)}} \right) + a_{1}^{2}(t) y_{(i)} \pdv{x_{(i)}} + a_{2}^{1}(t)  x_{(i)} \pdv{y_{(i)}} \right] \\
                            & \quad\  + \sum_{i = 1}^{2}\left[ b_{1}(t) \pdv{x_{(i)}} + b_{2}(t)\pdv{y_{(i)}}\right],
    \end{split}
    \nonumber
\end{equation}
and considering previous structures as defined on $(\mathbb{R}^2)^4$ in the natural manner, it follows that $I_1$ is also a constant of motion of $X_{E}$. By applying the permutation $S_{34}$ on $(\mathbb{R}^2)^4$, which interchanges the variables $ (x_{(3)}, y_{(3)} ) \leftrightarrow  (x_{(4)}, y_{(4)} )$, the $t$-dependent vector field $X_E$ remains invariant and a second constant of motion
\begin{equation}
    I_{2}:=  S_{34} (I_1) =  \big(x_{(4)}\big(y_{(1)}-y_{(2)}\big) - y_{(4)}\big(x_{(1)}-x_{(2)}\big)\big)^{2},
    \label{eq:mixed:I2}
\end{equation}
arises for $X_E$.
Since the regularity condition \eqref{eq:mixed:reg} is satisfied for $I_{1}, I_{2}$,  the system $I_{1} = k_{1}^{2}, I_{2} = k_{2}^{2}$, solved in terms of $(x_{(1)}, y_{(1)})$, gives rise to the general solution of $X$ in the form
\begin{equation*}
    x_{(1)}(t) = x_{(2)}(t) + \frac{1}{W}\big(k_{2} x_{(3)}(t) - k_{1} x_{(4)}(t) \big), \qquad y_{(1)}(t) = y_{(2)}(t)+ \frac{1}{W}\big(k_{2} y_{(3)}(t) - k_{1} y_{(4)}(t)\big),
    \label{eq:mixed:superposition1}
\end{equation*}
where $(x_{(2)}(t), y_{(2)}(t))$ is a particular solution of $X$, and $W:= x_{(4)}(t) y_{(3)}(t) - x_{(3)}(t)y_{(4)}(t) \neq 0$ is the Wronskian of two independent particular solutions $(x_{(3)}(t), y_{(3)}(t)), (x_{(4)}(t), y_{(4)}(t))$ of $Y$.  Since  $W$ is a constant  for every pair of particular solutions to $Y$, we can introduce the rescaling $k_i/W\to k_i$ and it follows that the general solution to $X$ can be written in terms of a map $\Phi:(\mathbb{R}^2)^4\times \mathbb{R}^2\rightarrow \mathbb{R}^2$ of the form
\begin{equation}
    x_{(1)}  = x_{(2)}  +  \big(k_{2} x_{(3)}  - k_{1} x_{(4)} \big), \qquad y_{ (1)} = y_{(2)} +  \big(k_{2} y_{(3)}  - k_{1} y_{(4)} \big), \label{eq:mixed:superposition}
\end{equation}  for
$\big(x_{(2)},y_{(2)},x_{(3)},y_{(3)},x_{(4)},y_{(4)},k_1,k_2\big)\in (\mathbb{R}^2)^3\times \mathbb{R}^2$ as
$$
    \big(x_{(1)}(t),y_{(1)}(t)\big)=\Phi\big(x_{(2)}(t),y_{(2)}(t),x_{(3)}(t),y_{(3)}(t),x_{(4)}(t),y_{(4)}(t),k_1,k_2\big),\qquad (k_1,k_2)\in \mathbb{R}^2.
$$

\begin{remark}
    The $t$-dependent vector field $X_E$ (\ref{XE4})  admits a direct physical interpretation through its  corresponding LH system $h_E$ on $(\R^{2})^{4}$. Explicitly, let us consider the Hamiltonian functions   $\big\{ h_{1}^{[4]},h_{2}^{[4]},h_{3}^{[4]} ,h_{4}^{[2]},h_{5}^{[2]}\bigr\}$  similarly to (\ref{eq:mixed:Ham3}), identify the variables $\big(x_{(i)} ,y_{(i)}\big)$ ($i=1,2,3,4$) with the usual canonical position and conjugate momentum $(q_i,p_i)\in \T^{*} \R^{4}$,  and denote $a_1^2(t)=1/m(t)$,  $a_2^1(t)=-m(t) \Omega^2(t)$ and $b_2(t)=-f(t)$. Hence, we obtain  the $t$-dependent  Hamiltonian $h_E$   expressed as
    \begin{equation}
        h_E=\frac 1{2m(t)} \sum_{i=1}^4 p_i^2+ \frac 12 {m(t) \Omega^2(t)}   \sum_{i=1}^4 q_i^2+   a_1^1(t)  \sum_{i=1}^4 q_ip_i +b_1(t)  (p_1+p_2  ) + f(t)  (q_1+q_2  ),
        \nonumber
    \end{equation}
    which describes a four-dimensional isotropic oscillator with time-dependent mass $m(t)$ and frequency $\Omega(t)$ along  with   two momenta- or velocity-dependent potentials, determined by the coefficients $a_1^1(t)$ and $b_1(t)$, and a linear position term with a $t$-dependent force $f(t)$. Clearly, $h_E$ Poisson-commutes with the two $t$-independent constants of the motion $I_1$ (\ref{eq:mixed:I1}) and $I_2$  (\ref{eq:mixed:I2}), now reading as
    $$
        I_1= \big(q_{3} (p_{1}-p_{2} ) - p_{3} (q_{1}-q_{2} )\big)^{2} ,\qquad   I_2= \big(q_{4} (p_{1}-p_{2} ) - p_{4} (q_{1}-q_{2} )\big)^{2}  ,
    $$
    with respect to the canonical symplectic form $\omega^{[4]}=\sum_{i=1}^4\dd q_i\wedge \dd p_i$.
\end{remark}

Let us now show how the constant of motion \eqref{eq:mixed:I1} can be obtained via an algebraic approach that possesses a natural geometric counterpart. Consider the Lie algebra $\widehat{\mathcal{S}}(1)$ over a basis $\set{v_{0}, \ldots, v_{5}}$ satisfying the same commutation relations as $\set{h_{0}, \ldots, h_{5}}$. Define the morphism of Lie algebras $\phi: \widehat{\mathcal{S}}(1) \to \mathcal{H}_{\omega}$ determined by $\phi(v_{i}):= h_{i}$ for $i = 0, \ldots, 5$. This morphism gives rise to a momentum map $\J_{\phi}: \mathbb{R}^{2} \to \widehat{\mathcal{S}}(1)^{*}$, defined by
$$
    \J_{\phi}(x, y)(v_{i}) := h_{i}(x, y), \qquad (x, y) \in \R^{2}, \qquad i = 0, \ldots, 5.
$$
Hence, its pull-back
$$\J_{\phi}^{*}: C^{\infty}\big(\widehat{\mathcal{S}} (1)^{*} \big) \ni f  \mapsto f \circ \J_{\phi} \in C^{\infty}(\R^{2})$$
is a morphism of Poisson algebras, where $C^\infty\big(\widehat{\mathcal{S}}(1)^{*} \big)$ is endowed with the Kirillov--Konstant--Souriau (KKS) Poisson bracket. Consequently, $\widehat{\mathcal{S}}(1)^{*}$ becomes a Poisson manifold. As usual, consider $(v_{0}, \ldots, v_{5})$ as global linear coordinates on $\widehat{\mathcal{S}}(1)^{*}$. Clearly,
$$
    m: \widehat{\mathcal{S}}(1)^{*} \times \widehat{\mathcal{S}}(1)^{*} \ni (\xi_{1}, \xi_{2}) \mapsto \xi_{1} + \xi_{2} \in \widehat{\mathcal{S}}(1)^{*}
$$
is a Poisson map, where $\widehat{\mathcal{S}}(1)^{*}\times\widehat{\mathcal{S}}(1)^{*}$ is endowed with the product Poisson manifold structure induced by the KKS Poisson structure of $\widehat{\mathcal{S}}(1)^{*}$. Since $\langle v_{0}, v_{4}, v_{5} \rangle \simeq \h_{1}$ is an ideal of $\widehat{\mathcal{S}}(1)$, then $\set{v_{0} = v_{4} = v_{5} = 0}$ is a Poisson submanifold of  $\widehat{\mathcal{S}}(1)^{*}$, canonically  diffeomorphic to $\mathfrak{sl}(2, \R) ^{*}$ with the Poisson bivector induced by its KKS structure.  In fact, note that the Poisson bracket on $C^\infty\big(\widehat{\mathcal{S}}(1)^{*} \times \widehat{\mathcal{S}}(1)^{*} \big)$ can be restricted on functions of the form $f(0,v_1,v_2,v_3,0,0)$ in the coordinates $(v_0,\ldots,v_5)$.
This yields a Poisson map
$$
    \overline{m}: \widehat{\mathcal{S}}(1)^{*} \times \mathfrak{sl}(2, \R)^{*} \ni (\xi_{1}, \xi_{2}) \mapsto \xi_{1} + \xi_{2} \in \widehat{\mathcal{S}}(1)^{*},
$$
and, consequently,
$$
    \overline{m} \circ (m \times \mathrm{id}): \widehat{\mathcal{S}}(1)^{*} \times \widehat{\mathcal{S}}(1)^{*} \times \mathfrak{sl}(2, \R)^{*} \to \widehat{\mathcal{S}}(1)^{*}
$$
is also a Poisson map.
Since $\mathfrak{sl}(2, \R) \simeq \langle v_{1}, v_{2}, v_{3} \rangle$ is a Lie subalgebra of $\widehat{\mathcal{S}}(1)$, the morphism  $\phi: \widehat{\mathcal{S}}(1) \to \cH_{\omega}$ restricts to a morphism of Lie algebras $\phi \vert_{\mathfrak{sl}(2, \R)} :\mathfrak{sl}(2, \R)\to \cH_{{ \omega}}$, with an associated momentum map  $\J_{\phi \vert_{\mathfrak{sl}(2, \R)}}: \R^{2} \to \mathfrak{sl}(2, \R)^{*}$.

The Hamiltonian functions \eqref{eq:mixed:Ham3} read as
\begin{equation}
    \begin{aligned}
        h_{i}^{[3]} & = \left(\J_{\phi} \times \J_{\phi} \times \J_{\phi \vert_{\mathfrak{sl} (2, \R)}}\right)^{*}\big[(\overline{m} \circ (m \times \mathrm{id}))^{*}(v_{i}) \big], \\
        h_{j}^{[2]} & = \left(\J_{\phi} \times \J_{\phi} \times \J_{\phi \vert_{\mathfrak{sl} (2, \R)}}\right)^{*}\big[(\overline{m} \circ (m \times \mathrm{id}))^{*}(v_{j}) \big],
    \end{aligned}
    \nonumber
\end{equation}
where $ i = 1, 2, 3$ and $ j = 0, 4, 5$, while a  non-zero Casimir function on $\widehat{\mathcal{S}}(1)^{*}$ is given by \cite{Campoamor2005,Campoamor2020}
$$
    C := 2\big(v_{3}v_{4}^{2} - v_{2}v_{5}^{2} - v_{1}v_{4}v_{5}\big) - v_{0}\big(v_{1}^{2} + 4 v_{2}v_{3}\big).
$$
Then, the constant of motion \eqref{eq:mixed:I1} is
$$
    I_{1} = \left(\J_{\phi} \times \J_{\phi} \times \J_{\phi \vert_{\mathfrak{sl} (2, \R)}}\right)^{*}\big[(\overline{m} \circ (m \times \mathrm{id}))^{*}(C)\big].
$$

In consequence, we see that such a constant of motion can be described geometrically in terms of a Casimir, momentum maps and natural Poisson morphisms.

Formulas \eqref{eq:mixed:msup}   and \eqref{eq:mixed:superposition}  represent a particular instance of the following definition, which is an immediate generalisation of that proposed in \cite{Grabowski2013} to more general manifolds than $\mathbb{R}^n$.

\begin{defi} \label{def:mixed}
    A {\it mixed superposition rule} for a system $X$ on $M$, in terms of some systems $X_{(1)}, \ldots, X_{(k)}$ on $M_{1}, \ldots, M_{k}$, is a map $\Phi: U\subset M_{1} \times \cdots \times M_{k} \times M \to M$ such that the general solution $x(t)$ of $X$ can be cast in the form
    $$
        x(t) = \Phi\big(x_{(1)}(t), \ldots, x_{(k)}(t); p\big),
    $$
    where $x_{(i)}(t)$ is a particular solution of $X_{(i)}$, for $i = 1, \ldots, k$, and $p \in M$ is a point related to the initial conditions of $X$.
\end{defi}
Standard superposition rules are particular cases of mixed superposition rules. Moreover, one of the main results of \cite[Theorem 13]{Grabowski2013} can be naturally written for manifolds.
\begin{thm}[Extended Lie--Scheffers Theorem] \label{thm:extendedLS} A system $X$ on $M$ admits a mixed superposition rule if and only if is a Lie system.
\end{thm}
Its proof follows analogously to that of \cite[Theorem 13]{Grabowski2013}, once the notion of {\it direct products} is appropriately defined on manifolds and the superposition rule notion is a local concept  { (see also the remark following \eqref{eq:bracket_direct})}.


\subsection{On direct products} \label{subsection:direct_products}
Let us formalise the notion of direct products of vector bundles and their corresponding sections, which shall serve as a key ingredient in the subsequent analysis of mixed superposition rules (see \cite[Chapter 3]{Husemoller1994} for further details on vector bundles).

Let $\tau_{i}: E_{i} \to M_{i}$ be vector bundles for $i = 1, \ldots, k$. Their {\it direct product}, also referred to as their {\it external Whitney sum} \cite[p. 21]{Karoubi1978},  is the vector bundle over $M_{1} \times \cdots \times M_{k}$ given by the projection
$$
    \tau_{1} \times \cdots \times \tau_{k}: E_{1} \times \cdots \times E_{k} \to M_{1} \times \cdots \times M_{k},
$$
endowed with the usual fibre-wise vector space structure
$$(E_{1} \times \cdots \times E_{k})_{(x_{(1)}, \ldots, x_{(k)})} = (E_{1})_{x_{(1)}} \times \cdots \times (E_{k})_{x_{(k)}}. $$
It is worth emphasising that the direct product $E_{1} \times \cdots \times E_{k}$ may equivalently be defined as a Whitney sum of vector bundles over $M_{1} \times \cdots \times M_{k}$ as follows. Let $\pr_{i}: M_{1} \times \cdots \times M_{k} \to M_{i}$ denote the canonical projection onto the $i^{\mathrm{th}}$-factor for $i = 1, \ldots, k$. The pull-back of the bundle $\tau_i: E_i \to M_i$ along $\pr_{i}$ yields a vector bundle $\pr_{i}^{*}E_{i}$ over $M_1 \times \cdots \times M_k$. Consequently, the direct product can be identified with the Whitney sum
\begin{equation}
    E_1 \times \cdots \times E_k = \pr_1^* E_1 \oplus \cdots \oplus \pr_k^* E_k.
    \label{eq:dir_prod}
\end{equation}
Let $s_{(i)}: M_{i} \to E_{i}$ be sections of $E_{i}$ for $i = 1, \ldots, k$. Then,
$$s_{(1)} \times \cdots \times s_{(k)}: M_{1} \times \cdots \times M_{k} \to E_{1} \times \cdots \times E_{k}$$
is a section of the direct product $E_{1} \times \cdots \times E_{k}$, referred to as the {\it direct product of $s_{(1)}, \ldots, s_{(k)}$}. As a consequence of  \eqref{eq:dir_prod}, this section can equivalently be expressed as
\[
    s_{(1)} \times \cdots \times s_{(k)} = \pr_1^* s_{(1)} + \cdots + \pr_k^* s_{(k)},
\]
where each individual section $s_{(i)}$ can naturally be regarded as a section $\widetilde{s}_{(i)}$ of  $E_{1} \times \cdots \times E_{k}$ by taking its direct product with the zero sections of the remaining bundles. That is,
\begin{equation}
    \widetilde{s}_{(i)} := 0 \times \cdots \times s_{(i)} \times \cdots \times 0 = \pr_i^* s_{(i)}.
    \label{eq:lift}
\end{equation}
It is readily seen that if $\big\{s_{(i)}^{j}: j = 1, \ldots, \mathrm{rk}\, E_{i} \big\}$ constitutes a basis of local sections of $E_{i}$, then the family  $\bigcup_{i = 1}^{k}\! \big\{\widetilde{s}_{(i)}^{\,j}: j = 1, \ldots,  \mathrm{rk}\, E_{i} \big\}$ is a local basis of $E_{1} \times \cdots \times E_{k}$.
Given functions $f_{(i)} \in C^{\infty}(M_{i})$, for $i = 1, \ldots, k$, their {\it direct product} is the function
\begin{equation}
    \lambda\big(f_{(1)}, \ldots, f_{(k)}\big):= \pr_{1}^{*}f_{(1)} + \cdots + \pr_{k}^{*}f_{(k)} \in C^{\infty}(M_{1} \times \cdots \times M_{k}).
    \nonumber
\end{equation}
Note that, when each $E_i = \mathrm{T} M_i$ or $E_i = \mathrm{T}^* M_i$, the canonical identification \eqref{eq:dir_prod} yields
\begin{equation}
    \T M_{1} \times \cdots \times \T M_{k} = \T (M_{1} \times \cdots \times M_{k}), \qquad \T^{*} M_{1} \times \cdots \times\T^{*} M_{k} = \T^{*}(M_{1} \times \cdots \times M_{k}).
    \label{eq:can_iden}
\end{equation}
Accordingly, if $X_{(i)}$ are vector fields on $M_i$, their direct product defines a vector field on $M_1 \times \cdots \times M_k$ via
$$\big[X_{(1)} \times \cdots \times X_{(k)} \big]\big(x_{(1)}, \ldots, x_{(k)} \big) = X_{(1)}\big(x_{(1)} \big) + \cdots + X_{(k)}\big(x_{(k)}\big).$$
Similarly, for differential forms $\eta_{(i)}$ on $M_i$, their direct product yields the differential form on $M_1 \times \cdots \times M_k$ given by
$$\big[\eta_{(1)} \times \cdots \times \eta_{(k)} \big]\big(x_{(1)}, \ldots, x_{(k)}\big) = \eta_{(1)}\big(x_{(1)}\big) + \cdots + \eta_{(k)}\big(x_{(k)}\big).$$
Following \eqref{eq:lift},  $\widetilde{X}_{(i)}$ and $\widetilde{\eta}_{(i)}$ denote the vector field and the differential form on $M_{1} \times \cdots \times M_{k}$ corresponding to $X_{(i)}$ and $\eta_{(i)}$, respectively.

Let $V_{i}$ be a vector space of vector fields on $M_{i}$, for $i = 1, \ldots, k$. The {\it direct product} of $V_{1}, \ldots, V_{k}$ is the vector space of vector fields on $M_{1} \times \cdots \times M_{k}$ spanned by the direct products of the elements of $V_{1}, \ldots, V_{k}$. The case of special interest for us is when each $V_{i}$ is a finite-dimensional Lie algebra of vector fields. We denote by $\widetilde{V}_{i}$ the Lie algebra on $M_{1} \times \cdots \times M_{k}$ generated by the vector fields $\widetilde{X}_{(i)}$, where $X_{(i)} \in V_{i}$. Clearly, $\widetilde{V}_{i}$ is isomorphic to $V_{i}$ as a real Lie algebra, and $\widetilde{V}_{i} \cap \widetilde{V}_{j} = 0$ for all $i \neq j$.  If $X_{(i)}, Y_{(i)}$ are vector fields on $M_{i}$, the Lie bracket of the direct products $X_{(1)}\times \cdots \times X_{(k)}$ and $Y_{(1)} \times \cdots \times Y_{(k)}$ is given by
\begin{equation}
    \big[X_{(1)} \times \cdots \times X_{(k)}, Y_{(1)} \times \cdots \times Y_{(k)}\big] = \big[X_{(1)}, Y_{(1)}\big] \times \cdots \times \big[X_{(k)}, Y_{(k)}\big].
    \label{eq:bracket_direct}
\end{equation}

{
    \begin{remark}
        The proof of the Extended Lie--Scheffers Theorem~\ref{thm:extendedLS} follows that of \cite[Theorem 13]{Grabowski2013}: the proof of its two auxiliary lemmas (Lemmas 11 and 12 therein) is purely local, resting respectively on \eqref{eq:bracket_direct} above and on a pointwise linear-algebra argument on the tangent spaces of the base of the projection, neither of which uses the vector space structure of $\R^{n_i}$, which corresponds to the general manifold $M_{i}$ in our setting.
    \end{remark}
}

As a straightforward consequence, we have the following result.

\begin{prop}
    Let $V_{i}$ be a Lie algebra of vector fields on $M_{i}$ for $ i = 1, \ldots, k$. Then,
    \begin{enumerate}[label=(\arabic*), font=\normalfont]
        \item If $Z_{(i)}$ is a Lie symmetry of $V_{i}$, for $ i = 1, \ldots, k$, their  direct product $Z_{(1)} \times \cdots \times Z_{(k)}$ is a Lie symmetry of $V_{1} \times \cdots \times V_{k}$; and
        \item  If each $V_i$ is a finite-dimensional Lie algebra, then so is their direct product $V_1 \times \cdots \times V_k$.
    \end{enumerate}
\end{prop}
\begin{proof}
    Assertion (1) is an immediate consequence of identity~\eqref{eq:bracket_direct}. As for (2), it suffices to observe that \eqref{eq:bracket_direct} also shows that $V_1 \times \cdots \times V_k$ is  a Lie subalgebra of $\widetilde{V}_1 \oplus \cdots \oplus \widetilde{V}_k$, from which the claim follows.
\end{proof}

\begin{remark}
    Given a vector bundle $E \to M$, the direct product $E \times \overset{k}{\cdots} \times E$ is commonly referred to as the {\it diagonal prolongation} $E^{[k]}$ of $E$  \cite{Carinena2011,Carinena2014}. In this context, the diagonal prolongation to $M^{k}$ of a vector field $X$ and  a differential form $\eta$ are given by the products
    $$
        X^{[k]}:= X \times \overset{k}{\cdots} \times X, \qquad \eta^{[k]}:= \eta \times \overset{k}{\cdots} \times \eta,
    $$
    respectively. In contrast, the diagonal prolongation to $M^{k}$ of a vector space of vector fields $V$ on $M$ is defined as the vector space $V^{[k]}$ spanned by the diagonal prolongations of the elements of $V$. In the case that $V$ is a Lie algebra of vector fields, this construction yields a Lie algebra isomorphic to $V$. It thus follows that, in general, the diagonal prolongation $V^{[k]}$ does not coincide with the direct product $V \times \overset{k}{\cdots} \times V$. These constructions have been extensively used in Lie systems theory (see \cite{Carinena2011, Carinena2014, Lucas2023} and references therein).
\end{remark}

The following definition is a natural generalisation of \cite[Definition 7]{Grabowski2013} to the setting of manifolds that are not necessarily diffeomorphic to $\mathbb{R}^n$.

\begin{defi} \label{defi:direc_prod}
    Let $X_{(1)}, \ldots, X_{(k)}$ be $t$-dependent vector fields on $M_{1}, \ldots, M_{k}$. Their {\it direct product} is the $t$-dependent vector field $X_{(1)} \times \cdots \times X_{(k)}$ on $M_{1} \times \cdots \times M_{k}$ such that
    $$
        \big(X_{(1)} \times \cdots \times X_{(k)}\big)_{t} := \big(X_{(1)}\big)_{t} \times \cdots \times \big(X_{(k)}\big)_{t}
    $$
    for all $t \in \R$.
\end{defi}
Equivalently, the direct product $X_{(1)} \times \cdots \times X_{(k)}$ may be characterised as the unique $t$-dependent vector field on $M_{1} \times \cdots \times M_{k}$ satisfying
\[
    (\pr_{i})_{*}\big[\big(X_{(1)} \times \cdots \times X_{(k)}\big)_{t\,} \big] = \big(X_{(i)}\big)_{t},
\]
for every $t \in \mathbb{R}$ and for each $i = 1, \ldots, k$.

Before reformulating one of the principal results of \cite{Grabowski2013}, we recall the following notation from therein.  Given manifolds $M_{1}, \ldots, M_{s}$, let us denote by $\pr_{\hat{i}}$ the natural projection
$$
    \pr_{\hat{i}}: M_{1} \times \cdots \times  M_{s} \to M_{1} \times \cdots \times \widehat{M}_{i} \times \cdots \times M_{s}
$$
where $\widehat{M}_{i}$ indicates that this manifold is not included in the Cartesian product.
\begin{thm} \label{th:mixed_injective}
    Let $X_{(i)}$ be systems on $M_{i}$, for $i = 0, \ldots, k$, and let  $$X_{E}:= X_{(0)} \times \cdots \times X_{(k)}$$ denote their direct product. Consider ${ \cD^{V^{X_{E}}}}$ the generalised distribution generated by the Lie algebra $V^{X_{E}}$. Then, $X_{(i)}$ admits a mixed superposition rule in terms of $X_{(0)}, \ldots, X_{(i-1)}, X_{(i+1)}, \ldots, X_{(k)}$ if and only if
    \begin{equation}
        \label{eq:Injectivity}\T \,\pr_{\hat{i}} \vert_{\cD^{V^{X_{E}}}}: \cD^{V^{X_{E}}} \to \T \big(M_{0} \times \cdots \times \widehat{M}_{i} \times \cdots { \times} M_{k}\big)
    \end{equation}
    is an injective map on an open and dense subset of $M_{0} \times \cdots \times M_{k}$.
\end{thm}
The proof of this result is obtained by adapting the arguments developed in \cite{Grabowski2013} to appropriately chosen local coordinates on $M_{0} \times \cdots \times M_{k}$.

Let us now sketch how to derive a mixed superposition rule.  Assume that $X := X_{(0)}$ is a Lie system on $M := M_{0}$, and that the systems $X_{(j)}$ on $M_{i}$, with $j=1,\ldots,k$, satisfy the injectivity condition stated in Theorem~\ref{th:mixed_injective} for $i = 0$. Suppose further that  $\dim M_{i} = n_{i}$ for all $i = 0, \ldots, k$, and let  $x_{(i)}:= \big(x_{(i)}^{1}, \ldots, x_{(i)}^{n_{i}}\big)$ be local coordinates on each $M_{i}$. Then,
$$
    \big(x_{(0)}, \ldots, x_{(k)}\big)=\big (x_{(0)}^1, \ldots, x_{(0)}^{n_0}, \ldots, x_{(k)}^1, \ldots, x_{(k)}^{n_{k}} \big)$$
provides a local coordinate system on $M_{0} \times \cdots \times M_{k}$.
To construct a mixed superposition rule for $X$ in terms of $k$ particular solutions of $X_{(1)}, \ldots, X_{(k)}$, respectively, one seeks $n_{0}$ functionally independent, $t$-independent constants of motion  $I_{1}, \ldots, I_{n_{0}}$ of  $X_{E} = X \times X_{(1)} \times \cdots \times X_{(k)}$, such that
\begin{equation}\label{eq:PartCon}
    \pdv{(I_{1}, \ldots, I_{n_{0})}}{\big(x^{1}_{(0)}, \ldots, x^{n_{0}}_{(0)}\big)} \neq 0.
\end{equation}
The equations $I_{1} = k_{1}, \ldots, I_{n_{0}} = k_{n_{0}}$, where $k_{1}, \ldots, k_{n_{0}} \in \R$ are  constants, can then be solved locally for the variables $x_{(0)}^{1}, \ldots, x_{(0)}^{n_{0}}$ in terms of $x_{(j)}^{1}, \ldots, x_{(j)}^{n_{j}}$, for $j = 1, \ldots, k$, and the constants $k_{1}, \ldots, k_{n_{0}}$. This yields a map
$$
    x_{(0)} = \Phi\big(x_{(1)}, \ldots, x_{(k)}; k_{1}, \ldots, k_{n_{0}}\big),
$$
which provides the searched mixed superposition rule. Note that condition \eqref{eq:Injectivity} ensures that the rank of the integrable distribution $\mathcal{D}^{V^{X_E}}$ is not larger than the dimension of $M_1\times\cdots \times M_k$ and $X_E$ has at least $m_0$ functionally independent constants of motion. Moreover, it can be proved that at least $m_0$ of them must satisfy condition \eqref{eq:PartCon} (see \cite{Carinena2011} for details). This procedure coincides with the ansatz previously employed to derive the mixed superposition rule \eqref{eq:mixed:superposition} for system \eqref{eq:mixed:inhom}.

    { \begin{remark}
            Regarding Theorem~\ref{th:mixed_injective}, the injectivity of the map \eqref{eq:Injectivity} on an open and dense subset of $M_{0} \times \cdots \times M_{k}$ implies that the co-rank of $\cD^{V^{X_E}}$ is at least $n_i$ on that subset. One may then select, among the local first integrals of $\cD^{V^{X_E}}$, a family of $n_i$ functionally independent ones $\Psi = (I_1,\ldots,I_{n_i})$ satisfying $\cD^{V^{X_E}} \subseteq \ker \T \Psi$ and, again by \eqref{eq:Injectivity}, $\ker \T\Psi \cap \ker \T \pr_{\hat{i}} = 0$, which, locally, is precisely condition \eqref{eq:PartCon} for $i = 0$. That is, the fibres of $\pr_{\hat{i}}$ are complete transversal sections of the regular local foliation defined by $\Psi$. Consequently, working locally and without loss of generality, $\Psi$ may be taken to be a submersion with values in $M_i$ itself. Hence, $(\pr_{\hat{i}}, \Psi)$ is a local diffeomorphism from $\widetilde{M} := M_{0} \times \cdots \times M_{k}$ onto $\big(M_{0} \times \cdots \times \widehat{M}_{i} \times \cdots \times M_{k}\big) \times M_{i}$, and $\Phi$, the sought mixed superposition rule, is precisely the $M_i$-component of a local inverse of $(\pr_{\hat{i}}, \Psi)$. For $i = 0$, this is exactly the construction sketched above in coordinates, with $\Psi = (I_1,\ldots,I_{n_0})$ and $\Phi$ obtained by solving $I_{1} = k_{1}, \ldots, I_{n_{0}} = k_{n_{0}}$ for the local coordinates $x_{(0)}^{1}, \ldots, x_{(0)}^{n_{0}}$ of $M_{0}$ in terms of the local coordinates $x_{(j)}^{1}, \ldots, x_{(j)}^{n_{j}}$ of $M_{1} \times \cdots \times M_{k}$, for $j = 1, \ldots, k$, and the constants $k_{1}, \ldots, k_{n_{0}}$. In fact, note that the use of local coordinates reduces the case of the mixed superposition rule on manifolds to the case in linear spaces.
        \end{remark}}


\section{Lie systems with natural mixed superposition rules}\label{Section:MixedNatural}

As a consequence of the Extended Lie--Scheffers Theorem~\ref{thm:extendedLS}, every Lie system admits a mixed superposition rule. Nevertheless, in the general setting, there exists no systematic procedure for constructing alternative systems, distinct from the original one, that can be used to derive such a mixed superposition rule. Let us analyse two families of Lie systems which naturally admit mixed superposition rules formulated in terms of systems other than the original. The first class admits a geometric characterisation, whereas the second is described by means of algebraic properties of certain VG Lie algebras.


\subsection{On imprimitive Vessiot--Guldberg Lie algebras}
The first class of Lie systems which admit a natural mixed superposition rule is that of Lie systems possessing an {\it imprimitive} VG Lie algebra \cite{GonzalezLopez1992,Shnider1984,Shnider1984a}. To formulate the problem precisely, let us begin by considering a representative example.

Let us consider the Riccati system discussed in Subsection~\ref{ex:Riccati_system}, defined on the open subset $M = \set{(x,y, z) \in \R^{3}: yz \neq 0} \subset \R^{3}$, and given by \eqref{eq:ex:Riccati_system}. If $(x(t), y(t), z(t))$ is a particular solution of system \eqref{eq:ex:Riccati_system} in $M$, then $y(t)$ and $z(t)$ can be expressed in terms of $x(t)$, where the latter satisfies the Riccati equation
\begin{equation}
    \dv{x}{t} = - b_{2}(t) + 2 b_{3}(t) x + 4 b_{1}(t) x^{2},
    \label{eq:Riccati_line}
\end{equation}
and the functions $b_{1}(t), b_{3}(t), b_{4}(t)$ by means of quadratures. In fact, note that the differential equations for $y(t)$ and $z(t)$ in \eqref{eq:ex:Riccati_system} can be integrated straightforwardly once $x(t)$ is known. This naturally raises the question as to whether the Riccati system \eqref{eq:ex:Riccati_system} admits a mixed superposition rule in terms only of particular solutions of the Riccati equation \eqref{eq:Riccati_line}. We proceed to demonstrate that this is not the case by applying Theorem~\ref{th:mixed_injective}.

Let $X$ be the $t$-dependent vector field \eqref{eq:ex:Riccati_tdep} associated with the Riccati system \eqref{eq:ex:Riccati_system}, and consider the projection $\pi: M \ni (x, y, z) \mapsto x \in \R$. Then,
\begin{equation}
    \pi_{*}X = \sum_{i = 1}^{4}b_{i}(t) \pi_{*}X_{i} = \big(4 b_{1}(t) x^{2}  - b_{2}(t)  + 2 b_{3}(t) x\big) \pdv{x}
    \label{eq:Riccati_Line_tdep}
\end{equation}
is the $t$-dependent vector field associated with \eqref{eq:Riccati_line}. Consider now the  direct product $X_{E}$ of $X$ and $k \geq 1$ copies of $\pi_{*}X$, explicitly given by
\begin{multline}
    X_{E}  = X \times (\pi_{*}X)^{[k]} = b_{1}(t) \left( 4 x_{(1)}y_{(1)} \pdv{y_{(1)}} - 4 x_{(1)}z_{(1)} \pdv{z_{(1)}}  + \sum_{i = 1}^{k+1} 4x_{(i)}^{2} \pdv{x_{(i)}}\right)   \\
    - b_{2}(t)\sum_{i = 1}^{k+1}  \pdv{x_{(i)}} + b_{3}(t) \left( y_{(1)} \pdv{y_{(1)}} +\sum_{i = 1}^{k+1} 2 x_{(i)} \pdv{x_{(i)}} \right) - 2 b_{4}(t) z_{(1)} \pdv{z_{(1)}},
    \label{eq:direct_counter}
\end{multline}
and let ${ \cD^{V^{X_{E}}}} \subset \T \big(M \times \R^{k}\big)$ be the generalised distribution generated by $V^{X_{E}}$. It follows that, whenever $b_{4}(t) \neq 0$, the vector field $z_{(1)} \pdv{z_{(1)}} $ takes values on $\cD^{V^{X_{E}}}$, implying that the restriction of the tangent map $\T \, \pr_{\hat{1}}: \T \big(M \times \R^{k}\big) \to \T \R^{k}$ to ${ \cD^{V^{X_{E}}}}$, where $\pr_{\hat{1}}: M \times \R^{k} \to \R^{k}$ is the canonical projection onto the second factor, fails to be injective. Consequently, by Theorem~\ref{th:mixed_injective}, the Riccati system \eqref{eq:ex:Riccati_system} does not admit a mixed superposition rule depending solely on particular solutions of the Riccati equation \eqref{eq:Riccati_line}.

Nevertheless, the situation changes abruptly when, instead of the direct product \eqref{eq:direct_counter}, one considers the direct product $X_{E} := X^{[2]} \times \pi_{*}X $  on $M^{2} \times \R$, explicitly given by
\begin{align}
    X_{E} & = X^{[2]} \times \pi_{*}X\nonumber                                                                                                                                                                                                         \\
          & = b_{1}(t) \left[\sum_{i = 1}^{2} \left(4x_{(i)}^{2} \pdv{x_{(i)}} + 4 x_{(i)}y_{(i)} \pdv{y_{(i)}} - 4 x_{(i)}z_{(i)} \pdv{z_{(i)}} \right) + 4 x_{(3)}^{2} \pdv{x_{(3)}} \right] \nonumber                                               \\
          & \quad - b_{2}(t) \sum_{i = 1}^{3} \pdv{x_{(i)}}     + b_{3}(t)  \left[\sum_{i = 1}^{2} \left( 2x_{(i)} \pdv{x_{(i)}} + y_{(i)} \pdv{y_{(i)}} \right)+2x_{(3)} \pdv{x_{(3)}} \right] - 2 b_{4}(t)  \sum_{i = 1}^{2} z_{(i)} \pdv{z_{(i)}} .
    \label{eq:direct_counter2}
\end{align}
In this context, the restriction of $\T \, \pr_{\hat{1}}: \T \big({ M^{2} \times \R}\big) \to \T( { M \times \R})$ to the generalised distribution ${ \cD^{V^{X_{E}}}}$ spanned by $V^{X_{E}}$ is injective on the open and dense subset
$$
    U := \set{\big(x_{(1)}- x_{(2)}\big)\big(x_{(1)}-x_{(3)}\big)\big(x_{(2)} - x_{(3)}\big) \neq 0} \subset M^{2} \times \R.
$$
Moreover, the restriction of $\cD^{V^{X_{E}}}$ to $U$ is a rank-four regular and integrable distribution $\cD^{V^{X_{E}}} \vert_{U} \subset \T U$. Hence, its annihilator $\big(\cD^{V^{X_{E}}} \vert_{U}\big)^{\circ} \subset \T^{*} U$ is a rank-three vector subbundle of $\T^{*} U$. Employing standard techniques from the theory of exterior differential systems \cite{Bryant1991,Flanders1963}, one obtains that the functions $I_{1}, I_{2}, I_{3} \in C^{\infty}(U)$ given by
\begin{equation}
    I_{1} := \frac{y_{(1)}\big(x_{(2)}- x_{(3)}\big)}{y_{(2)}\big(x_{(1)}- x_{(3)}\big)}, \qquad I_{2}:= \frac{z_{(2)}\big(x_{(1)}- x_{(2)}\big)}{z_{(1)}y_{(1)}^{2}}, \qquad I_{3}:= \frac{y_{(1)} y_{(2)}}{x_{(1)}- x_{(2)}}
    \label{eq:constants}
\end{equation}
are such that their differentials $\dd I_{1}, \dd I_{2}, \dd I_{3}$ generate $\big(\cD^{V^{X_{E}}} \vert_{U}\big)^{\circ}$. Therefore, the functions \eqref{eq:constants} are $t$-independent constants of motion of $X_{E}$ (\ref{eq:direct_counter2}). Since they fulfil the regularity condition $$\pdv{(I_{1}, I_{2}, I_{3})}{\big(x_{(1)}, y_{(1)}, z_{(1)}\big)} \neq 0,$$ the system $I_{1} = k_{1}\neq 0, I_{2} = k_{2}\neq 0, I_{3} = k_{3}\neq 0$, yields  the general solution of the Riccati system \eqref{eq:ex:Riccati_system} in the form
\begin{equation}
    \begin{split}
         & x_{(1)}(t) = x_{(2)}(t)+ \frac{ k_{1} \big[x_{(2)}(t) - x_{(3)}(t)\big]y_{(2)}^{2}(t)   }{k_{3}\big[x_{(2)}(t) - x_{(3)}(t)\big] - k_{1}y_{(2)}^{2}(t)}, \\[2pt]
         & y_{(1)}(t) = \frac{k_{1}k_{3} y_{(2)}(t) \big[x_{(2)}(t) - x_{(3)}(t)\big]}{k_{3}\big[x_{(2)}(t) - x_{(3)}(t)\big] - k_{1}y_{(2)}^{2}(t)},               \\[2pt]
         & z_{(1)}(t) = \frac{z_{(2)}(t)}{k_{2}k_{3}} \left[ \frac{1}{k_{1}} - \frac{y_{(2)}^{2}(t)}{k_{3}\big[x_{(2)}(t) - x_{(3)}(t)\big]} \right],
    \end{split}
    \label{eq:mixed:Riccati_system}
\end{equation}
where $\big(x_{(2)}(t), y_{(2)}(t), z_{(2)}(t)\big)$ is a particular solution of the Riccati system \eqref{eq:ex:Riccati_system}, while $x_{(3)}(t)$ is a particular solution of the Riccati equation \eqref{eq:Riccati_line}, and $k_{1}, k_{2}, k_{3}$ are non-zero real constants.

The aforementioned success naturally prompts the question of whether the techniques employed might also be extended to Lie systems exhibiting analogous properties to the Riccati system \eqref{eq:ex:Riccati_system}. In this regard, we stress that the Riccati equation \eqref{eq:Riccati_line}, utilised in the mixed superposition rule \eqref{eq:mixed:Riccati_system} for the Riccati system \eqref{eq:ex:Riccati_system}, admits a geometric derivation as follows.   More precisely, let $\cD \subset \T M$ be the regular and { integrable} distribution spanned by the vector fields $\partial_{y}, \partial_{z}$. The vector fields \eqref{eq:ex:Riccati_vf}, spanning the VG Lie algebra $V$ of the Riccati system, leave the distribution $\cD$ invariant, that is, $[V, \Gamma(\cD)] \subset \Gamma(\cD)$. Consider next the foliation $\mathcal{F}$ of $M$ induced by $\cD$, whose leaves are the connected components of $\set{x = c} \cap M$, with $c \in \R$.  { The corresponding leaf space  $M/\mathcal{F}$ is the disjoint union $M / \mathcal{F} = \R \times \set{+1, -1}^{2}$ of four real lines and the quotient projection $M \ni (x, y, z) \mapsto (x, \mathrm{sgn}(y), \mathrm{sgn}(z)) \in M / \mathcal{F}$ is a submersion, where $\mathrm{sgn}(\cdot)$ denotes the sign function. Since the projections of the vector fields \eqref{eq:ex:Riccati_vf} take the same expression on each of the four connected components of $M / \mathcal{F}$, it suffices to work on a single component with coordinate $x$ and restrict to the submersion $\pi: M \ni (x, y, z) \mapsto x \in \R$,   through which the $t$-dependent vector field $X$ associated with the Riccati system projects onto the $t$-dependent vector field \eqref{eq:Riccati_Line_tdep} associated with the Riccati equation. }

The property that the VG Lie algebra $V$ associated with the Riccati system \eqref{eq:ex:Riccati_system} leaves invariant a {\it proper}, namely of rank $1 \leq \mathrm{rk} \, \cD \leq \dim M -1$, regular and integrable distribution $\cD$ is commonly referred to as $V$ being {\it imprimitive} \cite{GonzalezLopez1992,Shnider1984, Shnider1984a}.

\begin{defi}
    A finite-dimensional Lie algebra of vector fields $V$ on $M$ is said to be {\it imprimitive} if there exists a  regular and integrable distribution $\cD \subset \T M$ that is proper such that $[V, \Gamma(\cD)] \subset \Gamma(\cD)$. Otherwise, $V$ is termed {\it primitive}.
\end{defi}

These two classes of Lie algebras play a fundamental role in the classification of finite-dimensional Lie algebras of vector fields on the real plane at generic points \cite{GonzalezLopez1992}. It was observed in  \cite{Shnider1984,Shnider1984a} that Lie systems admitting an imprimitive VG Lie algebra possess a mixed superposition in the sense of  Definition~\ref{def:mixed}. Nevertheless, the argument presented therein relies on the local Lie group action induced by such an imprimitive VG Lie algebra, the explicit determination of which is, in frequent practical situations, non-trivial.

Let us now show how the later result can be addressed by means of our mixed superposition rules formalism.

\begin{prop} \label{prop:exact}
    Let $V$ be an imprimitive Lie algebra on $M$ such that there exists a regular, proper, { integrable} distribution $\cD \subset \T M$
    such that $[V, \Gamma(\cD)] \subset \Gamma(\cD)$, and let $\mathcal{F}$ be the foliation induced by $\cD$. Assume that the leaf space $M / \mathcal{F}$ is a manifold and that the projection $\pi: M \to M / \mathcal{F}$ is a submersion. Then,
    \begin{equation}
        \begin{tikzcd}
            0 & {V \cap \Gamma(\mathcal{D})} & V & {\pi_{*}V} & 0
            \arrow[hook,from=1-1, to=1-2]
            \arrow[hook, from=1-2, to=1-3]
            \arrow["{\pi_{*}}", from=1-3, to=1-4]
            \arrow[from=1-4, to=1-5]
        \end{tikzcd}
        \label{eq:exact}
    \end{equation}
    is an exact sequence of Lie algebra morphisms.
\end{prop}

\begin{remark}
    The fact that \eqref{eq:exact} is an exact sequence of Lie algebra morphisms  means, precisely, that $V$ is a Lie algebra extension of $\pi_{*}V$ by $V \cap \Gamma(\cD)$  \cite{Alekseevsky2004}. Observe that, if $V$ is a simple Lie algebra, then $\pi_{*}V$ is isomorphic to $V$ or zero.
\end{remark}

Let $X$ be a Lie system admitting a VG Lie algebra $V$ that satisfies the hypotheses of Proposition~\ref{prop:exact}. Our aim is now to demonstrate the existence of a natural mixed superposition rule for $X$, in terms of a suitable number of particular solutions of $X$ and $\pi_{*}X$.
Prior to this, let us recall some relevant concepts (see \cite{Grabowski2013} for further details).

\begin{defi}
    A vector space of vector fields $V$ on $M$ is said to possess or admit a {\it modular basis} if there exists a basis of $V$ formed by linearly independent vector fields at a  {\it generic point} of $M$, that is, at points of an open and dense subset of $M$.
\end{defi}

{ \begin{remark}
        One of the key steps in the proof of the Lie--Scheffers Theorem by Cari\~nena, Grabowski and Marmo \cite{Carinena2007} is that, given a VG Lie algebra $V$ on $M$, there exists some $k \geq 1$ such that the diagonal prolongation $V^{[k]}$ admits a modular basis \cite[Proposition 3.12]{Lucas2020}. This immediately yields the necessary condition
        $$
            \dim V \leq k \dim M,
        $$
        commonly known as {\it Lie's condition} \cite{Carinena2011}; equivalently, it provides a lower bound $k \geq \dim V / \dim M$ for any admissible $k$ and a fixed $V$. It is worth stressing, however, that a matching sufficient condition is still missing: the minimal such $k$ remains, in general, an open problem in the theory of Lie systems.
    \end{remark}}
The following result is a  straightforward consequence of the above definition.
\begin{prop} \label{prop:modular}
    Let $V$ be a Lie algebra of vector fields on $M$ that possesses a modular basis. Then,
    \begin{enumerate}[label=(\arabic*), font=\normalfont]
        \item Every basis of $V$ is modular; and
        \item Every Lie subalgebra $V' \subset V$ also admits a modular basis.
    \end{enumerate}
\end{prop}

Let us now prove the main result of this subsection.

\begin{thm} \label{thm:imprimitive} Let $X$ be a Lie system on $M$ admitting a VG Lie algebra $V$ which is imprimitive and leaves invariant a regular, proper, and integrable distribution $\mathcal{D} \subset \T M$, and let $\mathcal{F}$ be the foliation induced by $\mathcal{D}$. Assume that the leaf space $M / \mathcal{F}$ is a manifold and that the projection $\pi: M \to M/\mathcal{F}$ is a submersion.
    Then, the following statements hold:
    \begin{enumerate}[label=(\arabic*), font=\normalfont]
        \item If $V \cap \Gamma(\cD) = 0$, then $X$ admits a mixed superposition rule depending on $k$ particular solutions of $\pi_{*}X$, where $k \geq 1$ is such that $(\pi_{*}V)^{[k]}$ possess a modular basis; and
        \item If $V \cap \Gamma(\cD) \neq 0$, then $X$  admits a mixed superposition rule in terms of $s$ particular solutions of $X$ and $k$ particular solutions of $\pi_{*}X$, where $s, k \geq 1$.  Moreover, if $V^{X} \cap \Gamma(\cD) \neq 0 $, then $X$ does not admit a mixed superposition rule depending solely on particular solutions of  $\pi_{*}X$.
    \end{enumerate}
\end{thm}
\begin{proof} Let us prove both points.
    \begin{itemize}
        \item[(1)] Since $V \cap \Gamma(\cD) = 0$, it follows that $\pi_{*}V$ is isomorphic to $V$. Let $\{X_1,\ldots,X_r\}$ be a basis of $V$ and  let  $\set{(\pi_{*}X_{1})^{[k]}, \ldots, (\pi_{*}X_{r})^{[k]}}$ be a modular basis of $(\pi_{*}V)^{[k]}$ around a generic point of $(M / \mathcal{F})^{k}$ for some $k$. Note that $k$ exists because $\pi_*X_1,\ldots,\pi_*X_r$ form a basis of $\pi_*V$ (see \cite[Proposition 1.17]{Carinena2011}  for details).
              Consider the Lie algebra of vector fields $V_{E}$ on $M \times (M / \mathcal{F})^{k}$ spanned by the vector fields $X_{1} \times (\pi_{*}X_{1})^{[k]}, \ldots, X_{r} \times (\pi_{*}X_{r})^{[k]}$.  Moreover,  $\big\{ X_{1} \times (\pi_{*}X_{1})^{[k]}, \ldots, X_{r} \times (\pi_{*}X_{r})^{[k]}\big\}$ is a modular basis of $V_{E}$, which is also isomorphic to $V$.  Hence, $ \T \, \pr_{\hat{1}} \vert_{\cD^{V^{X_{E}}}}: \cD^{V^{X_{E}}} \to \T (M / \mathcal{F})^{k}$ is injective on an open and dense subset of $M \times (M /\mathcal{F})^{k}$ and $X$ admits a mixed superposition rule in terms of $k$ particular solutions of $\pi_*X$.

        \item[(2)] Let $\{X_1,\ldots,X_r\}$ be again a basis of $V$ so that $\pi_*X_1,\ldots,\pi_*X_m$ is a basis of $\pi_*V$ and $X_{m+1},\ldots,X_r\in\Gamma(\mathcal{D})$. Let $k$ be such that $\big\{(\pi_*X_1)^{[k]}, \ldots, (\pi_*X_{m})^{[k]}\big\}$ is a modular basis of { $(\pi_{*}V)^{[k]}$}. Consider also some $s$ so that $X^{[s]}_{m+1},\ldots,X^{[s]}_r$  are linearly independent at a generic point of $M^s$.
              The direct products $ X^{[s+1]}_{1}  \times (\pi_{*}X_{1})^{[k]}, \ldots ,X^{[s+1]}_{r} \times (\pi_{*}X_{r})^{[k]}$ form a modular basis of a Lie algebra $V_E$, which is therefore isomorphic to $V$.
              Moreover, $V_{E}$ is a VG Lie algebra of the direct product $X_{E}:= X^{[s+1]}  \times (\pi_{*}X)^{[k]}$.
              It  follows that $\T \, \pr_{\hat{1}} \vert_{\cD^{V_{{E}}}}: \cD^{V_{E}} \to \T M^s\times \T(M/ \mathcal{F})^{k}$ is injective at a generic point. Therefore, a mixed superposition rule depending on $s$ particular solutions of $X$ and $k$ of $\pi_*X$ can be obtained.  Note that $s>0$, as otherwise the vector fields $X_{m+1}\times (\pi_*X_{m+1})^{[k]} ,\ldots, X_r\times (\pi_*X_r)^{[k]}$ project onto zero via $\T{\rm pr}_{\hat{1}}|_{\mathcal{D}^{V_E}}: \mathcal{D}^{V_{E}} \subset \T M\times\T (M/\mathcal{F})^k\rightarrow \T (M/\mathcal{F})^k$.    If $V^X\cap \Gamma(\mathcal{D})\neq 0$, the previous situation will exclude mixed superposition rules for $s=0$ and every possible VG Lie algebra $V$ of $X$.
    \end{itemize}
\end{proof}
The mixed superposition rule \eqref{eq:mixed:Riccati_system} for the Riccati system \eqref{eq:ex:Riccati_system}, involving of one particular solution of  \eqref{eq:ex:Riccati_system} and one particular solution of the Riccati equation \eqref{eq:Riccati_line}, exemplifies the applicability of Theorem~\ref{thm:imprimitive} when $V \cap \Gamma(\cD)  \neq 0$. Let us now provide an example illustrating  the complementary case, namely $V \cap \Gamma(\mathcal{D}) = 0$.

\begin{example} \label{ex:ho}
    Let us consider the $t$-dependent harmonic oscillator system on $\T^{*} \R_{0}$, { where $\R_{0}$ denotes the open subset $\R_{0} := \R - \set{0} \subset \R$}, given by the $t$-dependent Hamiltonian
    $$
        h := \frac{1}{2}p^{2} + \frac{1}{2} \Omega^{2}(t) q^{2},
    $$
    where $\Omega(t)$ is a $t$-dependent frequency. The associated Hamilton's equations with respect to the canonical symplectic form $\omega = \dd q \wedge \dd p$ on $\T^{*}\R_{0}$ read
    \begin{equation}
        \dv{q}{t} = p, \qquad \dv{p}{t} = -\Omega^{2}(t) q.
        \label{eq:tdho:system}
    \end{equation}
    It is well known that this a LH system relative to $\omega = \dd q \wedge \dd p$, with associated   $t$-dependent vector field  $X := X_{3} +\Omega^{2}(t) X_{1}$, where the vector fields
    \begin{equation}
        X_{1}:= - q \pdv{p}, \qquad X_{2}:= -q \pdv{q} + p \pdv{p}, \qquad X_{3}:= p \pdv{q}
        \label{eq:tdho:vf}
    \end{equation}
    span a VG Lie algebra $V \simeq \mathfrak{sl}(2, \R)$ of $X$ \cite{Ballesteros2013,Ballesteros2015,Blasco2015} with commutation relations given by
    \begin{equation}
        [X_{1},X_{2} ]=2 X_1,\qquad { [X_{1},X_{3}]=  X_{2}},\qquad
        [X_{2},X_{3} ]=2 X_3.
        \label{eq:ho}
    \end{equation}
    Under the diffeomorphism
    \begin{equation}
        \T^{*} \R_{0} \ni (q, p) \mapsto (x:= -p/q, y:= 1/q) \in \R^{2}_{y \neq 0},
        \label{eq:tdho:diff}
    \end{equation}
    the vector fields \eqref{eq:tdho:vf} take the form
    \begin{equation}
        X_{1} = \pdv{x}, \qquad X_{2} = 2 x \pdv{x} + y \pdv{y}, \qquad X_{3} = x^{2}\pdv{x} + xy \pdv{y}.
        \label{eq:tdho:I5_vf}
    \end{equation}
    Remarkably, they are the generators of the so-called I$_{5}$-LH class on the plane \cite{Ballesteros2015}. It is now clear that the VG Lie algebra $V$ spanned by $X_{1}, X_{2}$ and $X_{3}$ is imprimitive, since { it} leaves invariant the regular and integrable rank-one distribution $\mathcal{D} \subset \R^{2}_{y \neq 0}$ spanned by the vector field $\partial_{y}$.  { The leaves of the foliation $\mathcal{F}$ induced by $\mathcal{D}$ are given by the connected components of the semi-lines $\set{(c, y): y \neq 0}$, with $c \in \R$. The leaf space $\R^{2}_{y \neq 0}/ \mathcal{F}$ is the disjoint union $\R^{2}_{y \neq 0}/ \mathcal{F} = \R \times \set{+1, -1}$ of two real lines and the quotient projection $\R^{2}_{y \neq 0} \ni (x, y) \mapsto (x, \mathrm{sgn}(y)) \in \R^{2}_{y \neq 0} / \mathcal{F}$ is a submersion. Since the projections of the vector fields \eqref{eq:tdho:I5_vf} take the same expression on each of the connected components of the leaf space, we work in a single connected component and restrict to the submersion $\pi : \R^{2}_{y \neq 0} \ni (x, y) \mapsto x \in \R$, through which the $t$-dependent vector field $X =  X_{3} +\Omega^{2}(t) X_{1}$ projects onto
            \begin{equation}
                \pi_{*}X =  \pi_{*}X_{3} +\Omega^{2}(t) \pi_{*}X_{1}.
                \label{eq:tdho:Riccati_tdep}
            \end{equation}}
    The projected vector fields,
    \begin{equation}
        \pi_{*}X_{1} = \pdv{x}, \qquad \pi_{*}X_{2} = 2x \pdv{x}, \qquad \pi_{*}X_{3} = x^{2} \pdv{x},
        \label{eq:tdho:Riccativf}
    \end{equation}
    span the projection $\pi_{*}V$ of $V$, which is isomorphic to $\mathfrak{sl}(2, \R)$ as well, since $V$ is simple, and  hence { fulfils} the commutators  (\ref{eq:ho}).
    The ODE  associated with $\pi_{*}X$ is a Riccati equation of the form
    \begin{equation}
        \dv{x}{t} =  x^{2} +\Omega^{2}(t).
        \label{eq:tdho:Riccati}
    \end{equation}
    The diagonal prolongation $(\pi_{*}V)^{[3]}$ of $\pi_{*}V$ to $\R^{3}$ admits a modular basis. Hence, Theorem~\ref{thm:imprimitive} ensures that the system $X$ describing the Hamilton equations \eqref{eq:tdho:system} of the $t$-dependent harmonic oscillator admits a mixed superposition rule in terms of three particular solutions of the Riccati equation \eqref{eq:tdho:Riccati}. This feature was already observed in \cite{Carinena2008}, but is now addressed geometrically by means of Theorem~\ref{thm:imprimitive}.
\end{example}


\subsection{On semidirect sums of Vessiot--Guldberg Lie algebras}
Let us examine Lie systems possessing a VG Lie algebra that admits a semidirect sum decomposition, thereby enabling the construction of mixed superposition rules via a natural procedure stemming from the underlying algebraic structure.

Inhomogeneous linear systems of ODEs, such as system \eqref{eq:mixed:inhom}, admit a mixed superposition rule. In the context of Lie systems theory, such systems are referred to as {\it affine Lie systems}, since their associated $t$-dependent vector fields can be regarded as curves taking values in affine Lie algebras \cite{Carinena2003, Carinena2011, Colombo2025}. More specifically, an {\it affine Lie system} on $\mathbb{R}^{n}$ is a Lie system $X$ of the form
\begin{equation}
    \dv{t} \begin{pmatrix}
        x_{1}  \\
        \vdots \\
        x_{n}
    \end{pmatrix} = A(t) \begin{pmatrix}
        x_{1}  \\
        \vdots \\
        x_{n}
    \end{pmatrix} + \begin{pmatrix}
        b_{1}(t) \\
        \vdots   \\
        b_{n}(t)
    \end{pmatrix},
    \label{eq:affine_LieSystem}
\end{equation}
where $\R \ni t \mapsto A(t) \in \mathfrak{gl}(n,\R)$ is a curve taking values in the Lie algebra of $(n \times n)$-matrices with real coefficients, and $b_{1}(t), \ldots, b_{n}(t) \in C^{\infty}(\mathbb{R})$ are arbitrary $t$-dependent functions. System \eqref{eq:affine_LieSystem} admits a VG Lie algebra isomorphic to the affine Lie algebra $\mathfrak{gl}(n, \mathbb{R}) \overrightarrow{\oplus} \mathbb{R}^{n}$, which is spanned by the vector fields $x_{i} \partial_{x_{j}}, \partial_{x_{i}}$, for $i, j = 1, \ldots, n$ \cite{Carinena2011}.

Let $Y$ denote the homogeneous system associated with $X$, namely
\begin{equation}
    \dv{t} \begin{pmatrix}
        x_{1}  \\
        \vdots \\
        x_{n}
    \end{pmatrix} = A(t) \begin{pmatrix}
        x_{1}  \\
        \vdots \\
        x_{n}
    \end{pmatrix}.
    \label{eq:affine:inhomogeneous}
\end{equation}
Then, the general solution $x(t)$ of $X$ can be written as
$$x(t) = x_{(0)}(t) + \sum_{i = 1}^{n} k_{i} x_{(i)}(t),$$
where $x_{(0)}(t)$ is a particular solution of $X$, the functions $x_{(1)}(t), \ldots, x_{(n)}(t)$ are linearly independent particular solutions of $Y$, and $k_{1}, \ldots, k_{n}$ are constants. Furthermore, if the Wronskian $W := \det(x_{(1)}(t), \ldots, x_{(n)}(t))$ is non-zero at a certain $t_{0} \in \mathbb{R}$, then $W(t) \neq 0$ for all $t \in \mathbb{R}$. { However, the Wronskian is constant $W(t) \equiv W \neq 0$ if and only if the curve $\R \ni t \mapsto A(t) \in \mathfrak{gl}(n, \R)$ takes values in the Lie subalgebra $\mathfrak{sl}(n, \R) \subset \mathfrak{gl}(n,\R)$ consisting of traceless matrices. In such a case, the expression
        $$x(t) = x_{(0)}(t) + \frac{1}{W} \sum_{i = 1}^{n} k_{i} x_{(i)}(t),$$ }
also defines a mixed superposition rule   $\Phi:\mathbb{R}^{ n(n+1)} \times \mathbb{R}^n\rightarrow \mathbb{R}^n$ for $X$ of the form
$$
    x=x_{(0)}+\sum_{i=1}^nk_ix_{(i)},\qquad (k_1,\ldots,k_n)\in \mathbb{R}^n,
$$
provided that we have applied the rescaling $k_i/W\to k_i$.
This coincides with the mixed superposition rule presented in equation \eqref{eq:mixed:superposition} for the case of system \eqref{eq:mixed:inhom}.

We now show that, beyond affine Lie systems, any Lie system $X$ possessing a VG Lie algebra $V$ that is a semidirect sum, namely $V = V_{1}  \overrightarrow{\oplus} V_{2}$, also admits a mixed superposition rule involving particular solutions of $X$ and particular solutions of $Y$, where $Y$ is the part of $X$ taking values in $V_{1}$, that is, $Y(t) \in V_{1}$ and $X(t) - Y(t) \in V_{2}$ for all $t \in \R$.
\begin{prop} \label{prop:mixed}
    Let $X$ be a Lie system on $M$ that possesses a VG Lie algebra given by a semidirect sum $V = V_{1} \overrightarrow{\oplus} V_{2}$, and let $Y$ be the part of $X$ taking values in $V_{1}$.  Then, $X$ admits a mixed superposition rule depending on $s$ particular solutions of $X$ and $k$ particular solutions of $Y$, where $s, k \geq 1$ are such that $V_{1}^{[k]}$ and $V_{2}^{[s]}$ possess a modular basis.
\end{prop}
\begin{proof}
    Let $\big\{Y_{1}^{[k]}, \ldots, Y_{\ell}^{[k]}\big\}$ and $\big\{X_{1}^{[s]}, \ldots, X_{m}^{[s]}\big\}$ be modular bases of $V_{1}^{[k]}$ and $V_{2}^{[s]}$, respectively.
    Recall that their existence is ensured by \cite[Proposition 1.17]{Carinena2011}. Consider the vector fields on $M\times M^k\times M^s$ given by
    $$
        Y_i\times Y^{[k]}_i\times Y^{[s]}_i,\qquad X_j
        \times 0\times X^{[s]}_j,\qquad i=1,\ldots,\ell,\qquad j=1,\ldots,m.
    $$
    They form a modular basis for a VG Lie algebra  $V_E$ isomorphic to $V$ and related to  $X_E=X\times Y^{[k]}\times X^{[s]}$. Moreover,
    $$
        Y^{[k]}_i\times Y^{[s]}_i,\qquad  0\times X^{[s]}_j,\qquad i=1,\ldots,\ell,\qquad j=1,\ldots,m,
    $$
    form a modular basis on $M^{k+s}$.
    The map $\T \, \pr_{\hat{1}} \vert_{\cD^{V_{{E}}}}: \cD^{V_{{E}}} \to \T (M^{k+s})$ is  injective  on an open and dense subset of $M^{k+s+1}$, where $\pr_{\hat{1}}: M^{k+s+1} \to M^{k+s}$ denotes the projection onto the last $(k+s)$-copies. The conclusion then follows from Theorem~\ref{th:mixed_injective}.
\end{proof}
In particular, from this result, one easily recovers the following
\begin{cor}
    Every affine Lie system \eqref{eq:affine_LieSystem} admits a mixed superposition rule in terms of one particular solution of \eqref{eq:affine_LieSystem} and $n$ particular solutions of the associated linear homogeneous system \eqref{eq:affine:inhomogeneous}.
\end{cor}


\subsubsection{A time-dependent Calogero--Moser system subjected to an external force}
Let us now illustrate how Proposition~\ref{prop:mixed} can be applied to a system of physical interest, which serves as a representative example of the applicability of the preceding formalism.

Consider the system of ODEs on $\T^{*} \R^{2}_{q_{1}- q_{2} \neq 0}$, with coordinates $(q_{1}, q_{2}, p_{1}, p_{2})$, given by
\begin{equation}
    \begin{aligned}
         & \dv{q_{1}}{t} = p_{1}, \qquad                                                      &  & \dv{q_{2}}{t} = p_{2},                                                     \\[2pt]
         & \dv{p_{1}}{t} = - \Omega^{2}(t)q_{1} + \frac{c}{(q_{1}- q_{2})^{3}} - f(t), \qquad &  & \dv{p_{2}}{t} = - \Omega^{2}(t) q_{2} - \frac{c}{(q_{1}-q_{2})^{3}} -f(t),
    \end{aligned}
    \label{eq:CM:system}
\end{equation}
where $\Omega(t), f(t) \in C^{\infty}(\R)$ are arbitrary $t$-dependent functions, and $c \in \R$ is a constant. This system corresponds to the $t$-dependent vector field
$$
    X = X_{2} - \Omega^{2}(t) X_{3} - f(t) X_{5},
$$
where the vector fields
\begin{align}        X_{1} & := q_{1} \pdv{q_{1}} + q_{2} \pdv{q_{2}} - p_{1} \pdv{p_{1}} - p_{2} \pdv{p_{2}}, \qquad X_{2} := p_{1} \pdv{q_{1}} + p_{2} \pdv{q_{2}} + \frac{c}{(q_{1}- q_{2})^{3}} \left( \pdv{p_{1}} - \pdv{p_{2}} \right), \nonumber \\[2pt]
                     X_{3} & := q_{1} \pdv{p_{1}} + q_{2} \pdv{p_{2}},  \qquad X_{4}:= \pdv{q_{1}} + \pdv{q_{2}}, \qquad
                     X_{5} := \pdv{p_{1}} + \pdv{p_{2}},
                     \label{eq:CM:vf}
\end{align}
have non-vanishing commutation relations \eqref{eq:mixed:VG_cr}. Consequently, system \eqref{eq:CM:system} is a Lie system, with the vector fields above generating a VG Lie algebra $V$ isomorphic to the Schr{\"o}dinger algebra $\mathcal{S}(1)$. Moreover, the vector fields \eqref{eq:CM:vf}  are Hamiltonian vector fields relative to the canonical symplectic form  on $\T^{*}\R^{2}_{q_{1}-q_{2} \neq 0}$, given by
$$
    \omega := \dd q_{1} \wedge \dd p_{1} + \dd q_{2} \wedge \dd p_{2}.
$$
Associated Hamiltonian functions turn out to be
\begin{equation*}
    \begin{aligned}
         & h_{1} := q_{1}p_{1} + q_{2}p_{2}, \quad                      &  & h_{2}:= \frac{1}{2} \left( p_{1}^{2}+ p_{2}^{2} + \frac{c}{(q_{1}-q_{2})^{2}} \right), \\
         & h_{3}:= - \frac{1}{2} \big(q_{1}^{2} + q_{2}^{2}\big), \quad &  & h_{4}:= p_{1}+ p_{2}, \qquad h_{5}:= - q_{1}- q_{2},
    \end{aligned}
\end{equation*}
satisfying $\iota_{X_{i}} \omega = \dd h_{i}$, for $i = 1, \ldots, 5$.
Together with the constant Hamiltonian function $h_{0}:= 2$, these functions span a LH algebra $\cH_{\omega}$ with non-vanishing commutation relations \eqref{eq:mixed:LH_cr}, which is thus isomorphic to the centrally extended Schr{\"o}dinger algebra $\widehat{\mathcal{S}}(1)$. Hence, system \eqref{eq:CM:system} { is} the Hamilton's equations of the $t$-dependent Hamiltonian
$$
    h := h_{2} - \Omega^{2}(t)h_{3} - f(t) h_{5} = \frac{1}{2} \left( p_{1}^{2} + p_{2}^{2} + \frac{c}{(q_{1}-q_{2})^{2}} \right) + \frac{1}{2} \Omega^{2}(t) \big(q_{1}^{2} + q_{2}^{2}\big) + f(t) (q_{1}+ q_{2}).
$$
Observe that
$$
    h^{\mathrm{CM}} := h_{2} - \Omega^{2}(t) h_{3} = \frac{1}{2} \left( p_{1}^{2} + p_{2}^{2} + \frac{c}{(q_{1}-q_{2})^{2}} \right) + \frac{1}{2} \Omega^{2}(t) \big(q_{1}^{2} + q_{2}^{2}\big)
$$
is a {\it $t$-dependent two-body rational Calogero--Moser Hamiltonian} \cite{Calogero1969,Calogero1971,Moser1975}. It is worth noting that Calogero--Moser systems have been extensively investigated from various mathematical perspectives (see \cite{Etingof2007} and references therein). Furthermore, the term
$$F^{\mathrm{ext}} := -f(t) h_{5} = f(t)(q_{1}+ q_{2})$$
can be regarded as a {\it $t$-dependent external force} \cite{Godbillon1969,deLeon1989,deLeon2021},  since $\dd h_{5} = { -}\dd q_{1}  { -}\dd q_{2}$ is a semi-basic form on $\T^{*} \R^{2}_{q_{1}- q_{2} \neq 0}$; that is, it annihilates all vertical vector fields of the bundle $\T^{*} \R^{2}_{q_{1}-q_{2} \neq 0} \to \R^{2}_{q_{1}- q_{2} \neq 0}$. Therefore, the Hamiltonian
$$
    h = h^{\mathrm{CM}} + F^{\mathrm{ext}}
$$
describes a $t$-dependent two-body rational Calogero--Moser Hamiltonian subjected to a $t$-dependent external force.

Whenever $c \neq 0$, system \eqref{eq:CM:system} is not an affine Lie system. Nevertheless, since it {\it possesses} a VG Lie algebra $V \simeq \mathcal{S}(1) = \mathfrak{sl}(2, \R)  \overrightarrow{\oplus} \R^{2}$,  Proposition~\ref{prop:mixed} can be applied.
One has that $\langle X_{1}, X_{2}, X_{3} \rangle \simeq \mathfrak{sl}(2, \R)$ and $\langle X_{4}, X_{5} \rangle \simeq \R^{2}$ are modular Lie algebras of vector fields on  $\T^{*}\R^{2}_{q_{1} - q_{2} \neq 0}$. Therefore, system \eqref{eq:CM:system} admits a  mixed superposition rule involving one particular solution of $X$ and one particular solution of $Y$, where $Y = X_{2} - \Omega^{2}(t)X_{3}$ is the part of $X$ taking values in $\langle X_{1}, X_{2}, X_{3} \rangle$.


\section{Dirac--Lie systems} \label{Section:Dirac}
This section  recalls the main properties of the so-called {\it Dirac--Lie} systems introduced in \cite{Carinena2014}, and we examine direct products thereof. A salient feature of this class of Lie systems is that it encompasses several subclasses of Lie systems compatible with other geometric structures, such as  that of contact Lie systems of Liouville type (see Subsection~\ref{ex:Riccati_system}) and that of LH systems (cf.  system~\eqref{eq:mixed:inhom}), among others. Consequently, Dirac--Lie systems fit within a broad geometric framework for the study of Lie systems with compatible geometric structures. Let us begin by recalling the main aspects of Dirac manifolds \cite{Courant1990,Bursztyn2013,Carinena2014}. Recall that there are different conventions relative to the signs in the definitions of the structures used hereafter depending on the aims { or} purposes of each work \cite{Courant1990,Carinena2014}.

Let $\mathbb{T}M := \T M \oplus_{M} \T^{*}M$ be the  generalised tangent bundle of $M$. Consider the following two structures on $\mathbb{T}M$. First, the symmetric pairing defined at each $x \in M$ by
\begin{equation}
    \langle X_{x} + \alpha_{x}, Y_{x} + \beta_{x} \rangle_{+} := \frac{1}{2} (\alpha_{x}(Y_{x}) + \beta_{x}(X_{x})), \qquad X_{x}+ \alpha_{x}, Y_{x}+ \beta_{x} \in \mathbb{T}_{x}M = \T_{x}M \oplus \T_{x}^{*}M,
    \label{eq:Dirac:pairing}
\end{equation}
and second, the so-called {\it Dorfman} bracket $\llbracket \cdot, \cdot \rrbracket: \Gamma(\mathbb{T}M) \times \Gamma(\mathbb{T}M) \to \Gamma(\mathbb{T}M)$ defined for all $X + \alpha, Y + \beta \in \Gamma(\mathbb{T}M)$ by \cite{Dorfman1987}
$$
    \llbracket X+ \alpha, Y+ \beta \rrbracket := [X, Y] + \cL_{X} \beta - \iota_{Y} \dd \alpha.
$$
An {\it almost Dirac manifold} is a pair $(M, L)$, where $L \subset \mathbb{T}M$ is a maximally isotropic subbundle with respect to the { pairing} \eqref{eq:Dirac:pairing}. If, in addition, $L$ is involutive relative to the Dorfman bracket,  $(M, L)$ is a {\it Dirac manifold} \cite{Courant1990,Bursztyn2013,Carinena2014}.

Let $(M, L)$ be a Dirac manifold. A vector field $X$ on $M$ is   {\it $L$-Hamiltonian} if there exists an $f \in C^{\infty}(M)$ such that $X + \dd f \in \Gamma(L)$. In this case, $f$ is an {\it $L$-Hamiltonian function for $X$} and an {\it admissible function} of $(M, L)$. The space $\mathrm{Adm}(M,L)$ of admissible functions of $(M, L)$ is a Poisson algebra relative to the pointwise multiplication of functions and  the Poisson bracket defined by
$$
    \{f, g\}_{L} :=- X g,
$$
where $X$ is an $L$-Hamiltonian vector field for $f$.  Note that there are different conventions in the signs used in the Poisson bracket, the Hamiltonian vector fields and other related structures.

Particular examples of Dirac manifolds include presymplectic and Poisson manifolds. A {\it presymplectic manifold} is a pair $(M, \omega)$, where $\omega$ is a closed two-form on $M$, not necessarily of constant rank. The possible degeneracy of $\omega$ provides a natural generalisation of the symplectic case. The associated vector bundle morphism $
    \omega^{\flat} : \T M \ni v_{x} \mapsto \iota_{v_{x}}\omega_{x} \in \T^{*}M$
defines a Dirac structure via its graph
\[
    L^{\omega} := \mathrm{graph}\big(\omega^{\flat}\big) = \{ v_{x} + \iota_{v_{x}} \omega_{x}: v_{x} \in \T_{x}M, \; \forall x \in M \}.
\]

Similarly, a Poisson manifold $(M, \Lambda)$ determines a Dirac structure through the graph of the vector bundle morphism $\Lambda^{\sharp}: \T^{*} M \ni \alpha_{x} \mapsto -\iota_{\alpha_{x}}\Lambda_{x} \in \T M$,
given explicitly by
\[
    L^{\Lambda} := \mathrm{graph}\big(\Lambda^{\sharp}\big) = \{ \alpha_{x} - \iota_{\alpha_{x}} \Lambda_{x} : \alpha_{x} \in \T_{x}^{*}M , \; \forall x \in M\}.
\]
The minus sign is given to recover the geometric mechanics relation between Hamiltonian vector fields  and the differential of their Hamiltonian functions $-\iota_{\dd f}\Lambda=X_f$.

A {\it Jacobi manifold} is a triple $(M, \Lambda, E)$, where $\Lambda \in \mathfrak{X} ^{2}(M)$ is a bivector field on $M$ and $E \in \mathfrak{X}(M)$ is a vector field on $M$, hereafter referred to as the {\it Reeb vector field} of the Jacobi manifold $(M, \Lambda, E)$,  such that
$$
    [\Lambda, \Lambda] = 2 E \wedge \Lambda, \qquad [E, \Lambda] = 0
$$
with respect to the Schouten--Nijenhuis bracket $[\cdot, \cdot]$. Every Poisson manifold $(M, \Lambda)$ is a Jacobi manifold $(M, \Lambda, E = 0)$. Note that the definition of the Schouten--Nijenhuis bracket is the classical one used, for instance, in \cite{Schouten1953, Nijenhuis1955,Marle1997}. There exists a modern definition of the Schouten--Nijenhuis bracket, which differs from ours on a global proportional sign depending on the degree of $\Lambda$ (see \cite[Example 2.20]{Grabowski2013a} and references therein).

The class of Jacobi manifolds encompasses a relevant type of manifolds: the so-called {\it co-orientable contact manifolds}. A {\it (co-orientable) contact manifold} is a pair $(M, \eta)$, where $M$ is a $(2n+1)$-dimensional manifold and $\eta \in \Omega^{1}(M)$ is a {\it contact form} on $M$; that is,  a one-form $\eta$ on $M$ so that $\eta \wedge (\dd \eta)^{n}$ is a volume form on $M$. A contact manifold $(M, \eta)$ determines a unique vector field $\cR \in \mathfrak{X}(M)$, referred to as the {\it Reeb vector field} of $(M, \eta)$, characterised by the conditions
\begin{equation}
    \iota_{\cR} \eta = 1, \qquad \iota_{\cR} \dd \eta = 0.
    \label{eq:Reeb}
\end{equation}
The contact form $\eta$ induces a vector bundle isomorphism
$$\flat: \T M \ni v_{x} \mapsto \iota_{v_{x}} (\dd \eta)_{x} + (\iota_{v_{x}} \eta_{x})\eta_{x} \in \T^{*}M,$$
and denote by $\sharp = \flat^{-1}$ its inverse. One may then define the bivector field { $\Lambda \in \mathfrak{X}^{2}(M)$} such that $\Lambda^{\sharp}(\alpha_{x}) = \sharp \alpha_{x} - (\iota_{\cR_{x}} \alpha_{x}) \cR_{x}$ and the vector field $E:= - \cR$. In this way, $(M, \Lambda, E)$ is a Jacobi manifold.

Assuming certain regularity conditions on the Reeb vector field $E$ of a Jacobi manifold $(M, \Lambda, E)$, we can also obtain a Dirac structure in a natural way \cite{Courant1990,Liu2000}. Let $(M, \Lambda, E)$ be a Jacobi manifold such that $E \neq 0$. In such case $\langle E \rangle$ defines a vector subbundle of $\T M$.  Denote by { $\langle E \rangle^{\circ} \subset \mathrm{T}^{*}M$} { its} annihilator. Then
$$
    L^{\Lambda, E} := \langle E \rangle \oplus_{M} \mathrm{graph}\big(\Lambda^{\sharp} \vert_{\langle E \rangle^{\circ}} \big)
$$
is a Dirac structure on $M$. The space of {\it good Hamiltonian functions} of $(M, \Lambda, E)$ is \cite{Herranz2015}
$$
    \Adm\big(M, L^{\Lambda, E}\big) = \set{f \in C^{\infty}(M): Ef = 0}.
$$
Consider now $f \in \Adm\big(M, L^{\Lambda, E}\big)$, and
$$
    -\overline{f}E - \Lambda^{\sharp}(\dd f) + \dd f \in \Gamma\big(L^{\Lambda, E}\big)
$$
for every $\overline{f} \in C^{\infty}(M)$. In particular,
$$
    X_{f} := -f E - \Lambda^{\sharp}(\dd f)
$$
is the Hamiltonian vector field associated to $f$ with respect to the Jacobi manifold $(M, \Lambda, E)$. Moreover, if $f, g\in \Adm\big(M, L^{\Lambda, E}\big)$ are two admissible functions, their Poisson bracket is
\begin{equation*}
    \set{f,g}_{L^{\Lambda,E}} =- X_{f}g = \Lambda(\dd f, \dd g).
\end{equation*}
Identities \eqref{eq:Reeb} demonstrate that every contact manifold $(M, \eta)$ induces a Jacobi manifold $(M, \Lambda, E)$ with $E \neq 0$, and hence, from the previous construction, a Dirac structure $L^{\Lambda, E}$ on $M$. Note that the two-form $\omega := \dd \eta$ is an {\it exact presymplectic form} on $M$, and consequently the contact manifold $(M, \eta)$ also gives rise to a Dirac structure $L^{\omega}$ on $M$. The following result shows that both Dirac structures are identical.
\begin{prop}
    Let $(M, \eta)$ be a contact manifold, and let $L^{\omega}$ and $L^{\Lambda, E}$ be the Dirac structures on $M$ induced by the presymplectic form $\omega = \dd \eta$ and the Jacobi structure $(\Lambda, E)$, respectively. Then, $L^{\omega} = L^{\Lambda, E}$.
\end{prop}
\begin{proof}
    Let $v_{x} + \iota_{v_{x}} \omega_{x} \in L_{x}^{\omega}$.  Since $\Lambda_{x}^{\sharp}(\alpha_{x}) \in \ker \eta_{x}$ and  $\alpha_{x}:= \iota_{v_{x}} \omega_{x} \in \langle E_{x} \rangle^{\circ}$, it follows that
    $v_{x} = (\iota_{v_{x}} \eta_{x}) E_{x} + \Lambda^{\sharp}_{x}(\alpha_{x})$. Then,
    $$ v_{x} + \iota_{v_{x}} \omega_{x} = (\iota_{v_{x}}\eta_{x}) E_{x} + \alpha_{x}  + \Lambda_{x}^{\sharp}(\alpha_{x}) \in L^{\Lambda,E}_{x}.$$
    Conversely, let $a E_{x} + \alpha_{x} + \Lambda^{\sharp}_{x}(\alpha_{x}) \in L^{\Lambda, E}$, where $a \in \R$, and consider $v_{x} := a E_{x} + \Lambda_{x}^{\sharp}(\alpha_{x})$. Since $\iota_{E_{x}} \alpha_{x} = 0$, we have that $\Lambda^{\sharp}_{x}(\alpha_{x}) = \sharp \alpha_{x}$. Thus, $ \iota_{\sharp \alpha_{x}} \eta_{x} =  \iota_{\Lambda^{\sharp}_{x}(\alpha_{x})} \eta_{x}  =  0$, yielding $\iota_{\sharp \alpha_{x}} \omega_{x} = \alpha_{x}$. Then, $\iota_{v_{x}} \omega_{x} = \alpha_{x}$ and, consequently,
    $$
        a E_{x} + \alpha_{x} + \Lambda_{x}^{\sharp}(\alpha_{x})  = v_{x} + \iota_{v_{x}} \omega_{x} \in L_{x}^{\omega}.
    $$
\end{proof}

The relationships among these various geometric structures, relative to the standard inclusions discussed above, are summarised in the following diagram.


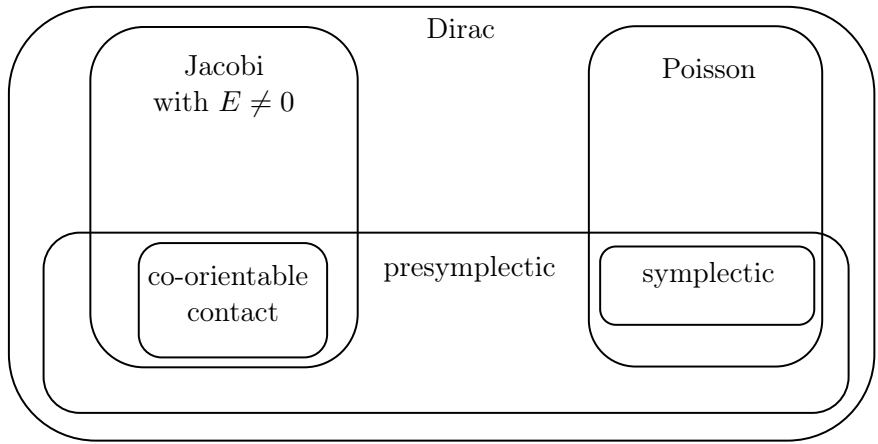
\begin{center}

    \tikzset{every picture/.style={line width=0.75pt}} 

    \begin{tikzpicture}[x=0.75pt,y=0.75pt,yscale=-1,xscale=1]

        \draw   (104,48.84) .. controls (104,24.66) and (123.61,5.05) .. (147.79,5.05) -- (496.71,5.05) .. controls (520.89,5.05) and (540.5,24.66) .. (540.5,48.84) -- (540.5,180.21) .. controls (540.5,204.39) and (520.89,224) .. (496.71,224) -- (147.79,224) .. controls (123.61,224) and (104,204.39) .. (104,180.21) -- cycle ;
        \draw   (121.5,137.23) .. controls (121.5,127.18) and (129.65,119.03) .. (139.69,119.03) -- (509.31,119.03) .. controls (519.35,119.03) and (527.5,127.18) .. (527.5,137.23) -- (527.5,191.81) .. controls (527.5,201.85) and (519.35,210) .. (509.31,210) -- (139.69,210) .. controls (129.65,210) and (121.5,201.85) .. (121.5,191.81) -- cycle ;
        \draw   (169.5,135.83) .. controls (169.5,129.45) and (174.67,124.29) .. (181.04,124.29) -- (252.96,124.29) .. controls (259.33,124.29) and (264.5,129.45) .. (264.5,135.83) -- (264.5,170.46) .. controls (264.5,176.83) and (259.33,182) .. (252.96,182) -- (181.04,182) .. controls (174.67,182) and (169.5,176.83) .. (169.5,170.46) -- cycle ;
        \draw   (402.2,133.93) .. controls (402.2,129.56) and (405.74,126.02) .. (410.11,126.02) -- (502.29,126.02) .. controls (506.65,126.02) and (510.19,129.56) .. (510.19,133.93) -- (510.19,157.65) .. controls (510.19,162.01) and (506.65,165.55) .. (502.29,165.55) -- (410.11,165.55) .. controls (405.74,165.55) and (402.2,162.01) .. (402.2,157.65) -- cycle ;
        \draw   (145.09,42.18) .. controls (145.09,27.29) and (157.16,15.22) .. (172.06,15.22) -- (252.96,15.22) .. controls (267.85,15.22) and (279.92,27.29) .. (279.92,42.18) -- (279.92,160.03) .. controls (279.92,174.93) and (267.85,187) .. (252.96,187) -- (172.06,187) .. controls (157.16,187) and (145.09,174.93) .. (145.09,160.03) -- cycle ;
        \draw   (396.5,38.37) .. controls (396.5,25.37) and (407.04,14.83) .. (420.04,14.83) -- (490.65,14.83) .. controls (503.65,14.83) and (514.19,25.37) .. (514.19,38.37) -- (514.19,163.08) .. controls (514.19,176.08) and (503.65,186.61) .. (490.65,186.61) -- (420.04,186.61) .. controls (407.04,186.61) and (396.5,176.08) .. (396.5,163.08) -- cycle ;

        \draw (313.19,9.45) node [anchor=north west][inner sep=0.75pt]   [align=left] {Dirac};
        \draw (291.45,129.59) node [anchor=north west][inner sep=0.75pt]   [align=left] {presymplectic};
        \draw (172.5,133.83) node [anchor=north west][inner sep=0.75pt]   [align=left] {\begin{minipage}[lt]{64.53pt}\setlength\topsep{0pt}
                co-orientable
                \begin{center}
                    contact
                \end{center}

            \end{minipage}};
        \draw (431.86,29.31) node [anchor=north west][inner sep=0.75pt]   [align=left] {Poisson};
        \draw (422.06,132.24) node [anchor=north west][inner sep=0.75pt]   [align=left] {symplectic};
        \draw (161.12,27.89) node [anchor=north west][inner sep=0.75pt]   [align=left] {\begin{minipage}[lt]{74.76pt}\setlength\topsep{0pt}
                \begin{center}
                    Jacobi\\
                    with $E \neq 0$
                \end{center}

            \end{minipage}};

    \end{tikzpicture}
\end{center}

Let us now recall the basic aspects of {\it Dirac--Lie systems} \cite{Carinena2014}.

A {\it Dirac--Lie system} is a triple $(M, L, X)$, where $(M, L)$ is a Dirac manifold and $X$ is a Lie system on $M$ which admits a VG Lie algebra of $L$-Hamiltonian vector fields. From the previous discussion it follows that the Riccati system   \eqref{eq:ex:Riccati_system} and the inhomogeneous linear system \eqref{eq:mixed:hom} studied before are Dirac--Lie systems. Indeed, every Lie system $X$ is related to a Dirac--Lie system $(M,\T M,X)$, although this is rather useless, as $L$-Hamiltonian functions are the key structures to study Dirac--Lie systems, and every vector field is $\T M$-Hamiltonian with zero $\T M$-Hamiltonian function.

\begin{remark}
    Let $\big(M, L^{\Lambda, E}\big)$ be the Dirac manifold induced by a Jacobi manifold $(M, \Lambda, E)$ with $E \neq 0$. Then, a Dirac--Lie system $\big(M, L^{\Lambda, E}, X\big)$ is a particular case of a Jacobi--Lie system \cite{Herranz2015}. If  $(M, \Lambda, E)$ is the Jacobi manifold induced by a co-orientable contact manifold $(M, \eta)$, then Dirac--Lie systems $\big(M, L^{\Lambda, E}, X\big)$ coincide with contact Lie systems of Liouville type relative to the contact form $\eta$. Similarly, if $\big(M, L^{\Lambda}\big)$ is the Dirac manifold arising from a Poisson bivector $\Lambda$,  then the Dirac--Lie systems $\big(M, L^{\Lambda}, X\big)$ reduce to LH systems with respect to $\Lambda$.
\end{remark}

It was proven in \cite[Theorem 6.4 ]{Carinena2014} that, for every finite-dimensional Lie algebra of $L$-Hamiltonian vector fields, there exists a finite-dimensional Lie algebra $\mathfrak{W} \subset \big(\Adm(M, L), \set{\cdot, \cdot}_{L}\big)$ consisting { of} $L$-Hamiltonian functions for the elements of $V$. Nevertheless, as also noted therein, the converse does not hold in general; that is, the fact that a Lie algebra of $L$-admissible functions is finite-dimensional does not imply that a corresponding Lie algebra of $L$-Hamiltonian  vector fields is finite-dimensional.


\subsection{Direct products of Dirac--Lie systems}

Let us now briefly show that the results proven in \cite{Carinena2014} for diagonal prolongations of Dirac structures hold analogously for direct products.

Recall that a {\it forward Dirac map} between two Dirac manifolds $(M, L_{M})$ and $(N, L_{N})$ is a  map $\Psi: M \to N$ such that
$$
    (L_{N})_{\Psi(x)} = (\T_{x}\Psi)_{!}(L_{M})_{x} := \set{\T_{x} \Psi(v_{x}) + \alpha_{x}\in  \mathbb{T}_{\Psi(x)}N: v_{x} + \T_{x}^{*}\Psi (\alpha_{x}) \in (L_{M})_{x}}
$$
for all $x \in M$.

\begin{prop} \label{prop:Dirac_prod}
    Let $(M_{i}, L_{i})$ be Dirac manifolds for $i = 1, \ldots, k$. Then, the direct product $L_{1} \times \cdots \times L_{k}$ is a Dirac structure on $M_{1} \times \cdots \times M_{k}$. Moreover, the forward image of $L_{1} \times \cdots \times L_{k}$ through $\pr_{i}: M_{1} \times \cdots \times M_{k} \to M_{i}$ equals $L_{i}$, with $i = 1, \ldots, k$.
\end{prop}

\begin{proof}[Sketch of the proof]
    The first statement was proved in \cite[Proposition 2]{Jacobs2014}. For the second, let $n_{i} := \dim M_{i}$ for $i = 1, \ldots, k$, and let $p = \big(x_{(1)}, \ldots, x_{(k)}\big) \in M_{1} \times \cdots \times M_{k}$ be any point. Consider a basis of local sections $\big\{X_{(i)}^{j} + \alpha_{(i)}^{j}: j = 1, \ldots, n_{i} \big\}$ of $L_{i}$ around $x_{(i)} \in M_{i}$. Then,   $\big\{\widetilde{X}_{(i)}^{j} + \widetilde{\alpha}_{(i)}^{j}: i = 1, \ldots, k, j = 1, \ldots, n_{i}\big\}$ is a basis of local sections of $L_{1} \times \cdots \times L_{k}$ around $p$. The result follows by adopting the same ansatz used in \cite[Proposition 7.3]{Carinena2014}.
\end{proof}

The canonical identifications \eqref{eq:can_iden} of the tangent (resp.~cotangent) bundle of $M_{1} \times \cdots \times M_{k}$ as the direct product of tangent (resp.~cotangent) bundles of $M_{i}$ gives rise to the canonical identification $\mathbb{T}(M_{1} \times \cdots \times M_{k}) = \mathbb{T}M_{1} \times \cdots  \times \mathbb{T}M_{k}$. Let $\rho_{i}$ and $\rho_{i}^{*}$ denote the canonical projections $\rho_{i}: \mathbb{T}M_{i} \to \T M_{i}$ and $\rho_{i}^{*}: \mathbb{T}M_{i} \to \T^{*}M_{i}$, for $i =1, \ldots, k$. Then,
\begin{equation*}
    \begin{split}
         & \rho_{1} \times \cdots \times \rho_{k}: \mathbb{T}(M_{1} \times \cdots \times M_{k}) \to \T (M_{1} \times \cdots \times M_{k}),            \\
         & \rho_{1}^{*} \times \cdots \times \rho_{k}^{*}: \mathbb{T}(M_{1} \times \cdots \times M_{k}) \to \T^{*} (M_{1} \times \cdots \times M_{k})
    \end{split}
\end{equation*}
are the canonical projections of $\mathbb{T}(M_{1} \times \cdots \times M_{k})$ onto $\T(M_{1} \times \cdots \times M_{k})$ and $\T^{*}(M_{1} \times \cdots \times M_{k})$, respectively.

The following results are immediate consequences of Proposition~\ref{prop:Dirac_prod}.

\begin{cor}
    Let $(M_{i}, L_{i})$ be Dirac manifolds for $i = 1, \ldots, k$. Then,
    $$[\rho_{1} \times \cdots \times \rho_{k}](L_{1} \times \cdots \times L_{k}) = \rho_{1}(L_{1}) \times \cdots  \times \rho_{k}(L_{k}).$$
    Consequently, if the $X_{(i)}$ are $L_{i}$-Hamiltonian vector fields on $M_{i}$, their direct product $X_{(1)} \times \cdots \times X_{(k)}$ is an $(L_{1} \times \cdots \times L_{k})$-Hamiltonian vector field. Moreover,
    $$
        [\rho_{1}^{*} \times \cdots \times \rho_{k}^{*}](L_{1} \times \cdots \times L_{k}) = \rho_{1}^{*}(L_{1}) \times \cdots \times \rho_{k}^{*}(L_{k}).
    $$
\end{cor}
\begin{cor}
    Let $\big(M_{i}, L_{i}, X_{(i)}\big)$ be Dirac--Lie systems for $i = 1, \ldots, k$. Then,
    $$\big(M_{1} \times \cdots \times M_{k}, L_{1} \times \cdots \times L_{k}, X_{(1)} \times \cdots \times X_{(k)}\big)
    $$ is a Dirac--Lie system.
\end{cor}

Recall from Subsection~\ref{subsection:direct_products} that, given functions $f_{(i)} \in C^{\infty}(M_{i})$, with $i = 1, \ldots, k$, their direct product is the function  $\lambda\big(f_{(1)}, \ldots, f_{(k)}\big)  \in C^{\infty}(M_{1} \times \cdots \times M_{k})$ given by
$$
    \lambda\big(f_{(1)}, \ldots, f_{(k)}\big)= \pr_{1}^{*}f_{(1)} + \cdots + \pr_{k}^{*}f_{(k)}.
$$
\begin{prop} \label{prop:adm_Dirac} Let $(M_{i}, L_{i})$ be Dirac manifolds, let $X_{(i)}$ be a vector field, and let $f_{(i)}$ be a function on $M_{i}$ for $i = 1, \ldots, k$. Then:
    \begin{enumerate}[label=(\arabic*), font=\normalfont]
        \item If $f_{(i)}$ is an $L_{i}$-Hamiltonian function for $X_{(i)}$,   the direct product $\lambda\big(f_{(1)}, \ldots, f_{(k)}\big)$ is an $(L_{1} \times \cdots \times L_{k})$-Hamiltonian function for the direct product $X_{(1)} \times \cdots \times X_{(k)}$; and
        \item The map
              $$ \big(\mathrm{Adm}(M_{i}, L_{i}), \set{\cdot, \cdot}_{L_{i}}\! \big) \ni f_{(i)} \mapsto \pr_{i}^{*} f_{(i)} \in \big(\mathrm{Adm}(M_{1} \times \cdots \times M_{k} ), \set{\cdot, \cdot}_{L_{1} \times \cdots \times L_{k}} \!\big)$$
              is an injective morphism of Poisson algebras. Consequently, the map
              \begin{equation}
                  \big(f_{(1)}, \ldots, f_{(k)}\big) \mapsto \lambda\big(f_{(1)}, \ldots f_{(k)}\big),
                  \label{eq:prol_morph}
              \end{equation}
              mapping admissible functions $f_{(i)} \in \Adm(M_{i}, L_{i})$ to their direct product is a morphism of Lie algebras.
    \end{enumerate}
\end{prop}
\begin{proof}
    Let us prove both points.
    \begin{itemize}
        \item[(1)] Since $f_{(i)}$ is an $L_{i}$-Hamiltonian function for $X_{(i)}$, then $X_{(i)} + \dd f_{(i)} \in \Gamma(L_{i})$. Hence,
              $$
                  \big(X_{(1)} \times \cdots \times X_{(k)} \big) + \dd \big(f_{(1)}\times \ldots\times f_{(k)} \big) = \big(X_{(1)} + \dd f_{(1)}\big) \times \cdots \times \big(X_{(k)} + \dd f_{(k)}\big) \in \Gamma(L_{1}\times \cdots  \times L_{k}),
              $$
              as required.
        \item[(2)] Let $f_{(i)}, g_{(i)} \in \mathrm{Adm}(M_{i}, L_{i})$ for $i = 1, \ldots, k$. Since
              $$
                  \big(\mathrm{Adm}(M_{i}, L_{i}), \set{\cdot, \cdot}_{L_{i}}\! \big) \ni f_{(i)} \mapsto \pr_{i}^{*} f_{(i)} \in \big(\mathrm{Adm}(M_{1} \times \cdots \times M_{k}), \set{\cdot, \cdot}_{L_{1} \times \cdots \times L_{k}}\!\big)
              $$
              is the restriction of $\pr_{i}^{*}$ to $\Adm(M_{i}, L_{i})$, it is clear that { it} is an injective morphism of $\R$-algebras. Moreover, it is also a morphism of Poisson algebras, since
              \begin{equation*}
                  \begin{split}
                      \pr_{i}^{*}\big(\{f_{(i)}, g_{(i)}\}_{L_{i}}\big) & = -\pr_{i}^{*}\big(X_{f_{(i)}} g_{(i)}\big) = -\widetilde{X}_{f_{(i)}}\big(\pr_{i}^{*}g_{(i)}\big)                                                          \\
                                                                        & =- X_{\pr_{i}^{*}f_{(i)}}\big(\pr_{i}^{*}g_{(i)}\big) = \set{\pr_{i}^{*}\big(f_{(i)}\big), \pr_{i}^{*}\big(g_{(i)}\big)}_{L_{1} \times \cdots \times L_{k}}
                  \end{split}
              \end{equation*}
              for every $i = 1, \ldots, k$. Finally, observe that  $\lambda\big(f_{(1)}, \ldots, f_{(k)}\big)= \pr_{1}^{*}f_{(1)} + \cdots + \pr_{k}^{*}f_{(k)}$.
    \end{itemize}
\end{proof}

In general, the map \eqref{eq:prol_morph} is not a morphism of $\R$-algebras, and hence neither a morphism of Poisson algebras. For example, consider the Hamiltonian function $h_{1} = xy$ on $\R^{2}$ from \eqref{eq:mixed:Ham}. Then, its diagonal prolongation  to $(\R^{2})^{2}$ is $x_{(1)}y_{(1)} + x_{(2)}y_{(2)} \neq \big(x_{(1)}+ x_{(2)})(y_{(1)} + y_{(2)}\big) = \lambda(x, x) \lambda(y, y)$.


\section{Coalgebra formalism for Dirac--Lie systems}\label{Section:coalgebras}
Let us develop a novel coalgebra-based framework to study Dirac--Lie systems, generalising the ideas introduced in \cite{Carinena2014} and geometrically improving the method proposed in \cite{Ballesteros2013} for LH systems.


\subsection{On the underlying algebraic structures}
The coalgebra method for LH systems \cite{Ballesteros2013} relies on the Poisson coalgebra structure of the symmetric algebra $ S(\g) $ of a finite-dimensional real Lie algebra $\g$. A subsequent formulation in \cite{Carinena2014} replaced $S(\g)$ by $C^{\infty}(\g^{*})$, where $\g^{*}$ is equipped with the KKS bracket, albeit without addressing its underlying coalgebra structure. This subsection presents a careful treatment of the coalgebra structure of $C^{\infty}(\g^{*})$, which requires recalling the Fr\'echet topology on spaces of smooth functions and some properties of  topological tensor products. Particular attention is devoted to the case in which $ \g $ is a semidirect sum, since this naturally leads to a topological formulation of Poisson comodule algebras~\cite{Ballesteros2002}.


\subsubsection{The Fr\'echet topology on the space of smooth functions}
Let us describe the standard nuclear and Fr\'echet topology on the space $C^{\infty}(M)$ of smooth functions of $M$. Denote by $J^{k}(M, \mathbb{R})$ the space of $k$-jets on $M$, with $0 \leq k \leq \infty$. { Write $C(J^{k}(M, \R))$ for the space of continuous sections of the bundle $J^{k}(M, \R) \to M$.} Equip $C\big(J^{k}(M, \R)\big)$ with the {\it $C^{\infty}$-compact open topology}; that is, the topology generated by the subbasis consisting of sets of the form
$$\big\{ { \sigma}  \in C\big(J^{k}(M, \R)\big): { \sigma}(K) \subset W\big\},  $$
where $K$ ranges over all compact subsets of ${ M}$ and $W$ ranges over all open subsets of ${ J^{k}(M, \R)}$. We now consider the initial topology on $C^{\infty}(M)$ induced by the $k$-jet extensions
\[ j^{k}: C^{\infty}(M) \to C\big(J^{k}(M, \R)\big).\]
With this topology, $C^{\infty}(M)$ is a nuclear and Fr\'echet space \cite[p. 67]{Kriegl1997}.

Given two manifolds $M$ and $N$, one can endow the algebraic tensor product $C^{\infty}(M) \otimes C^{\infty}(N)$ with two main topologies (see \cite[Chapter 43]{Treves1967}).
\begin{itemize}
    \item The {\it projective tensor product topology} (also known as the $\pi$-topology) is the strongest locally convex topology making the canonical bilinear map $C^{\infty}(M) \times C^{\infty}(N) \to C^{\infty}(M) \otimes C^{\infty}(N)$ continuous.
    \item  The {\it injective tensor product topology} (also referred to as the $\varepsilon$-topology) is the finest locally convex topology making the canonical bilinear map $C^{\infty}(M) \times C^{\infty}(N) \to C^{\infty}(M) \otimes C^{\infty}(N)$ separately continuous.
\end{itemize}
We denote by $C^{\infty}(M) \,\widehat{\otimes}_{\pi}\, C^{\infty}(N)$ and by $C^{\infty}(M) \, \widehat{\otimes}_{\varepsilon} \, C^{\infty}(N)$ the completions of $C^{\infty}(M) \otimes C^{\infty}(N)$ with respect to the projective and the injective tensor product topologies, respectively. As both $C^{\infty}(M)$ and $C^{\infty}(N)$ are nuclear and  Fr\'echet spaces, it follows that $$C^{\infty}(M) \, \widehat{\otimes}_{\pi} \, C^{\infty}(N) = C^{\infty}(M) \, \widehat{\otimes}_{\varepsilon} \, C^{\infty}(N).$$
Accordingly, we shall omit the indices $\pi$ and $\varepsilon$, and write $\widehat{\otimes}$ for both $\widehat{\otimes}_{\pi}$ and $\widehat{\otimes}_{\varepsilon}$. A noteworthy result, originally due to Grothendieck \cite[pp. 104--106]{Grothendieck1952}, establishes that
$$
    C^{\infty}(M) \, \widehat{\otimes} \, C^{\infty}(N) = C^{\infty}(M \times N).
$$
Within this context, from now on we identify $C^{\infty}(M) \otimes C^{\infty}(N)$  with the subspace of $C^{\infty}(M \times N)$ consisting of  finite linear combinations of elements of the form $(\pr_{1}^{*}f)(\pr_{2}^{*}g)$, with $f \in C^{\infty}(M), g \in C^{\infty}(N)$, and $\pr_{1}: M \times N \to M$ and $\pr_{2}: M \times N \to N$ the canonical projections.

Consider a map $\Psi: M \to N$. Its pull-back $\Psi^{*}: C^{\infty}(N) \to C^{\infty}(M)$ is continuous with respect to the Fr\'echet topologies. Given  maps $\Psi_{i}: M_{i} \to N_{i}$, for $i = 1, 2$, the  tensor product of their pull-backs is a continuous linear map
$\Psi_{1}^{*} \otimes \Psi_{2}^{*}: C^{\infty}(N_{1}) \otimes C^{\infty}(N_{2}) \to C^{\infty}(M_{1}) \otimes C^{\infty}(M_{2})$, where both tensor products are endowed with either the projective or the injective tensor product topology \cite[Proposition 43.6]{Treves1967}. The map $\Psi_{1}^{*} \otimes \Psi_{2}^{*}$ can be uniquely extended to a continuous  linear map $\Psi_{1}^{*} \, \widehat{\otimes} \, \Psi_{2}^{*}: C^{\infty}(M_{1} \times M_{2}) \to C^{\infty}(N_{1} \times N_{2})$. The pull-back of the product map $\Psi_{1} \times \Psi_{2}: M_{1} \times M_{2} \to N_{1} \times N_{2}$ is a continuous linear map $(\Psi_{1} \times \Psi_{2})^{*} : C^{\infty}(N_{1} \times N_{2}) \to C^{\infty}(M_{1} \times M_{2})$ which coincides with $\Psi_{1}^{*} \times \Psi_{2}^{*}$ on $C^{\infty}(N_{1}) \otimes C^{\infty}(N_{2})$. Hence,
$$\Psi_{1}^{*} \, \widehat{\otimes} \, \Psi_{2}^{*} = (\Psi_{1} \times \Psi_{2})^{*}.$$


\subsubsection{Well-behaved topological Poisson coalgebras and comodules} \label{subsection:well_Poisson}
Let us now examine the Poisson coalgebra structure of $C^{\infty}(\g^{*})$, where $\g$ is a finite-dimensional real Lie algebra.

First of all, recall that the KKS bracket on $\g^{*}$ is given by
$$
    \set{f, g}(\xi):= \langle \xi, [\dd f_{\xi}, \dd g_{\xi}] \rangle, \qquad f, g \in C^{\infty}(\g^{*}), \qquad \xi \in \g^{*},
$$
where $\dd f_{\xi}, \dd g_{\xi}: \T_{\xi} \g^{*} \to \R$ are naturally identified as elements of $\g \simeq (\T_{\xi} \g^{*})^{*}$ and $\langle \cdot, \cdot \rangle$ denotes the canonical pairing between $\g^{*}$ and $\g$. Then,
$$
    m: \g^{*} \times \g^{*} \ni (\xi_{1}, \xi_{2}) \mapsto \xi_{1} + \xi_{2} \in \g^{*}
$$
is a Poisson map, where $\g^{*} \times \g^{*}$ is equipped with the product Poisson manifold structure.
Hence, the pull-back of $m$ is a morphism of Poisson algebras
$$
    \Delta: = m^{*}: C^{\infty}(\g^{*}) \to C^{\infty}(\g^{*}) \, \widehat{\otimes}\, C^{\infty}(\g^{*}) = C^{\infty}(\g^{*} \times \g^{*}).
$$
From the associativity of $m$, it follows that
\[\begin{tikzcd}
        {C^{\infty}(\mathfrak{g}^{*}) } && {C^{\infty}(\g^{*}) \, \widehat{\otimes}\, C^{\infty}(\g^{*})} \\
        {C^{\infty}(\g^{*}) \, \widehat{\otimes}\, C^{\infty}(\g^{*})} && {C^{\infty}(\g^{*}) \, \widehat{\otimes}\, C^{\infty}(\g^{*}) \, \widehat{\otimes}\, C^{\infty}(\g^{*})}
        \arrow["\Delta", from=1-1, to=1-3]
        \arrow["\Delta"', from=1-1, to=2-1]
        \arrow["{\Delta \, \widehat{\otimes} \,\mathrm{id}}", from=1-3, to=2-3]
        \arrow["{\mathrm{id} \, \widehat{\otimes} \,\Delta}"', from=2-1, to=2-3]
    \end{tikzcd}\]
is a commutative diagram of morphisms of Poisson algebras.

The latter entails that the triple $(C^{\infty}(\mathfrak{g}^*), \{\cdot, \cdot\}, \Delta)$ constitutes a {\it well-behaved topological Poisson coalgebra} \cite{Bonneau1994, Bidegain1996}. This refers to the standard definition of a Poisson coalgebra $(\mathcal{A}, \set{\cdot, \cdot}_{\mathcal{A}}, \Delta_{\mathcal{A}})$ given in Subsection~\ref{subsection:coalgebras} with the distinction that the underlying vector space $\mathcal{A}$ is now  a nuclear, and Fr\'echet or dual of a Fr\'echet space,  the structure maps are considered to be continuous in the Fr\'echet topology, and the algebraic tensor product $\otimes$ is replaced by the topological tensor product $\widehat{\otimes}$. The principal advantage of this formulation lies in its enhanced geometric character: the coproduct arises from a Poisson map $m: \g^{*} \times \g^{*} \to \g^{*}$, and the Poisson structure on $C^{\infty}(\mathfrak{g}^{*} \times \mathfrak{g}^{*})$ is given by the product Poisson structure on $\mathfrak{g}^{*} \times \mathfrak{g}^{*}$.

The  {\it $k^{th}$-order coproduct map} $\Delta^{[k]} : C^{\infty}(\g^{*}) \to C^{\infty}\big((\g^{*})^{k}\big)$,
recursively defined by
$$
    \Delta^{[2]} := \Delta, \qquad \Delta^{[k]}: = \big(\mathrm{id}^{\widehat{\otimes}(k-2)} \, \widehat{\otimes} \,  \Delta^{[2]}\big) \circ \Delta^{[k-1]} \qquad k > 2,
$$
where $\mathrm{id}^{\widehat{\otimes}(k-2)} := \mathrm{id} \, \widehat{\otimes} \, \overset{(k-2)}{\cdots} \, \widehat{\otimes} \, \mathrm{id}$,
is a morphism of Poisson algebras.  Equivalently, $\Delta^{[k]}$ is the pull-back of the Poisson map   $m^{[k]}: (\g^{*})^{k} \to \g^{*}$ recursively defined by
$$
    m^{[2]}:= m, \qquad m^{[k]} :=  m^{[k-1]} \circ \big(\mathrm{id}^{\times (k-2)} \times m^{[2]}\big) \qquad k > 2,
$$
where $\mathrm{id}^{\times (k-2)}:= \mathrm{id} \times \overset{(k-2)}{\cdots} \times \mathrm{id}$.

Let $\h$ be a Lie subalgebra of $\g$, and consider $\p \subset \g$ a vector subspace such that $\g = \h \oplus \p$ as a direct sum of vector spaces. Then, the annihilator $\p^{\circ} \subset \g^{*}$ of $\p$ is a Poisson submanifold of $\g^{*}$ if, and only if, $\p$ is an ideal of $\g$. In such case, $\g$ is a semidirect sum $\g = \h \overrightarrow{\oplus} \p$, and $\p^{\circ}$ and $\h^{*}$ are canonically isomorphic Poisson manifolds \cite[Example 8.16]{Crainic2021}.

Suppose now that $\g$ is a semidirect sum $\g = \h \overrightarrow{\oplus} \p$, and consider the Poisson submanifold $\h^{*} \subset \g^{*}$. Then,
$$\overline{m}: \g^{*} \times \h^{*}\ni (\xi_{1}, \xi_{2}) \mapsto \xi_{1} + \xi_{2} \in   \g^{*}$$
is a Poisson map. We denote its pull-back by
$$
    \varphi:= \overline{m}^{*}: C^{\infty}(\g^{*}) \to C^{\infty}(\g^{*} \times \h^{*}).
$$
Clearly,
\[\begin{tikzcd}
        {\mathfrak{g}^{*}} && {\mathfrak{g}^{*} \times \mathfrak{h}^{*}} \\
        {\mathfrak{g}^{*} \times \mathfrak{h}^{*}} && {\mathfrak{g}^{*} \times \mathfrak{h}^{*} \times \mathfrak{h}^{*}}
        \arrow["\overline{m}"', from=1-3, to=1-1]
        \arrow["\overline{m}", from=2-1, to=1-1]
        \arrow["{\overline{m} \times \mathrm{id}}"', from=2-3, to=1-3]
        \arrow["{\mathrm{id} \times \overline{m}}", from=2-3, to=2-1]
    \end{tikzcd}\]
is a commutative diagram of Poisson maps and, therefore,
\[\begin{tikzcd}
        {C^{\infty}(\mathfrak{g}^{*})} && {C^{\infty}(\mathfrak{g}^{*} \times \mathfrak{h}^{*})} \\
        {C^{\infty}(\mathfrak{g}^{*} \times \mathfrak{h}^{*})} && {C^{\infty}(\mathfrak{g}^{*} \times \mathfrak{h}^{*} \times \mathfrak{h}^{*})}
        \arrow["\varphi", from=1-1, to=1-3]
        \arrow["\varphi"', from=1-1, to=2-1]
        \arrow["{\varphi \, \widehat{\otimes} \, \mathrm{id}}", from=1-3, to=2-3]
        \arrow["{\mathrm{id} \, \widehat{\otimes} \, \Delta}"', from=2-1, to=2-3]
    \end{tikzcd}\]
is a commutative diagram of morphism of Poisson algebras.
In analogy with the notion of well-behaved topological (co)algebras \cite{Bonneau1994}, we refer to $C^{\infty}(\g^{*})$ as a {\it well-behaved topological $C^{\infty}(\h^{*})$-Poisson comodule algebra}. This terminology extends the standard definition of a Poisson $\mathcal{B}$-comodule algebra $\mathcal{A}$ \cite{Ballesteros2002}, as recalled in Subsection~\ref{subsection:coalgebras}, to the setting where the vector spaces $\mathcal{A}$ and $\mathcal{B}$ are nuclear, Fr\'echet or dual of a Fr\'echet, all structure maps are continuous with respect to the Fr\'echet topology, and the algebraic tensor product $\otimes$ is replaced by the topological tensor product $\widehat{\otimes}$.

Similarly to the $k^{\textup{th}}$-order coproduct, the coaction $\varphi$ also gives rise to a {\it $\ell^{\textup{th}}$-order coaction map} $\varphi^{[\ell]}: C^{\infty}(\g^{*}) \to C^{\infty}\big(\g^{*} \times (\h^{*})^{\ell -1}\big)$
recursively defined by \cite{Ballesteros2002}
\begin{equation}
    \varphi^{[2]} := \varphi, \qquad \varphi^{[\ell]} := \big(\varphi^{[2]} \, \widehat{\otimes} \, \mathrm{id}^{\widehat{\otimes} (\ell -2)}\big) \circ \varphi^{[\ell -1]}, \qquad \ell  >2,
    \label{eq:lth_coaction}
\end{equation}
which is a morphism of Poisson algebras. Of course, this is just the pull-back of the Poisson map $\overline{m}^{[\ell]}: \g^{*} \times (\h^{*})^{\ell-1} \to \g^{*}$ recursively defined by
$$
    \overline{m}^{[2]} := \overline{m}, \qquad \overline{m}^{[\ell]} := \overline{m}^{[\ell -1]} \circ \big(\overline{m}^{[2]} \times \mathrm{id}^{\times (\ell -2)} \big), \qquad \ell > 2.
$$
Combining with the $k^{\textup{th}}$-order coproduct $\Delta^{[k]}$ of $C^{\infty}(\g^{*})$ and the $\ell^{\textup{th}}$-order coaction $\varphi^{[\ell]}$, we obtain the following commutative diagram of morphisms of Poisson algebras
\[\begin{tikzcd}
        {C^{\infty}(\mathfrak{g}^{*})} && {C^{\infty}\big(\mathfrak{g}^{*} \times (\mathfrak{h}^{*})^{\ell -1}\big)} \\
        \\
        {C^{\infty}\big((\mathfrak{g}^{*})^{k}\big)} && {C^{\infty}\big((\mathfrak{g}^{*})^{k} \times (\mathfrak{h}^{*})^{\ell -1}\big)}
        \arrow["{\varphi^{[\ell]}}", from=1-1, to=1-3]
        \arrow["{\Delta^{[k]}}"', from=1-1, to=3-1]
        \arrow["{\Delta^{[k]} \, \widehat{\otimes} \, \mathrm{id}^{\widehat{\otimes}(\ell -1)}}", from=1-3, to=3-3]
        \arrow["{\mathrm{id}^{\widehat{\otimes}(k-1)} \, \widehat{\otimes} \, \varphi^{[\ell]}}"', from=3-1, to=3-3]
    \end{tikzcd}\]


\subsection{On momentum maps and time-independent constants of motion}
This section develops the coalgebra method for Dirac--Lie systems, extending some results from \cite{Carinena2014} and introducing new techniques that will be useful in the construction of mixed superposition rules.

It was shown in \cite[Proposition 3.2]{Carinena2014} that the Poisson algebra $\big(\Adm(M,L), \set{\cdot, \cdot}_{L}\!\big)$ of admissible functions of a Dirac manifold $(M,L)$ is a {\it closed geometrical $\R$-algebra} in the sense of \cite[p.~33]{Nestruev2020}. That is, if $f_{1}, \ldots, f_{r} \in \Adm(M, L)$ and $F \in C^{\infty}(\R^{r})$, then $F(f_{1}, \ldots, f_{r}) \in \Adm(M, L)$.  This is one of the key properties underlying the coalgebra method for Dirac--Lie systems.

An {\it infinitesimal $L$-Hamiltonian action} of a finite-dimensional Lie algebra $\g$ on a Dirac manifold $(M, L)$ is a morphism of Lie algebras $\phi: \g \to \Adm(M,L)$. The {\it momentum map} associated to $\phi$ is the map $\J_{\phi}: M \to \g^{*}$ given by
$$
    \langle \J_{\phi}(x), v \rangle  := [\phi(v)](x), \qquad x\in M, \qquad  v \in \g.
$$
Then,
$$
    \J_{\phi}^{*}: C^{\infty}(\g^{*}) \ni f \mapsto f \circ \J_{\phi} \in  \Adm(M, L)
$$
is a morphism of Poisson algebras  taking  values { in} $\Adm(M, L)$.
Let $\g_{i}$, for $i = 1, \ldots, k$, be finite-dimensional real Lie algebras.  Let $f_{\mathbf{v}}$ denote the linear function on $\g_{1}^{*} \times \cdots \times \g_{k}^{*}$ associated to a $\mathbf{v} := (v_{(1)}, \ldots, v_{(k)}) \in \g_{1} \times \cdots \times \g_{k}$, given by $f_{\mathbf{v}}(\xi_{(1)}, \ldots, \xi_{(k)}) = \sum_{i = 1}^{k} \langle \xi_{(i)},v_{(i)}\rangle $ for every $(\xi_{(1)}, \ldots, \xi_{(k)}) \in \g_{1}^{*} \times \cdots \g_{k}^{*}$.
\begin{prop} \label{prop:momentum}
    Let $(M_{i}, L_{i})$ be Dirac manifolds, and let $\phi_{i}: \g_{i} \to \Adm(M_{i}, L_{i})$ be Lie algebra morphisms for finite-dimensional Lie algebras $\g_{i}$, with $i = 1, \ldots, k$. Then,  $\J_{\phi_{1}} \times \cdots \times \J_{\phi_{k}}: M_{1} \times \cdots \times M_{k} \to \g_{1}^{*} \times \cdots \times \g_{k}^{*}$ induces a Poisson algebra morphism
    $$
        (\J_{\phi_{1}} \times \cdots \times \J_{\phi_{k}})^{*}: C^{\infty}(\g_{1}^{*} \times \cdots \times \g_{k}^{*}) \to \Adm(M_{1} \times \cdots \times M_{k}, L_{1} \times \cdots \times L_{k}).
    $$
\end{prop}
\begin{proof}
    Note that $(\J_{\phi_{1}} \times \cdots \times \J_{\phi_{k}})^{*}$ takes values in $\Adm(M_{1} \times \cdots \times M_{k}, L_{1} \times \cdots \times L_{k})$ since every $\J_{\phi_{i}}^{*}$ takes values in $\Adm(M_{i}, L_{i})$, and $\Adm(M_{1} \times \cdots \times M_{k}, L_{1} \times \cdots \times L_{k})$ is a closed geometrical $\R$-algebra. As $(\J_{\phi_{1}} \times \cdots \times \J_{\phi_{k}})^{*}$ is a pull-back, it is a morphism of $\R$-algebras. Let us now prove that it is a morphism of Lie algebras. Let $\mathbf{v} = (v_{(1)}, \ldots, v_{(k)}), \mathbf{w} = (w_{(1)}, \ldots, w_{(k)}) \in \g_{1} \times \cdots \times \g_{k}$. By definition,
    \begin{multline}
        \set{(\J_{\phi_{1}} \times \cdots \times \J_{\phi_{k}})^{*} f_{\mathbf{v}}, (\J_{\phi_{1}} \times \cdots \times\J_{\phi_{k}})^{*} f_{\mathbf{w}}}_{L_{1} \times \cdots \times L_{k}} \\
        = \set{f_{\mathbf{v}} \circ (\J_{\phi_{1}} \times \cdots \times \J_{\phi_{k}}), f_{\mathbf{w}} \circ (\J_{\phi_{1}} \times \cdots \times  \J_{\phi_{k}})}_{L_{1} \times \cdots \times L_{k}}.
        \label{eq:p1}
    \end{multline}
    From Proposition~\ref{prop:adm_Dirac}, it follows that
    \begin{multline}
        \set{f_{\mathbf{v}} \circ (\J_{\phi_{1}} \times \cdots \times \J_{\phi_{k}}), f_{\mathbf{w}} \circ (\J_{\phi_{1}} \times \cdots \times  \J_{\phi_{k}})}_{L_{1} \times \cdots \times L_{k}}      \\[2pt]
        = \set{\lambda\big(\phi_{1}(v_{(1)}), \ldots, \phi_{k}(v_{(k)})\big), \lambda\big(\phi_{1}(w_{(1)}), \ldots, \phi_{k}(w_{(k)}) \big)}_{L_{1} \times \cdots \times L_{k}}  \\[2pt]
        = \lambda \big(\{\phi_{1}(v_{(1)}), \phi_{1}(w_{(1)})\}_{L_{1}}, \ldots, \{\phi_{k}(v_{(k)}), \phi_{k}(w_{(k)})\}_{L_{k}} \big),
        \label{eq:p2}
    \end{multline}
    and, since each $\phi_{i}: \g_{i} \to \Adm(M_{i}, L_{i})$ is a morphism of Lie algebras and each $\J_{\phi_{i}}^{*}: C^{\infty}(\g_{i}^{*}) \to \Adm(M_{i}, L_{i})$ is a morphism of Poisson algebras, it follows that
    \begin{multline}
        \lambda  \big(\{\phi_{1}(v_{(1)}), \phi_{1}(w_{(1)})\}_{L_{1}}, \ldots, \{\phi_{k}(v_{(k)}), \phi_{k}(w_{(k)})\}_{L_{k}} \big) = \lambda \big(\phi_{1}([v_{(1)}, w_{(1)}]), \ldots, \phi_{k}([v_{(k)}, w_{(k)}]) \big) \\[2pt]
        = \lambda \big(f_{[v_{(1)}, w_{(1)}]} \circ \J_{\phi_{1}}, \ldots, f_{[v_{(k)}, w_{(k)}]} \circ \J_{\phi_{k}} \big) = \lambda\big(\J_{\phi_{1}}^{*}\big(\{f_{v_{(1)}}, f_{w_{(1)}}\}_{\g_{1}^{*}} \big), \ldots, \J_{\phi_{k}}^{*} \big(\{f_{v_{(k)}}, f_{w_{(k)}}\}_{\g_{k}^{*}} \big) \big) \\[2pt]
        = \big(\J_{\phi_{1}} \times \cdots \times \J_{\phi_{k}}\big)^{*}\big(\{f_{\mathbf{v}}, f_{\mathbf{w}}\}_{\g_{1}^{*} \times \cdots \times \g_{k}^{*}}\big).
        \label{eq:p3}
    \end{multline}
    The identities \eqref{eq:p1}--\eqref{eq:p3} show that $\big(\J_{\phi_{1}} \times \cdots \times \J_{\phi_{k}}\big)^{*}$ is a morphism of Lie algebras and, since it is also a morphism of $\R$-algebras, it becomes a morphism of Poisson algebras.
\end{proof}

\begin{cor}
    Let $(M_{i}, L_{i})$ be Dirac manifolds, and let $\phi_{i}: \g \to \Adm(M_{i}, L_{i})$ be infinitesimal $L_{i}$-Hamiltonian actions of a finite-dimensional real Lie algebra $\g$, for $i = 1, \ldots, k$. If $C \in C^{\infty}(\g^{*})$ is a Casimir function, then $\big(\J_{\phi_{1}} \times \cdots \times \J_{\phi_{k}}\big)^{*}\big[\Delta^{[k]}(C) \big]$ commutes with  the direct products $\lambda\big(\phi_{1}(v), \ldots, \phi_{k}(v)\big)$ for all $v \in \g$. Hence, it is a constant of motion of all direct products $X_{\phi_{1}(v)} \times \cdots \times X_{\phi_{k}(v)}$ for all $v \in \g$.
\end{cor}

\begin{proof}
    If $C \in C^{\infty}(\g^{*})$ is a Casimir function then, for every $v \in \g$, we have that
    \begin{multline*}
        \big[X_{\phi_{1}(v)} \times \cdots \times X_{\phi_{k}(v)}\big] \big(\J_{\phi_{1}} \times \cdots \times \J_{\phi_{k}} \big)^{*}\big[\Delta^{[k]}(C)\big] \\[2pt]
        = -\big\{\lambda\big(\phi_{1}(v), \ldots, \phi_{k}(v)\big), \big(\J_{\phi_{1}} \times \cdots \times \J_{\phi_{k}}\big)^{*}\big[\Delta^{[k]}(C)\big]\big\}_{L_{1} \times \cdots \times L_{k}}.
    \end{multline*}
    Since $\lambda\big(\phi_{1}(v), \ldots, \phi_{k}(v)\big) = \big(\J_{\phi_{1}} \times \cdots \times \J_{\phi_{k}}\big)^{*}\big [\Delta^{[k]}(f_{v})\big]$, it follows that
    \begin{multline*}
        \big\{\lambda\big(\phi_{1}(v), \ldots, \phi_{k}(v)\big), \big(\J_{\phi_{1}} \times \cdots \times \J_{\phi_{k}}\big)^{*}\big[\Delta^{[k]}(C)\big]\big\}_{L_{1} \times \cdots \times L_{k}}
        \\[2pt]
        = \big\{\big(\J_{\phi_{1}} \times \cdots \times \J_{\phi_{k}}\big)^{*} \big[\Delta^{[k]}(f_{v})\big], \big(\J_{\phi_{1}} \times \cdots \times \J_{\phi_{k}}\big)^{*}\big[\Delta^{[k]}(C)\big]\big\}_{L_{1} \times \cdots \times L_{k}} \\[2pt]
        = \big(\J_{\phi_{1}} \times \cdots \times \J_{\phi_{k}}\big)^{*}\big \{\Delta^{[k]}(f_{v}), \Delta^{[k]}(C)\big\}_{(\g^{*})^{k}} \\[2pt]
        = \big(\J_{\phi_{1}} \times \cdots \times \J_{\phi_{k}}\big)^{*} \big[\Delta^{[k]}\big(\set{f_{v}, C}_{\mathfrak{g}^{*}}\big)\big] = 0.
    \end{multline*}
\end{proof}

The following diagram summarises the previous construction.
\[\begin{tikzcd}[column sep = scriptsize]
        {C^{\infty}(\mathfrak{g}^{*})} & {C^{\infty}(\mathfrak{g}^{*} \times \mathfrak{g}^{*})} & \cdots & {C^{\infty}\big((\mathfrak{g}^{*})^{k}\big)} \\
        {\mathrm{Adm}(M_{1}, L_{1})} & {\mathrm{Adm}(M_{1} \times M_{2}, L_{1} \times L_{2})} & \cdots & {\mathrm{Adm}(M_{1} \times\cdots \times M_{k}, L_{1} \times \cdots \times L_{k})}
        \arrow["{\Delta^{[2]}}", from=1-1, to=1-2]
        \arrow["{\Delta^{[k]}}", curve={height=-24pt}, from=1-1, to=1-4]
        \arrow["{\mathbf{J}_{\phi_{1}}^{*}}"', from=1-1, to=2-1]
        \arrow[from=1-2, to=1-3]
        \arrow["{(\mathbf{J}_{\phi_{1}} \times \mathbf{J}_{\phi_{2}})^{*}}"', from=1-2, to=2-2]
        \arrow[from=1-3, to=1-4]
        \arrow["{(\mathbf{J}_{\phi_{1}} \times \cdots \times \mathbf{J}_{\phi_{k}})^{*}}"', from=1-4, to=2-4]
        \arrow[hook, from=2-1, to=2-2]
        \arrow["\pr_{1}^{*}"', hook, curve={height=24pt}, from=2-1, to=2-4]
        \arrow[hook, from=2-2, to=2-3]
        \arrow[hook, from=2-3, to=2-4]
    \end{tikzcd}\]

\begin{example}
    As established in Example~\ref{ex:ho}, the harmonic oscillator system \eqref{eq:tdho:system} admits a mixed superposition rule in terms of three particular solutions of the Riccati equation \eqref{eq:tdho:Riccati}.  Let us  derive such a mixed superposition rule using the preceding results. Consider the $t$-dependent vector field $X = X_{3} +\Omega^{2}(t)  X_{1}$ associated with system \eqref{eq:tdho:system} expressed in the coordinates $(x, y)$ on $\R^{2}_{y \neq 0}$ given in \eqref{eq:tdho:diff}, and let $\pi_{*}X$ be the $t$-dependent vector field \eqref{eq:tdho:Riccati_tdep} associated with the Riccati equation \eqref{eq:tdho:Riccati}. We seek two $t$-independent constants of motion $I_{1}, I_{2}$ defined at least locally for the direct product $X_{E}:= X \times (\pi_{*}X)^{[3]}$ on $\mathbb{R}^5$, namely
    \begin{equation}
        X_{E} :=    x^{2}\pdv{x} + xy \pdv{y} + \sum_{i=1}^3  x^2_{(i)} \pdv{x_{(i)}}+\Omega^2(t)\left[\pdv{x}+   \sum_{i=1}^3    \pdv{x_{(i)}}\right] ,
        \label{eq:ho:XE}
    \end{equation}
    such that $\partial(I_{1}, I_{2})/\partial(x, y) \neq 0$.

    In what follows, consider the Lie algebra $\mathfrak{sl}(2, \R)$ and a basis $\set{e_{1}, e_{2}, e_{3}}$ with commutation relations $[e_{1}, e_{2}] = -2 e_{1},\ [e_{1}, e_{3}] = -e_{2},\ [e_{2}, e_{3}] = - 2e_{3}$, and the   Casimir function $C$ on $\mathfrak{sl}(2, \R)^{*}$ given by
    \begin{equation}
        C := e_{1}e_{3}- \frac{1}{4}e_{2}^{2}.
        \label{eq:tdho:Cas}
    \end{equation}
    We recall from \cite{Ballesteros2013,Ballesteros2015} that the { vector fields} $X_{i}$ given in  \eqref{eq:tdho:I5_vf}, and the diagonal prolongations $(\pi_{*}X_{i})^{[2]}$ of the vector fields \eqref{eq:tdho:Riccativf}, are Hamiltonian vector fields relative to the symplectic forms on $\R^{2}_{y \neq 0}$ and $\R^{2}_{x_{(1)}\neq x_{(2)}}$,  respectively, given by
    $$
        \omega_{1} := \frac{1}{y^{3}} \dd x \wedge \dd y, \qquad  \omega_{2} := \frac{1}{ (x_{(1)} - x_{(2)} )^{2}} \dd x_{(1)} \wedge \dd x_{(2)}.
    $$
    The morphisms of Lie algebras $\phi_{1}: \mathfrak{sl}(2, \R) \to C^{\infty}\big(\R^{2}_{y \neq 0}\big)$ and $\phi_{2}: \mathfrak{sl}(2, \R) \to C^{\infty}\big(\R^{2}_{x_{(1)}\neq x_{(2)}}\big)$ defined by
    \begin{equation*}
        \begin{aligned}
             & \phi_{1}(e_{1}):= - \frac{1}{2y^{2}}, \qquad         &  & \phi_{1}(e_{2}):= - \frac{x}{y^{2}}, \qquad                          &  & \phi_{1}(e_{3}):= - \frac{x^{2}}{2y^{2}},                   \\[2pt]
             & \phi_{2}(e_{1}):= \frac{1}{x_{(1)}- x_{(2)}}, \qquad &  & \phi_{2}(e_{2}):= \frac{x_{(1)} + x_{(2)}}{x_{(1)}- x_{(2)}}, \qquad &  & \phi_{2}(e_{3}):= \frac{x_{(1)}x_{(2)}}{x_{(1)} - x_{(2)}},
        \end{aligned}
    \end{equation*}
    yield Hamiltonian functions associated with $X_{i}$ and $(\pi_{*}X_{i})^{[2]} $, respectively, for $i = 1, 2, 3$. Therefore, the Casimir function \eqref{eq:tdho:Cas} leads to a constant of motion of $X \times (\pi_{*}X)^{[2]}$, and hence of $X_{E}$, given by
    $$
        (\J_{\phi_{1}} \times \J_{\phi_{2}})^{*}{ [\Delta^{[2]}(C)]} = - \left[\frac{1}{4} + \frac{(x - x_{(1)})(x -x_{(2)})}{2y^{2}(x_{(1)}- x_{(2)})} \right],
    $$
    from which the constant of motion
    $$I_{1}:=\frac{(x - x_{(1)})(x -x_{(2)})}{y^{2}(x_{(1)}- x_{(2)})}$$
    follows. Note that this constant of motion can also be considered as a constant of motion to $X\times (\pi_*X)^{[3]}$ (\ref{eq:ho:XE}) on an open dense subset of $\mathbb{R}^5$ in a natural manner.

    To derive the required second constant of motion, observe that the vector fields $X_{i} \times \pi_{*}X_{i}$, for $i = 1, 2, 3$, span a VG Lie algebra that is locally automorphic on  $M:= \{(x, y, x_{(1)}) \in \R^{3}: y(x- x_{(1)}) \neq 0\}$. Hence, by employing the standard methods for determining invariant forms for locally automorphic Lie systems \cite{Gracia2019}, one readily finds that
    $$\eta := \frac{1}{y} \dd y - \frac{1}{x - x_{(1)}} \dd x$$
    is a contact form on ${ M}$ turning $X \times (\pi_{*}X) $ into a contact Lie system of Liouville type on that open subset. Consequently, $X \times \pi_{*}X$  is a Dirac--Lie system relative to the Dirac structure $L^{\omega}$ induced by the presymplectic form $\omega := \dd \eta$. In this context, the morphism of Lie algebras $\phi_{3}: \mathfrak{sl}(2, \R) \to \Adm(M, L^{\omega})$ defined by
    $$
        \phi_{3}(e_{1}) := \frac{1}{x - x_{(1)}}, \qquad \phi_{3}(e_{2}):= \frac{ x+ x_{(1)}}{x - x_{(1)}}  , \qquad \phi_{3}(e_{3}):= \frac{xx_{(1)}}{x- x_{(1)}}  ,
    $$
    is such that $\phi_{3}(e_{i})$ is an $L^{\omega}$-Hamiltonian function for $X_{i} \times \pi_{*}X_{i}$, where $i = 1, 2, 3$. Then, a second  constant of motion
    $$
        I_{2}:= (\J_{\phi_{3}} \times \J_{\phi_{2}})^{*} {[\Delta^{[2]}(C)]} = \frac{(x- x_{(3)})(x_{(1)} - x_{(2)})}{(x - x_{(1)})(x_{(2)}- x_{(3)})},
    $$
    arises, where $C$ is the Casimir function \eqref{eq:tdho:Cas}. This constant of motion is defined on an open subset of $M\times \mathbb{R}^2\subset \mathbb{R}^2\times \mathbb{R}^3$ and it becomes a constant of motion of $X_{E}$ (\ref{eq:ho:XE}).

    Since $I_{1}$ and $I_{2}$ satisfy the regularity condition
    $\partial(I_{1}, I_{2})/\partial(x, y) \neq 0$, the system $I_{1} = k_{1}, I_{2} = k_{2}$ for $(x,y)$, gives rise to the general solution  of $X$ in the form
    \begin{equation*}
        \begin{split}
             & x(t) = \frac{ \big[x_{(1)}(t) - x_{(2)}(t) \big]x_{(3)}(t) + k_{2} x_{(1)}(t) \big[x_{(3)}(t) - x_{(2)}(t) \big]}{x_{(1)}(t) - (1 + k_{2}) x_{(2)}(t) + k_{2}x_{(3)}(t)},                                                           \\[2pt]
             & y(t) = \pm \left[ \frac{ (1 + k_{2} )\big(x_{(1)}(t) - x_{(2)}(t) \big)\big(x_{(1)}(t) - x_{(3)}(t)\big)\big(x_{(2)}(t) - x_{(3)}(t)\big)}{k_{1} \big(x_{(1)}(t) - (1 + k_{2})x_{(2)}(t) + k_{2}x_{(3)}(t)\big)^{2}} \right]^{1/2},
        \end{split}
    \end{equation*}
    where $x_{(1)}(t), x_{(2)}(t), x_{(3)}(t)$ are three particular solutions of the Riccati equation \eqref{eq:tdho:Riccati} and $k_{1}, k_{2} \in \R$ are  constants.
\end{example}
Finally, let us show how the preceding results { can} be used in the derivation of mixed superposition rules when combined with the Poisson comodule algebras studied in Subsection~\ref{subsection:well_Poisson}.

Let $(M, L)$ be a Dirac manifold, and let $\phi: \mathfrak{g} \to \Adm(M, L)$ be an $L$-infinitesimal Hamiltonian action of a finite-dimensional real Lie algebra $\g$ that is a semidirect sum $\g =  \h \overrightarrow{\oplus} \p$.   Since $\h \subset \g$ is a Lie subalgebra, $\phi$ restricts to an $L$-Hamiltonian action $\phi \vert_{\h}: \h \to \Adm(M, L)$ of $\h$. For each $s \geq 1$, let us denote the product of the momentum maps as
$$
    \J_{\phi}^{[s]}:= \J_{\phi} \times \overset{s}{\cdots} \times \J_{\phi}, \qquad \J_{\phi \vert_{\h}}^{[s]}:= \J_{\phi\vert_{\h}} \times \overset{s}{\cdots} \times \J_{\phi \vert_{\h}}.$$
According to
Proposition~\ref{prop:momentum}, for any $r, s \geq 1$, the mapping  $\J_{\phi}^{[r]} \times \J_{\phi \vert_{\h}}^{[s]}: M^{r+s} \to (\g^{*})^{r} \times (\h^{*})^{s}$ induces  a morphism of Poisson algebras
$$
    \big(\J_{\phi}^{[r]} \times \J_{\phi \vert_{\h}}^{[s]} \big)^{*}: C^{\infty}\big((\g^{*})^{r} \times (\h^{*})^{s}\big) \to \Adm\big(M^{r+s}, L^{[r+s]}\big).
$$
Recall that, for  $\ell \geq 2$, the $\ell^{\textup{th}}$-coaction map $\varphi^{[\ell]}: C^{\infty}(\g^{*}) \to C^{\infty}\big(\g^{*} \times (\h^{*})^{\ell -1}\big)$, as defined in \eqref{eq:lth_coaction}, is a morphism of Poisson algebras. As a consequence, the map $\big(\Delta^{[k]} \, \widehat{\otimes} \, \mathrm{id}^{\widehat{\otimes}(\ell -1)}\big) \circ \varphi^{[\ell]}: C^{\infty}(\g^{*}) \to C^{\infty}\big((\g^{*})^{k} \times (\h^{*})^{\ell -1}\big)$ is also a morphism of Poisson algebras for $k, \ell \geq 2$. The following result  follows from Proposition~\ref{prop:momentum}.

\begin{cor} \label{cor:coactionl} Let $(M, L)$ be a Dirac manifold, and let $\phi: \g \to \Adm(M, L)$ be an $L$-infinitesimal Hamiltonian action of a Lie algebra $\g$ which is a semidirect sum $\g = \h \overrightarrow{\oplus} \p$. If $C \in C^{\infty}(\g^{*})$ is a Casimir function of $\g^{*}$, then
    \begin{enumerate}[label=(\arabic*), font=\normalfont]
        \item $\big(\J_{\phi} \times \J_{\phi \vert_{\h}}^{[\ell -1]} \big)^{*}\big[\varphi^{[\ell]}(C)\big]$ commutes with all  diagonal prolongations $\phi(v)^{[\ell]}$ and all pull-backs  $\widetilde{\phi(u)}$ of $\phi(u)$ along the projection $M^{\ell} \to M$ , for every $v \in \h$ and $u \in \p$.
        \item $\big(\J_{\phi}^{[k]} \times \J_{\phi \vert_{\h}}^{[\ell -1]}\big)^{*}  \big[\big(\Delta^{[k]} \, \widehat{\otimes} \, \mathrm{id}^{\widehat{\otimes}(\ell -1)}\big) \circ \varphi^{[\ell]}(C)\big]$ commutes with all diagonal prolongations $\phi(v)^{[k+\ell-1]}$ and all pull-backs $\widetilde{\phi(u)^{[k]}}$ of the diagonal prolongations $\phi(u)^{[k]}$ along the projection $M^{k + \ell - 1} \to M^{k}$, for every $v \in \h$ and $u \in \p$.
    \end{enumerate}
\end{cor}


\section{Applications}\label{Section:Applications}
Let us now provide several applications of the methods introduced in previous sections.


\subsection{A time-dependent thermodynamical system}
Let us consider the thermodynamical phase space $\R^{5}$ with global coordinates $(U, S, \nu, T, P)$ \cite{Balian2001}. Among them, $(U, S, \nu)$ are {\it extensive variables} corresponding, in this order, to the internal energy, the entropy, and the volume. In thermodynamics, the volume is typically denoted by $V$ instead of $\nu$. However, as discussed in Section~\ref{section:fundamentals}, we reserve the symbol $V$  for Lie algebras. On the other hand, $(T, P)$ are {\it intensive variables} corresponding to the temperature and the pressure, respectively. The so-called {\it Gibbs one-form} \cite{Baz}, given by
\begin{equation}
    \eta_{\mathrm{Gibbs}}:= \dd U - T \dd S + P \dd \nu
    \label{eq:thermo:Gibbs}
\end{equation}
is a contact form on $\R^{5}$. Consider now the following first-order system of ODEs,
\begin{equation}
    \begin{split}
        \dv{U}{t}   & =- \big(b_{1}(t) - b_{6}(t) \big) U + b_{3}(t) \nu + b_{5}(t)S, \\[2pt]
        \dv{S}{t}   & = b_{6}(t) S,                                                   \\[2pt]
        \dv{\nu}{t} & = b_{2}(t) U + \big(b_{1}(t) + b_{6}(t) \big) \nu + b_{4}(t) S, \\[2pt]
        \dv{T}{t}   & = -b_{1}(t) T + b_{2}(t) TP + b_{4}(t) P + b_{5}(t),            \\[2pt]
        \dv{P}{t}   & = b_{2}(t) P^{2} - 2 b_{1}(t) P - b_{3}(t),
    \end{split}
    \label{eq:thermo:system}
\end{equation}
where $b_{i}(t) \in C^{\infty}(\R)$ are arbitrary $t$-dependent functions, for $i = 1, \ldots, 6$. This system is associated with the $t$-dependent vector field given by
$$
    X := \sum_{i = 1}^{6} b_{i}(t) X_{i},
$$
where the vector fields
\begin{equation}
    \begin{aligned}
         & X_{1} := -U \pdv{U} + \nu \pdv{\nu} - T \pdv{T} - 2P \pdv{P}, \qquad &  & X_{2}:= U \pdv{\nu} + TP \pdv{T} + P^{2} \pdv{P}, \\[2pt]
         & X_{3}:= \nu \pdv{U} - \pdv{P}, \qquad                                &  & X_{4}:= S \pdv{\nu} + P \pdv{T},                  \\[2pt]
         & X_{5}:= S \pdv{U} + \pdv{T}, \qquad                                  &  & X_{6} :=U \pdv{U} + S \pdv{S} + \nu\pdv{\nu}.
    \end{aligned}
    \label{eq:thermo_vf}
\end{equation}
have non-vanishing commutation relations \eqref{eq:mixed:VG_cr},   i.e. the same that determine  the five-dimensional Schr\"odinger Lie algebra $\mathcal{S}(1)$  in Subsection~\ref{subsection:rodinger}.
Accordingly, $X$ is a Lie system, and the above vector fields generate a VG Lie algebra $V$ of $X$ and, therefore, $V$ is isomorphic to $\mathcal{S}(1) \oplus \R$, where $\langle X_{1}, \ldots, X_{5} \rangle \simeq \mathcal{S}(1)$ and $X_{6}$ generates the centre of $V$. It readily follows that $V$ admits the Levi decomposition $V = \langle X_{1}, X_{2}, X_{3} \rangle \overrightarrow{\oplus} \langle X_{4}, X_{5}, X_{6} \rangle \simeq \mathfrak{sl}(2, \R) \overrightarrow{\oplus} \R^{3}$. From now on, let us denote $\g := \mathfrak{sl}(2, \R) \overrightarrow{\oplus} \R^{3}$. System \eqref{eq:thermo:system} is a particular case of the family of systems recently examined in~\cite{Campoamor-Stursberg2025a}, where it was established that $X$ is a contact Lie system with respect to the Gibbs one-form \eqref{eq:thermo:Gibbs}. Nevertheless, a direct computation reveals that the system is not of Liouville type, since the Reeb vector field corresponding to $\eta_{\mathrm{Gibbs}}$, namely $\mathcal{R}_{\mathrm{Gibbs}} := \pdv{U}$, fails to be a Lie symmetry of all the elements of $V$.

System \eqref{eq:thermo:system} exhibits several remarkable features. Most notably, the extensive variables $(U, S, \nu)$ evolve independently of the intensive variables $(T, P)$, a property that aligns  with the formalism of direct products.  More concretely, let us regard $\R^{5}$ as the Cartesian product $\R^{5} = \R^{3} \times \R^{2}$, where $(U, S, \nu)$ define global coordinates on $\R^{3}$, while $(T,P)$ furnish global coordinates on $\R^{2}$. Consider the canonical projections
$$
    \pr_{1}: \R^{5} \ni (U, S, \nu, T, P) \mapsto (U, S, \nu) \in \R^{3}, \qquad \pr_{2}: \R^{5} \ni (U, S, \nu, T, P) \mapsto (T, P) \in \R^{2}.
$$
The decoupling between the extensive and the intensive variables arises from the fact that $X$ can be expressed as the direct product
\begin{equation}
    X = (\pr_{1})_{*}X \times (\pr_{2})_{*}X.
    \label{eq:direct_prod2}
\end{equation}
It then follows that, if   $(U(t), S(t), \nu(t), { T}(t), { P}(t))$ is the general solution of $X$, then $(U(t), S(t), \nu(t))$ and $({ T}(t), { P}(t))$ are the general solutions of $(\pr_{1})_{*}X$ and $(\pr_{2})_{*}X$, respectively. The system corresponding to $(\pr_{1})_{*}X$, given by
\begin{equation*}
    \begin{split}
        \dv{U}{t}   & = -\big( b_{1}(t) - b_{6}(t) \big) U + b_{3}(t) \nu + b_{5}(t)S, \\[2pt]
        \dv{S}{t}   & = b_{6}(t) S,                                                    \\[2pt]
        \dv{\nu}{t} & = b_{2}(t) U + \big(b_{1}(t) + b_{6}(t) \big) \nu + b_{4}(t) S,
    \end{split}
\end{equation*}
is a linear homogeneous system of ODEs on $\R^{3}$ and, as such, admits a linear superposition rule in terms of three constants and three particular solutions. On the other hand, the system corresponding to $(\pr_{2})_{*}X$ reads
\begin{equation}
    \begin{split}
        \dv{T}{t} & = -b_{1}(t) T + b_{2}(t) TP + b_{4}(t) P + b_{5}(t), \\[2pt]
        \dv{P}{t} & = b_{2}(t) P^{2} - 2 b_{1}(t) P - b_{3}(t),
    \end{split}
    \label{eq:thermo:Riccati_proj}
\end{equation}
and is a particular case of a projective Riccati equation on $\R^{2}$ \cite{Anderson1981,Anderson1982,Grundland2017}. As proved in \cite{Anderson1981,Anderson1982}, it admits a superposition rule depending on four particular solutions and two constants. It is also worth noting that this system belongs to the so-called I$_{19}^{r = 1}$ class on $\R^{2}$, which cannot be described by means of any of the geometric structures presented in Section~\ref{Section:Dirac} (see \cite{Lucas2020}).

Let us now derive compatible geometric structures for $X$ turning it into a Dirac--Lie system by applying Proposition~\ref{prop:invariants}, following the same procedure of   Subsection~\ref{ex:Riccati_system}. Clearly, the VG Lie algebra $V$ spanned by \eqref{eq:thermo_vf} is not locally automorphic  on $\R^{5}$. Nevertheless, we note that
$$
    X_{1} \wedge X_{3} \wedge X_{4} \wedge X_{5} \wedge X_{6} = -S^{2}(U + P\nu -ST ) \pdv{U} \wedge \pdv{S} \wedge \pdv{\nu} \wedge \pdv{T} \wedge \pdv{P}
$$
does not vanish on the submanifold $M:= \set{S(U + P\nu-ST) \neq 0} \subset \R^{5}$, and that these vector fields span a Lie subalgebra $V':=\langle X_{1}, X_{3}, X_{4}, X_{5},X_{6} \rangle$ of $V$, which is isomorphic to $\big(\mathfrak{b}_{2} \overrightarrow{\oplus} \R^{2} \big) \oplus \R$,  where $\mathfrak{b}_{2}=\langle X_{1}, X_{3} \rangle$ denotes the Borel subalgebra of $\mathfrak{sl}(2, \R)$. Moreover, observe that
$$
    G:= U + P \nu - ST
$$
corresponds to the so-called {\it Gibbs free-energy}. Let $(U(t), S(t), \nu(t), T(t), P(t))$ be a particular solution of $X$ (\ref{eq:thermo:system}). Then, its Gibbs free energy $G(t) = U(t) + P(t) \nu(t) -S(t) T(t)$ satisfies the ODE
\begin{equation}
    \dv{G}{t} = -\big(b_{1}(t) - b_{6}(t) - b_{2}(t) P\big)G.
    \label{eq:thermo_GibbsF}
\end{equation}
From \eqref{eq:thermo:system} and \eqref{eq:thermo_GibbsF}, it follows that, if the entropy $S(t)$ (resp., the Gibbs free energy $G(t)$) of a particular solution of $X$ is such that $S(t_{0}) = 0$  (resp.~$G(t_{0}) = 0$) at a certain $t_{0} \in \R$, then $S(t)$ (resp.~$G(t)$) vanishes identically. Therefore, in what follows we restrict ourselves to studying system $X$ on the submanifold $M = \set{S(U + P \nu - ST) \neq 0} \subset \R^{5}$, where both the entropy and the Gibbs free energy are non-vanishing.

The foregoing discussion entails that the Lie algebra $V' = \langle X_{1}, X_{3}, X_{4}, X_{5}, X_{6} \rangle$ is locally automorphic on $M$. The Lie algebra of Lie symmetries of $V'$, denoted~$\mathrm{Sym}(V')$,  is generated by
\begin{equation}
    \begin{aligned}
         & Z_{1} := - (U + P \nu - ST) \pdv{U}, \qquad                                       &  & Z_{2}:= - \left( \frac{U + P \nu - ST}{S} \right)^{2} \left(  \frac{\nu}{S} \pdv{T} + \pdv{P} \right), \\[2pt]
         & Z_{3}:= \frac{S^{2}}{U + P \nu - ST} \left( P \pdv{U} - \pdv{\nu} \right), \qquad &  & Z_{4}:= - (U + P \nu -ST) \left( \frac{1}{S} \pdv{T} + \pdv{U} \right),                                \\[2pt]
         & Z_{5}:= -X_{6} = - U \pdv{U} - S \pdv{S} - \nu \pdv{\nu},
    \end{aligned}
    \label{eq:thermo:syn}
\end{equation}
and   likewise $\mathrm{Sym}(V') \simeq \big(\mathfrak{b}_{2} \overrightarrow{\oplus} \R^{2} \big) \oplus \R$ with  non-vanishing commutation rules given by
\begin{equation*}
    [ Z_1,Z_2] =-2 Z_2 ,\qquad [ Z_1,Z_3] =  Z_3 ,\qquad [ Z_1,Z_4] =-  Z_4 ,\qquad [ Z_2,Z_3] = Z_4 ,
\end{equation*}
so that $Z_5$ generates the centre of $\mathrm{Sym}(V')$.
Note that $Z_{3}$ is also a Lie symmetry of $V$. The dual frame to $\set{Z_{1}, \ldots, Z_{5}}$ is given by the one-forms
\begin{equation*}
    \begin{split}
         & \alpha_{1} :=- \frac{1}{U + P \nu - ST} \left(  \dd U -   \frac {U + P \nu}{S} \dd S  +P \dd \nu - S \dd T + \nu \dd P \right),                                                                     \\[2pt]
         & \begin{aligned}
                & \alpha_{2}:= - \left( \frac{S}{U + P \nu - ST} \right)^{2} \dd P, \qquad &  & \alpha_{3}:= \left( \frac{U + P \nu - ST}{S^{2}} \right) \left( \frac{\nu}{S} \dd S - \dd  \nu \right), \\[2pt]
                & \alpha_{4}:= \frac{1}{U + P \nu - ST} (- S \dd T + \nu \dd P), \qquad    &  & \alpha_{5}:= - \frac{1}{S} \dd S,
           \end{aligned}
    \end{split}
\end{equation*}
which satisfy the structure equations
\begin{equation}
    \dd \alpha_{2} = 2 \alpha_{1} \wedge \alpha_{2}, \qquad \dd \alpha_{3} = - \alpha_{1} \wedge \alpha_{3}, \qquad \dd \alpha_{4} = \alpha_{1} \wedge \alpha_{4} - \alpha_{2} \wedge \alpha_{3}, \qquad \dd \alpha_{1} = \dd \alpha_{5} = 0.
    \label{eq:thermo:MC}
\end{equation}

Let $\alpha:= \sum_{i = 1}^{5} c_{i} \alpha_{i}$, with $c_{i} \in \R$, be an invariant form for $V'$. Then, $\alpha$ is invariant for $V$ if and only if
$$
    \cL_{X_{2}} \alpha = c_{3} \frac{U + P \nu - ST}{S^{2}} \left( - \dd U + \frac{U+ P \nu}{S} \dd S - P \dd \nu \right) -(c_{1} - c_{4}) \dd P = 0
    \Longleftrightarrow c_{3} = 0,  c_{4} = c_{1}.
$$
Hence, $\alpha = c_{1} (\alpha_{1} + \alpha_{4}) + c_{2} \alpha_{2} + c_{5} \alpha_{5}$ is an invariant form for $V$ for every $c_{1}, c_{2}, c_{5} \in \R$. Among these, those that define contact structures on $M$ are precisely those for which $c_{1} c_{5} \neq 0$, since the structure equations \eqref{eq:thermo:MC} imply that
$$
    \alpha \wedge \dd \alpha \wedge \dd \alpha = -2 c_{1}^{2}c_{5} \,\alpha_{1} \wedge  \alpha_{2} \wedge \alpha_{3} \wedge\alpha_{4}  \wedge \alpha_{5} \neq 0 \Longleftrightarrow c_{1}c_{5} \neq 0.
$$
In particular, taking $c_{1} = c_{5} = -1$ and $c_{2} = 0$, we obtain the contact form
$$
    \eta := - \alpha_{1} - \alpha_{4} - \alpha_{5} = \frac{1}{U + P \nu - ST} \, \eta_{\mathrm{Gibbs}} = \frac{1}{U + P \nu - ST} (\dd U - T \dd S +  P \dd \nu )
$$
which turns $X$ into a contact Lie system of Liouville type. Therefore, $(M, L^{\omega}, X)$ is a Dirac--Lie system, where $L^{\omega}$ is the Dirac structure on $M$ induced by the exact presymplectic form $\omega = \dd \eta$.  Note also that the contact form $\eta$ lies in the same conformal class of the restriction of $\eta_{\mathrm{Gibbs}}$ to $M$; that is, $\ker \eta = \ker \eta_{\mathrm{Gibbs}} \vert_{M}$ defines the same contact distribution, and that the Reeb vector field associated with $\eta$ is $\cR := X_{6}$. The contact Hamiltonian functions corresponding to the vector fields \eqref{eq:thermo_vf} are given by
\begin{equation*}
    \begin{aligned}
         & h_{1} := \frac{U - P \nu}{U + P \nu - ST}, \qquad &  & h_{2}:= - \frac{UP}{U + P \nu - ST}, \qquad &  & h_{3}:= - \frac{\nu}{U + P \nu - ST}, \\[2pt]
         & h_{4}:= - \frac{SP}{U + P \nu - ST}, \qquad       &  & h_{5}:= - \frac{S}{U + P \nu - ST}, \qquad  &  & h_{6} := - 1,
    \end{aligned}
\end{equation*}
and they generate a Lie algebra isomorphic to $\g = \mathfrak{sl}(2, \R) \overrightarrow{\oplus} \R^{3}$ with respect to the Poisson bracket $\set{\cdot, \cdot}_{\omega}$ induced by the  presymplectic form $\omega = \dd \eta$ on $\Adm(M, L^{\omega})$. Let us now consider $\g$ over a basis $\set{v_{1}, \ldots, v_{5}}$ satisfying the same commutation relations as $\set{h_{1}, \ldots, h_{5}}$. Define the morphism of Lie algebras $\phi: \g \to \Adm(M, L^{\omega})$ determined by $\phi(v_{i}):= h_{i}$, for $i = 1, \ldots, 5$, and let $\J_{\phi}: M \to \g^{*}$ be its associated momentum map. Similarly, let $\phi\vert_{\mathfrak{sl}(2, \R)}: \mathfrak{sl}(2, \R) \to \Adm(M, L^{\omega})$ be the restriction of $\phi$ to the subalgebra $\mathfrak{sl}(2, \R) \simeq \langle v_{1}, v_{2}, v_{3} \rangle$, with associated momentum map $\J_{\phi \vert_{\mathfrak{sl}(2, \R)}}: M \to \mathfrak{sl}(2, \R)^{*}$.

Let us now obtain a mixed superposition rule for the projective Riccati equation \eqref{eq:thermo:Riccati_proj} by means of the results of the preceding sections applied to the Dirac--Lie system $(M, L^{\omega}, X)$. Let now $Y$ be the part of $X$ taking values in  $\langle X_{1}, X_{2}, X_{3} \rangle$. Consider the direct product $X_{E} := X \times Y$ on $M^{2}$, explicitly given by
\begin{align}
    X_E & =b_1(t) \sum_{i = 1}^{2} \left( -U_{(i)} \pdv{U_{(i)}} + \nu_{(i)} \pdv{\nu_{(i)}} - T_{(i)} \pdv{T_{(i)}} - 2P_{(i)} \pdv{P_{(i)}} \right)\nonumber                                                                             \\[2pt]
        & \quad +b_2(t) \sum_{i = 1}^{2} \left(U_{(i)} \pdv{\nu_{(i)}} + T_{(i)}P_{(i)} \pdv{T_{(i)}} + P_{(i)}^{2}      \pdv{P_{(i)}} \right)+b_3(t) \sum_{i = 1}^{2} \left(  \nu_{(i)} \pdv{U_{(i)}} - \pdv{P_{(i)}} \right)   \nonumber \\[2pt]
        & \quad +b_4(t)\left( S_{(1)} \pdv{\nu_{(1)}} + P_{(1)} \pdv{T_{(1)}} \right)+b_5(t)\left(S_{(1)} \pdv{U_{(1)}} + \pdv{T_{(1)}}  \right)
    \nonumber                                                                                                                                                                                                                              \\[2pt]
        & \quad  +b_6(t)\left( U_{(1)} \pdv{U_{(1)}} + S_{(1)} \pdv{S_{(1)}} + \nu_{(1)}\pdv{\nu_{(1)}} \right),
    \label{eq:XEthermo}
\end{align}
and which, by means of \eqref{eq:direct_prod2}, can be expressed as
$$
    X_{E} = X \times Y = (\pr_{1})_{*}X \times (\pr_{2})_{*}X \times (\pr_{1})_{*}Y \times (\pr_{2})_{*}Y.
$$
Let now ${ \cD^{V^{X_{E}}}} \subset {\T M^{2}}$ be the generalised distribution spanned by $V^{X_{E}}$, and consider the projection
$$\pr_{\hat{2}}: \big(U_{(1)}, S_{(1)}, \nu_{(1)}, P_{(1)}, T_{(1)}, U_{(2)}, S_{(2)}, \nu_{(2)}, P_{(2)}, T_{(2)} \big) \mapsto \big(U_{(1)}, S_{(1)}, \nu_{(1)}, U_{(2)}, S_{(2)}, \nu_{(2)}, P_{(2)}, T_{(2)}\big).
$$
Then,  $\T \, \pr_{\hat{2}} \vert_{{ \cD^{V^{X_{E}}}}}: { \cD^{V^{X_{E}}}} \to \T \big(\R^{3} \times M\big)$ is  an injective map on an open and dense subset of $M^{2}$. Hence, from Theorem~\ref{th:mixed_injective} it follows that $(\pr_{2})_{*}X$ admits a mixed superposition rule in terms of one particular solution of $(\pr_{1})_{*}X$, one particular solution of $(\pr_{1})_{*}Y$, and one particular solution of $(\pr_{2})_{*}Y$.

Consider the $(L^{\omega})^{[2]}$-admissible functions on $M^{2}$ given by
\begin{equation*}
    \begin{aligned}
         & h_{1}^{[2]} := \sum_{i = 1}^{2} \frac{U_{(i)} - P_{(i)} \nu_{(i)}}{U_{(i)} + P_{(i)} \nu_{(i)} - S_{(i)}T_{(i)}}, \qquad &  & h_{2}^{[2]}:=  - \sum_{i = 1}^{2} \frac{U_{(i)}P_{(i)}}{U_{(i)} + P_{(i)} \nu_{(i)} - S_{(i)}T_{(i)}}, \\
         & h_{3}^{[2]}:= - \sum_{i = 1}^{2} \frac{\nu_{(i)}}{U_{(i)} + P_{(i)} \nu_{(i)} - S_{(i)}T_{(i)}},  \qquad                 &  & h_{4}^{[1]} := - \frac{S_{(1)}P_{(1)}}{U_{(1)} + P_{(1)} \nu_{(1)} - S_{(1)}T_{(1)}},                  \\
         & h_{5}^{[1]} := - \frac{S_{(1)}}{U_{(1)} + P_{(1)} \nu_{(1)} - S_{(1)}T_{(1)}}, \qquad                                    &  & h_{6}^{[1]} :=  - 1.
    \end{aligned}
\end{equation*}
Let $\varphi$ be the coaction $\varphi: C^{\infty}(\g^{*}) \to C^{\infty}(\g^{*} \times \mathfrak{sl}(2, \R)^{*})$ according to the decomposition $\g = \mathfrak{sl}(2, \R) \overrightarrow{\oplus} \R^{3}$ (see Subsection~\ref{subsection:well_Poisson}), and $C$ be the Casimir function on $\g^{*}$ given by
$$
    C := v_{3} v_{4}^{2} - v_{2} v_{5}^2- v_{1}v_{4} v_{5},
$$
which is the usual Casimir function on $\mathcal{S}(1)^{*}$ \cite{Campoamor2005,Campoamor2020}.
    { Here, it suffices to take $\ell = 2$ for the higher-order coaction map $\varphi^{[\ell]}$ in part (1) of Corollary~\ref{cor:coactionl}, since it is enough to yield the following non-trivial constant of motion of $X_{E}$ (\ref{eq:XEthermo}),}
\begin{equation*}
    \begin{aligned}
        F_{1} := (\J_{\phi} \times \J_{\phi \vert_{\mathfrak{sl}(2, \R}})^{*}[\varphi(C)] & = h_{3}^{[2]} \big(h_{4}^{[1]} \big)^{2} - h_{2}^{[2]}\big(h_{5}^{[1]}\big)^2 - h_{1}^{[2]} h_{4}^{[1]}h_{5}^{[1]}                                                                                     \\[2pt]
                                                                                          & = - \frac{S_{(1)}^{2}\big(P_{(1)}- P_{(2)}\big)\big(U_{(2)} + P_{(1)} \nu_{(2)}\big)}{\big(U_{(1)} + P_{(1)} \nu_{(1)} - S_{(1)}T_{(1)}\big)^2\big(U_{(2)} + P_{(2)} \nu_{(2)} - S_{(2)}T_{(2)}\big)}.
    \end{aligned}
\end{equation*}
Since $Z_{3}$ given in \eqref{eq:thermo:syn} is a symmetry of $X$, its diagonal prolongation
$$
    Z_{3}^{[2]} = \sum_{i = 1}^{2} \frac{S_{(i)}^{2}}{U_{(i)} + P_{(i)}\nu_{(i)} - S_{(i)}T_{(i)}} \left( P_{(i)} \pdv{U_{(i)}} - \pdv{\nu_{(i)}} \right)
$$
is a symmetry of $X_{E}$. This gives rise to another constant of motion of $X_{E}$ in the form
$$
    Z_{3}^{[2]} F_{1} =  \left( \frac{S_{(1)}S_{(2)} \big(P_{(1)}- P_{(2)}\big)}{\big(U_{(1)}+P_{(1)} \nu_{(1)} - S_{(1)}T_{(1)}\big)\big(U_{(2)} + P_{(2)} \nu_{(2)} - S_{(2)}T_{(2)}\big)} \right)^{2}.
$$
From this, it follows that
$$
    I_{2} := \frac{S_{(1)}S_{(2)} \big(P_{(1)}- P_{(2)}\big)}{\big(U_{(1)}+P_{(1)} \nu_{(1)} - S_{(1)}T_{(1)}\big)\big(U_{(2)} + P_{(2)} \nu_{(2)} - S_{(2)}T_{(2)}\big)}
$$
is a constant of motion of $X_{E}$, since $I_{2}^{2} = Z_{3}^{[2]}F_{1}$. { Moreover, we have that $F_{1} = - I_{1} I_{2}$, where
        $$
            I_{1}:=\frac{S_{(1)}\big(U_{(2)} + P_{(1)}\nu_{(2)}\big)}{S_{(2)}\big(U_{(1)} + P_{(1)} \nu_{(1)} - S_{(1)}T_{(1)}\big)}
        $$
        is also a  constant of motion of $X_{E}$}. These constants of motion satisfy the regularity condition
$$
    \pdv{(I_{1}, I_{2})}{\big(P_{(1)}, T_{(1)}\big)} \neq 0
$$
and, therefore, the system $I_{1} = k_{1}, I_{2} = k_{2}$ for $\big(T_{(1)},P_{(1)}\big)$ gives rise to the general solution of the projective Riccati equation \eqref{eq:thermo:Riccati_proj}. This expression can be   simplified by taking into account that $S_{(2)}$ is a constant of motion of $X_{E}$ and introducing thus a third constant $k_{3} = S_{(2)}$, namely
\begin{equation*}
    \begin{split}
        T_{(1)}(t) & =   \frac{1}{S_{(1)}(t)\Theta(t)} \Big[  k_{1}k_{3}^2 \big[U_{(1)}(t) + P_{(2)}(t) \nu_{(1)}(t)\big] - k_{3}S_{(1)}(t)\big[U_{(2)}(t) + P_{(2)}(t) \nu_{(2)}(t)\big]   \\
                                                           & \qquad \qquad \qquad \quad  + k_{2} \big[U_{(2)}(t) + P_{(2)}(t) \nu_{(2)}(t) - k_{3}   T_{(2)}(t) \big] \big[U_{(2)}(t) \nu_{(1)}(t) - U_{(1)}(t) \nu_{(2)}(t) \big]   \Big] , \\[2pt]
        P_{(1)}(t) & = \frac{1}{\Theta(t)} \Big[ k_{1}k_{3}^{2}P_{(2)}(t)  + k_{2}U_{(2)}(t) \big[U_{(2)}(t) +P_{(2)}(t)\nu_{(2)}(t) - k_{3} T_{(2)}(t)\big]\Big],
    \end{split}
\end{equation*}
where
$$
    \Theta(t):= k_{1}k_{3}^{2} - k_{2} \nu_{(2)}(t) \big[U_{(2)}(t) + P_{(2)}(t) \nu_{(2)}(t) - k_{3} T_{(2)}(t)\big],
$$
being $\big(U_{(1)}(t), S_{(1)}(t), \nu_{(1)}(t)\big)$    a particular solution of $(\pr_{1})_{*}X$, $\big(U_{(2)}, S_{(2)}(t) = k_{3}, \nu_{(2)}(t)\big)$   a particular solution of $(\pr_{1})_{*}Y$, $\big(T_{(2)}(t), P_{(2)}(t)\big)$   a particular solution of $(\pr_{2})_{*}Y$, and $k_{1}, k_{2}, k_{3} \in \R$   real constants.


\subsection{Extension to PDE Lie systems}
Consider a system of PDEs of the form
\begin{equation}
    \frac{\partial u}{\partial t^i}=F_i(t,u),\qquad u\in \mathbb{R}^n,\qquad i=1,\ldots,s,\qquad t\in \mathbb{R}^s,
    \label{eq:pdes}
\end{equation}
and a particular known solution $u_0(t)$. { Let us formally derive, by means of a power-series ansatz in $\epsilon$, a system whose solutions $v(t), w(t)$ yield a curve $u_{0}(t)+\epsilon v(t)+\epsilon^{2}w(t)$ remaining close to $u_{0}(t)$ as $\epsilon \to 0$, where $v(t)$ and $w(t)$ are assumed to be small enough}. To do so, consider $u(t)=u_0(t)+\epsilon v(t) + \epsilon^{2} w(t)$ for $0\leq \epsilon \ll 1$ and substitute this into system~\eqref{eq:pdes}. Then,
$$
    \frac{\partial u_0}{\partial t^i}+\epsilon \frac{\partial v}{\partial t^i}+\epsilon^2\frac{\partial w}{\partial t^i}= F_{i}(t, u_{0}) + \epsilon \sum_{j = 1}^{n} \pdv{F_{i}}{u^{j}}(t,u_0(t)) v^{j}+ \epsilon^{2} \left( \sum_{j = 1}^{n} \pdv{F_{i}}{u^{j}}(t,u_0(t))w^{j}  + \frac{1}{2} \sum_{j, k = 1}^{n} \pdv{F_{i}}{u^{j}}{u^{k}}(t,u_0(t)) v^{j}v^{k} \right) + O(\epsilon^{3}),
$$with $i = 1, \ldots, s.
$
Considering $v(t), w(t)$ uniformly bounded, such that $\norm{u(t, \cdot ) - u_{0}(t, \cdot )}_{L^{\infty}(\R^{n})} \to 0$  as $\epsilon \to 0$ for every $t \in \R$, we are led to the following system of PDEs
\begin{equation*}
    \begin{split}
         & \pdv{v}{t^{i}} = \sum_{j = 1}^{n} \pdv{F_{i}}{u^{j}}(t,u_0(t))v^{j},                                                                               \\
         & \pdv{w}{t^{i}} = \sum_{j = 1}^{n} \pdv{F_{i}}{u^{j}}(t,u_0(t))w^{j} + \frac{1}{2} \sum_{j,k= 1}^{n} \pdv{F_{i}}{u^{j}}{u^{k}}(t,u_0(t))v^{j}v^{k},
    \end{split} \qquad i = 1, \ldots, s.
\end{equation*}
This is a PDE Lie system \cite{Carinena2007,Odzijewicz2000} related to the VG Lie algebra spanned by
$$
    v^{j}\pdv{v^{k}}, \qquad w^{j} \pdv{w^{k}}, \qquad v^{j}v^{k} \pdv{w^{\ell}}, \qquad j, k, \ell = 1, \ldots, n,
$$
provided that the system of PDEs is integrable. This VG Lie algebra is isomorphic to the semidirect sum  $\big(\mathfrak{gl}(n,\mathbb{R})\oplus \mathfrak{gl}(n,\mathbb{R}) \big)\overrightarrow{\oplus }\mathbb{R}^{n^2(n+1)/2}$, where $\mathfrak{gl}(n,\mathbb{R})$ is the matrix Lie algebra of $n\times n$ matrices with real entries.

As an illustrative example, let us consider the first-order system of PDEs on $\R^{2}$ given by
\begin{equation}\label{eq:Initial}
    \begin{split}
         & \pdv{u_{1}}{x} = \left( \pdv{\psi}{x} u_{1} + a u_{2}\right)\frac{1}{\psi} - u_{1}^{2},             \\[2pt]
         & \pdv{u_{2}}{x} = \psi - u_{1}u_{2},                                                                 \\[2pt]
         & \pdv{u_{1}}{y} = \psi - u_{1}u_{2},                                                                 \\[2pt]
         & \pdv{u_{2}}{y} = \left( \pdv{\psi}{y} u_{2} + \frac{1}{a} u_{1} \right) \frac{1}{\psi} - u_{2}^{2},
    \end{split}
\end{equation}
where $\psi(x, y) \in C^{\infty}(\R^{2})$ is non-vanishing and $a \neq 0$ is a constant. This system is integrable if and only if $\psi$ satisfies the {\it hyperbolic Tzitz\'eica equation} \cite{Tzitzeica1908,Tzitzeica1909,Rogers2002,Hu2020}
$$
    \pdv{\ln \psi}{x}{y} = \psi - \frac{1}{\psi^{2}}.
$$
This is not a PDE Lie system as it is not written in normal form. Suppose now that $u_{0}(t)$ is a known particular solution for \eqref{eq:Initial}, and let us find solutions close to $u_{0}(t)$ by means of the previous ansatz.  In this way, we arrive at the system of PDEs
\begin{equation*}
    \begin{gathered}
        \begin{pmatrix}
            \displaystyle \pdv{v_{1}}{x} \\[2ex]
            \displaystyle \pdv{v_{2}}{x}
        \end{pmatrix}
        = A(u,\psi,\psi_x)
        \begin{pmatrix}
            v_{1} \\[2pt] v_{2}
        \end{pmatrix},
        \qquad
        \begin{pmatrix}
            \displaystyle \pdv{v_{1}}{y} \\[2ex]
            \displaystyle \pdv{v_{2}}{y}
        \end{pmatrix}
        = B(u,\psi,\psi_y)
        \begin{pmatrix}
            v_{1} \\[2pt] v_{2}
        \end{pmatrix},
        \\[8pt]
        \begin{pmatrix}
            \displaystyle \pdv{w_{1}}{x} \\[2ex]
            \displaystyle \pdv{w_{2}}{x}
        \end{pmatrix}
        = A(u,\psi,\psi_x)
        \begin{pmatrix}
            w_{1} \\[2pt] w_{2}
        \end{pmatrix}
        +
        \begin{pmatrix}
            -\,v_{1}^{2} \\[2pt] -\,v_{1}v_{2}
        \end{pmatrix},
        \qquad
        \begin{pmatrix}
            \displaystyle \pdv{w_{1}}{y} \\[2ex]
            \displaystyle \pdv{w_{2}}{y}
        \end{pmatrix}
        = B(u,\psi,\psi_y)
        \begin{pmatrix}
            w_{1} \\[2pt] w_{2}
        \end{pmatrix}
        +
        \begin{pmatrix}
            -\,v_{1}v_{2} \\[2pt] -\,v_{2}^{2}
        \end{pmatrix},
    \end{gathered}
\end{equation*}
where $\psi=\psi(x,y)$ is a solution of the Tzitz\'eica equation, $u_0(t)=(u_1(t),u_2(t))$ and
\[
    A(u,\psi,\psi_x=\partial \psi/\partial x)=\begin{pmatrix} \displaystyle -2 u_{1} + \frac{1}{\psi} \pdv{\psi}{x} & \displaystyle \frac{a}{\psi} \\
                    \displaystyle - u_{2}                                 & \displaystyle - u_{1}
    \end{pmatrix} ,\qquad
    B(u,\psi,\psi_y=\partial \psi/\partial y)=\begin{pmatrix}
        \displaystyle - u_{2}          & \displaystyle - u_{1}                                  \\
        \displaystyle \frac{1}{a \psi} & \displaystyle - 2 u_{2} + \frac{1}{\psi} \pdv{\psi}{y}
    \end{pmatrix}
\]
\begin{table}[t!]
    \footnotesize
    \caption{\small Commutation rules of the vector fields \eqref{eq:pde:vfapp}. They span a  non-solvable VG Lie algebra.}\label{Tab:commutators}
    $
        \!\!\!\!\begin{array}{c|cccccccccccccc}
            \multispan{15}{ }                                                                                                                                  \\[6pt]
            [\cdot, \cdot] & X_1 & X_2  & X_3           & X_4  & X_5 & X_6  & X_7   & X_8  & X_9   & X_{10}           & X_{11} & X_{12}   & X_{13}  & X_{14}   \\[4pt]
            \hline                                                                                                                                             \\[-8pt]
            X_1            & 0   & -X_2 & X_3           & 0    & 0   & 0    & 2 X_7 & X_8  & 0     & 0                & 0      & 2 X_{12} & X_{13}  & 0        \\[2pt]
            X_2            &     & 0    & \! -X_1 + X_4 & -X_2 & 0   & 0    & 2 X_8 & X_9  & 0     & 0                & 0      & 2 X_{13} & X_{14}  & 0        \\[2pt]
            X_3            &     &      & 0             & X_3  & 0   & 0    & 0     & X_7  & 2 X_8 & 0                & 0      & 0        & X_{12}  & 2 X_{13} \\[2pt]
            X_4            &     &      &               & 0    & 0   & 0    & 0     & X_8  & 2 X_9 & 0                & 0      & 0        & X_{13}  & 2 X_{14} \\[2pt]
            X_5            &     &      &               &      & 0   & -X_6 & -X_7  & -X_8 & -X_9  & X_{10}           & 0      & 0        & 0       & 0        \\[2pt]
            X_6            &     &      &               &      &     & 0    & 0     & 0    & 0     & \! -X_5 + X_{11} & -X_6   & -X_7     & -X_8    & -X_9     \\[2pt]
            X_7            &     &      &               &      &     &      & 0     & 0    & 0     & X_{12}           & 0      & 0        & 0       & 0        \\[2pt]
            X_8            &     &      &               &      &     &      &       & 0    & 0     & X_{13}           & 0      & 0        & 0       & 0        \\[2pt]
            X_9            &     &      &               &      &     &      &       &      & 0     & X_{14}           & 0      & 0        & 0       & 0        \\[2pt]
            X_{10}         &     &      &               &      &     &      &       &      &       & 0                & X_{10} & 0        & 0       & 0        \\[2pt]
            X_{11}         &     &      &               &      &     &      &       &      &       &                  & 0      & -X_{12}  & -X_{13} & -X_{14}  \\[2pt]
            X_{12}         &     &      &               &      &     &      &       &      &       &                  &        & 0        & 0       & 0        \\[2pt]
            X_{13}         &     &      &               &      &     &      &       &      &       &                  &        &          & 0       & 0        \\[2pt]
            X_{14}         &     &      &               &      &     &      &       &      &       &                  &        &          &         & 0
        \end{array}
    $
\end{table}

This is a PDE Lie system related to the VG Lie algebra spanned by
\begin{equation}
    \label{eq:pde:vfapp}
    \begin{aligned}
         & X_{1} := v_{1} \pdv{v_{1}}, \qquad      &  & X_{2}:= v_{2} \pdv{v_{1}}, \qquad  &  & X_{3}:= v_{1} \pdv{v_{2}}, \qquad     &  & X_{4}:= v_{2} \pdv{v_{2}},      \\[2pt]
         & X_{5}:= w_{1}\pdv{w_{1}}, \qquad        &  & X_{6}:= w_{2} \pdv{w_{1}}, \qquad  &  & X_{7}:= v_{1}^{2} \pdv{w_{1}}, \qquad &  & X_{8}:= v_{1}v_{2} \pdv{w_{1}}, \\[2pt]
         & X_{9} := v_{2}^{2} \pdv{w_{1}}, \qquad  &  & X_{10}:= w_{1} \pdv{w_{2}}, \qquad &  & X_{11}:= w_{2} \pdv{w_{2}}, \qquad    &  & X_{12}:= v_{1}^{2} \pdv{w_{2}}, \\[2pt]
         & X_{13}:= v_{1}v_{2} \pdv{w_{2}}, \qquad &  & X_{14}:= v_{2}^{2} \pdv{w_{2}}.
    \end{aligned}
\end{equation}
{ The commutation relations between these vector fields are detailed in Table \ref{Tab:commutators}.}
These vector fields can be easily seen to generate a non-solvable Lie algebra $\mathfrak{s} \overrightarrow{\oplus}_{\Gamma}\mathfrak{r}$ with Levi subalgebra $\mathfrak{s}$ isomorphic to $\mathfrak{sl}(2,\mathbb{R})\oplus\mathfrak{sl}(2,\mathbb{R})$,  generated by $\{X_1-X_4,X_2,X_3\}$ and $\{ X_5-X_{11},X_6,X_{10}\}$,   and solvable radical $\mathfrak{r}$ generated by $\{X_1+X_4,X_5+X_{11}, X_7,X_8,X_9,X_{12},X_{13},X_{14}\}$. {The action of $\mathfrak{s}$ over the radical $\mathfrak{r}$ is given by the characteristic representation $\Gamma=R_0^2\oplus\left(R_{\frac{1}{2}}\otimes R_1\right) $, where $R_k$ denotes the $(2k+1)$-dimensional irreducible representation of $\mathfrak{sl}(2,\mathbb{R})$ with highest weight $\lambda=2k$. The solvable algebra $\mathfrak{r}$ further has an Abelian nilpotent radical and a two-dimensional external torus of derivations generated by $X_1+X_4$ and $X_5+X_{11}$.}
As the VG Lie algebra is a semidirect sum, { we can potentially apply} Proposition~\ref{prop:mixed} as mixed superposition rules for PDE Lie systems are obtained exactly as in the case of Lie systems in terms of an associated VG Lie algebra (cf.~\cite{Carinena2011}). Moreover,  the vector fields of the VG Lie algebra leave invariant the regular and integrable rank-two distribution $\cD \subset \T \R^{4}$ spanned by $\partial_{w_{1}}$ and $\partial_{w_{2}}$, { implying that} Theorem~\ref{thm:imprimitive} can be applied as well.


\section{Concluding remarks}\label{Section:Concluding}
In this work, { a formal and comprehensive framework to derive mixed superposition rules for Lie systems admitting some specific geometric and algebraic features has been developed. In this context, the Poisson coalgebra method, previously restricted to the case of  superposition rules for LH systems, has been successfully extended to deal with mixed superposition rules, as well as Dirac--Lie systems. The formalism has also shown its usefulness in solving and stating correctly a technical problem about the Poisson coalgebra method and the isomorphism between Poisson algebras on manifolds and the Poisson algebra naturally related to tensor products, that had not been explained satisfactorily in the existing literature. Several physically relevant examples,   including Riccati systems, Schr\"odinger Lie systems, time-dependent oscillators, and  a time-dependent Calogero--Moser system, have been presented, illustrating the  effectiveness of our approach and the relevant connections to both symplectic and contact geometries. In this context, it is worthy to be mentioned that the application to time-dependent thermodynamical models shows the wide range of systems of differential equations potentially analyzable by the method. Of special interest are first-order systems of Mayer type \cite{Flanders1963}, as well as differential equations  associated with systems related to geometric thermodynamics \cite{Baz}, and, more generally, to the context of Pfaff systems \cite{Flanders1963,CAR},  where the techniques explored in this work could lead to new criteria for either separating effectively equivalence classes  or establish new hierarchies  of Pfaff systems in arbitrary dimensions. In particular, it is of interest  to study under which conditions these systems constitute a contact Lie system for an appropriate Pfaff form, and whether they can be further described as Dirac--Lie systems.

        The case of PDE Lie systems also constitutes a rich field where the use of mixed superposition rules can be of relevance for  advancing in the study of structural properties of differential equations and integrable systems. As shown in the examples, the techniques can be { suitably} adapted to the structural analysis of PDE Lie systems, showing, in addition, that  Lie algebras admitting a { non-trivial} Levi decomposition appear naturally as VG algebras in this context. The procedure provides a natural ansatz for the derivation of approximate solutions of systems of { PDEs}. }

    { There are several possible continuations of this work. The notion of contractions of realizations of Lie algebras of vector fields has already been used in the context of Lie systems and { generalisations} (see  e.g.~\cite{Campoamor-Stursberg2025a} and references therein), and the question whether this approach can be extended to the case of PDE Lie systems with non-solvable VG algebras as a contraction of (nonlinear) realizations of semisimple algebras is still unanswered. On the other hand, the  formalism of quantum deformations has been successfully formulated for LH systems \cite{Ballesteros2021}, and as natural continuation of this work, it is conceivable to generalize this formalism to the case of the mixed Hamiltonian coalgebra realm, which special focus on the derivation and characterization of approximate solutions of the $n^{th}$-order perturbations associated to the truncation of the analytical series determined by the   quantum  deformation parameters.   These problems, among other formal considerations, are currently under study and will be the topic of future works.
    }


\section*{Acknowledgements}

R.C.-S., O.C.~and  F.J.H.~have been partially supported by Agencia Estatal de Investigaci\'on (Spain) under  the grant PID2023-148373NB-I00 funded by MCIN/AEI/10.13039/501100011033/FEDER, EU.
O.C.~acknowledges financial support from the IDUB I.1.5 program of the University of Warsaw with project number PSP: 501-D111-20-0001150, as well as a fellowship (grant C15/23) funded by the Universidad Complutense de Madrid and Banco  Santander. F.J.H.~acknowledges support  by the  Q-CAYLE Project  funded by the Regional Government of Castilla y Le\'on (Junta de Castilla y Le\'on, Spain) and by the Spanish Ministry of Science and Innovation (MCIN) through the European Union funds NextGenerationEU (PRTR C17.I1).    The authors also acknowledge the contribution of RED2022-134301-T funded by  MCIN/AEI/10.13039/501100011033 (Spain).

\section*{Data availability}

This research did not generate or analyse any datasets. All results are derived from theoretical and analytical work presented within the article.

\section*{Funding and competing interests}

The authors have no relevant financial or non-financial interests to disclose.

\end{document}